\documentclass[twoside,11pt]{article}
\usepackage[preprint]{jmlr2e}
\usepackage{etoolbox}
\makeatletter
\patchcmd{\@maketitle}{\if@preprint}{\if@preprint
  \noindent{\small\vphantom{Ay}}\vskip\aftermaketitskip
}{}{\PackageError{manuscript}{Title spacing patch failed}{}}
\makeatother
\usepackage{smile}
\usepackage{amsthm,amsmath,amssymb}
\usepackage{natbib}
\usepackage{graphicx}
\usepackage{epsfig}
\usepackage{multirow}
\usepackage{subfigure}
\usepackage{makecell}
\usepackage{booktabs}
\usepackage{array}
\usepackage{url}
\usepackage{xr}
\usepackage{changes}
\usepackage{algorithm}
\usepackage{algorithmic}
\usepackage{bm}
\usepackage{smile}
\usepackage{mathtools}
\usepackage{wrapfig}
\usepackage{lipsum}
\usepackage{mathrsfs}
\usepackage{dsfont}
\usepackage{rotating}
\usepackage{float}
\usepackage{natbib}
\usepackage{multirow}
\usepackage{apalike}
\usepackage{xcolor}
\usepackage{romannum}
\usepackage{hyperref}
\usepackage{pifont}
\usepackage{boldline}
\usepackage{lastpage}
\usepackage[utf8]{inputenc}
\usepackage{diagbox}
\usepackage{threeparttable} %

\ShortHeadings{Tucker Personalized Subspace PCA}{Kangxiang Qin, et al.}
\firstpageno{1}

\def\reals{{\mathbb R}}
\def\P{{\mathbb P}}

\def\E{{\mathbb E}}
\def\argmin{\mathop{\text{\rm arg\,min}}}

\def\op{\text{op}}

\def\argmax{\mathop{\text{\rm arg\,max}}}

\def\span{{\mathop{\text{\rm span}}}}
\def\tr{{\mathop{\text{\rm tr}}}}
\def\vec{{\mathop{\text{\rm vec}}}}

\begin{document}
\pagenumbering{arabic}
\addtocontents{toc}{\protect\setcounter{tocdepth}{-10}}

\title{Optimal Personalized Subspace Learning for Multi-view Tensor Observations}

\author{\name Kangxiang Qin$^*$ \email kangxiangqin@mail.sdu.edu.cn \\
       \addr Institute for Financial Studies\\
       Shandong University
       \AND
       \name Zeyu Li$^*$ \email zeyuli@seu.edu.cn \\
       \addr School of Statistics and Data Science
       \\
       Southeast University
       \AND
       \name
       Xinbing Kong$^\dagger$ \email
       xinbingkong@126.com
       \\
       \addr School of Statistics and Data Science
       \\
       Southeast University
          \AND
       \name
       Wang Zhou
       \email
       wangzhou@nus.edu.sg
       \\
       \addr Department of Statistics and Applied Probability
       \\
       National  University of Singapore
       }

\maketitle
\thispagestyle{plain}
\footnotetext[1]{* Equal contribution.}
\footnotetext[2]{$\dagger$ Corresponding author.}

\begin{abstract}
In this work, we model the observed multi-view tensors by decomposing the underlying signal in each view into two components: (i) the shared component that captures common dynamics across all views, and (ii) the private component that accounts for view-wise unique variations. To decouple the shared and private components, we introduce a novel Tucker personalized subspace principal component analysis (TPS-PCA) approach for tensors, which admits a one-step closed-form solution and serves as an ideal surrogate for our extended tensor-version personalized PCA (TP-PCA), adapted from the seminal work by \cite{shi2024personalized}. The theoretical analysis reveals that the proposed TPS-PCA estimators reach the minimax lower bound in terms of view-wise tensor decoupling, whereas the TP-PCA estimators only achieve a rate of average decoupling error across views, which is still slower than that of the TPS-PCA estimators. Extensive numerical experiments are conducted on synthetic and real datasets, demonstrating the wide applicability of the proposed method in fields including power management, financial analysis, and activity recognition.
\end{abstract}


\section{Introduction}
Due to the substantial increase in computing capacity, tensors are now common subjects in statistics, applied mathematics, and engineering \citep{kolda2009tensor}, thanks to their ability to naturally represent multi-way data while preserving the underlying structural information across different modes. Tensor-oriented statistical
methods, such as tensor SVD \citep{zhang2014novel,zhang2018tensor}, tensor completion \citep{xia2019polynomial,xia2021statistically,xia2021statistical}, tensor PCA and factor analysis 
\citep{han2022tensor,han2024tensor,barigozzi2026statistical,chen2026estimation}, and tensor inference \citep{xu2025statistical,agterberg2026statistical,wen2026online}, have experienced rapid development over the last few years. These representative methods are useful and easy to implement. Therefore, they are widely applied in recommendation systems \citep{bi2018multilayer,zhang2021dynamic}, multi-imaging data analysis \citep{zhou2013tensor,tang2020individualized}, econometrics \citep{li2015tensor,wang2024high}, and network structure analysis \citep{ke2019community,jing2021community}. 

However, in some data-intensive fields such as computer science and finance \citep{li2015multi,zhou2023fast}, a single tensor is often insufficient to fully characterize the statistical objects in hand. It is increasingly common to observe  heterogeneous data that are naturally structured as multiple tensors \citep{DBLP:journals/corr/abs-1911-04226,gahrooei2021multiple}. These tensors tend to align along a certain mode of objects, such as time or experimental subjects, while each tensor captures the corresponding characteristics of the objects from a distinct view \citep{wu2019essential,li2023orthogonal}. Indeed, multi-view tensors are intrinsically connected to each other, and the cross-view dependence is often driven by a latent shared factor structure  \citep{khan2014bayesian,cheng2018tensor}. Meanwhile, as each view provides distinct characteristics of the objects, the corresponding tensor typically contains view-specific information \citep{shi2025tensor,wang2026embedded}, resulting in a private structure. As such, multi-view tensor data typically exhibit both shared and view-specific structures.

To illustrate this point, we present two concrete examples. In activity recognition tasks, signals from sensors attached to different body parts are collected simultaneously while a given subject performs the same activity (see the left panel of Figure \ref{fig:intro_examples_both}) \citep{barshan2014recognizing,wang2018stratified}. The resulting multi-view tensors naturally admit both shared and view-specific structures \citep{yu2019transfer,chen2020fedhealth,ray2023transfer}.
Similarly, in the financial system, stock returns and financial characteristics \citep{xu2021tensor,fan2022structural,mo2025act} can be collected at aligned timestamps (see the right panel of Figure \ref{fig:intro_examples_both}), and these two views also tend to contain shared and view-specific dynamics \citep{alti2019dynamic,qu2022identification}.
To analyze datasets similar to the mentioned examples, it is natural to consider the following question: \textit{how can we extract meaningful shared and view-specific private features given a set of multi-view tensor observations?}

\begin{figure}[h]
    \centering
    \begin{minipage}{\linewidth}
        \centering
        \includegraphics[width=\linewidth]{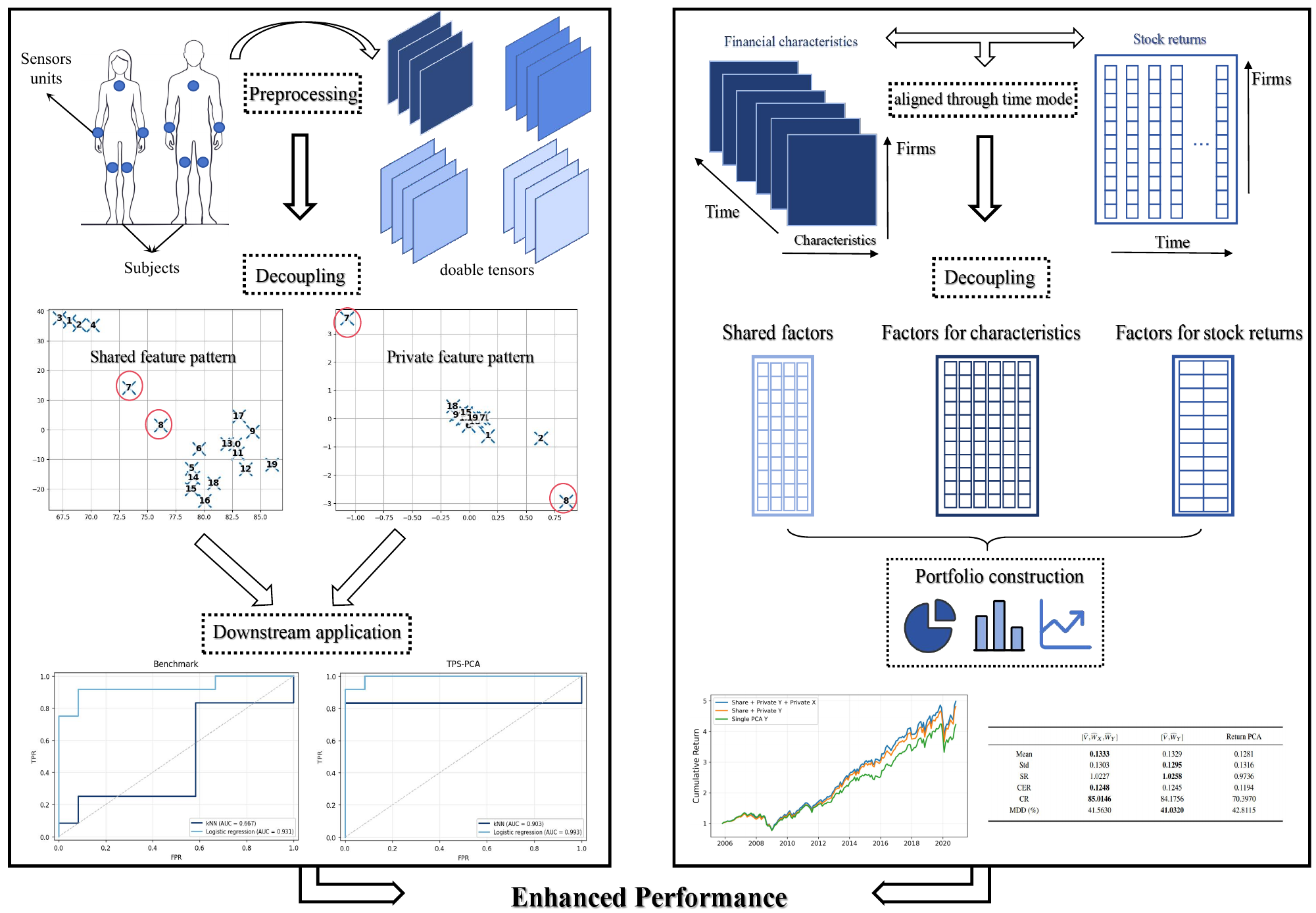} 
    \end{minipage}
    \caption{
    Two multi-view tensor data examples. $(1)$ The left panel provides a multi-view activity recognition dataset. The sensors are attached to different body parts of a participant. The collected signals from each body part naturally form multiple views to describe the same activity of a given subject. $(2)$ The right panel presents a two-view financial example. Both views are collected from the Center for Research in Security Prices (CRSP) and are respectively formulated into a matrix (stock returns) and a tensor (financial characteristics). In both examples, we identify the shared and view-specific private structures, so as to enhance the performance of downstream applications such as activity classification (left panel) and portfolio construction (right panel).  More details can be found in Section \ref{sec:4}.
    \label{fig:intro_examples_both}}
\end{figure}

In essence, the idea of identifying the commonality and specificity among multiple datasets is ubiquitous in statistical learning, which can be found in fields like supervised transfer learning \citep{he2022transfer,li2023transfer,tian2022transfer,li2024simultaneous,he2024representational}, unsupervised transfer learning \citep{li2024knowledge,he2025transpca,hu2025aggregated}, federated learning \citep{li2024federated,cai2024optimal,shi2024personalized}, etc. Among these studies, \cite{shi2024personalized} is one of the first works to consider cross-dataset commonality and dataset-specific heterogeneity in an unsupervised learning setting. They proposed a novel optimization problem named personalized PCA that can identify both common and source-specific components for vector-type datasets. Meanwhile, \cite{hu2025aggregated} used an aggregated projection method to identify the common and private factors across multiple economic panels, whereas \cite{sergazinov2026spectral} applied canonical-correlation-type analysis to characterize shared and private components in biological datasets.

Despite the growing literature, several intrinsic gaps still remain. First, most existing methods for identifying shared and private components are tailored for vector-shaped observations. For tensor observations, however, the direct vectorization not only increases computational costs but also breaks down the latent tensor structure \citep{feng2023tensorir,wang2024low}, thereby reducing statistical efficiency for common component estimation and downstream analysis.  We refer to Table \ref{table:subspace_comparison} in Section \ref{sec:simulation} for a numerical comparison between the inherent tensor-based and vectorization-based algorithms. 

Second, while statistical learning for multiple tensors has been recently considered in studies such as \cite{chen2025distributed,hu2025personalized}, the existing methods typically require all tensors to have identical shapes. In contrast, multi-view tensors often exhibit heterogeneous structures (e.g., different mode dimensions and even different numbers of modes) across different views. For example, in the financial example presented in the right panel of Figure \ref{fig:intro_examples_both}, the firm stock returns at each time point form a vector (order-one), while the firm characteristic matrix is of order-two. As far as we know, the existing multi-source tensor methods are not directly applicable to multi-view tensors with different shapes \citep{chen2025distributed,hu2025personalized}. 

Third, the computational efficiency and statistical optimality of personalized decomposition methods are still not fully understood. 
Although sharp convergence rates and statistical lower bounds have been established for some single-tensor estimation problems 
\citep{han2024tensor}, existing shared-private decomposition methods typically rely on iterative algorithms and mainly provide consistency or average-error guaranties \citep{shi2024personalized,hu2025personalized}. That is, computationally efficient estimation with view-wise optimal guaranties for heterogeneous multi-view tensors remains largely unexplored.

To address the methodological and theoretical gaps mentioned above, inspired by vector personalized PCA in \cite{shi2024personalized}, we introduce an optimization problem named \textit{Tucker personalized subspace principal component analysis (TPS-PCA)} on top of our extended tensor-version personalized PCA (TP-PCA). We summarize our main contributions in the following section.

\subsection{Contributions}  
The contributions of this work lie in the following aspects.

First, we model the observed multi-view tensors by decomposing each tensor into (i) a shared component that captures common patterns across all views and (ii) a view-specific private component that accounts for idiosyncratic variation unique to each view. Unlike the existing algorithms that require identically shaped tensors \citep{chen2025distributed,hu2025aggregated}, our model can handle multi-view tensors regardless of the tensor orders (shapes) and mode dimensions (sizes). It not only facilitates model interpretability but also enhances the proposed method's practical utility in downstream tasks with potentially heterogeneous tensor shapes and sizes; see Section \ref{sec:4} for details.

Second, inspired by \cite{shi2024personalized}, we introduce an optimization-based method named Tucker personalized subspace PCA (TPS-PCA) to identify the shared and private components across multi-view tensors in a computationally efficient manner: our optimization problem admits a one-step global optimum (Theorem \ref{theorem:go}). For comparison, the personalized PCA method in \cite{shi2024personalized} uses a multi-step Stiefel gradient descent algorithm \citep{boumal2023introduction}, whose local numerical convergence relies on Polyak-Lojasiewicz-type conditions that can be difficult to verify in practice \citep{karimi2016linear,chen2022faster}. Also, due to the joint optimization nature, the statistical property of personalized PCA can only be theoretically analyzed from a view-averaged perspective (see Theorem $1$ of \cite{shi2024personalized}). The proposed TPS-PCA estimators, on the other hand, possess a one-step minimax view-wise decoupling property which, to the best of our knowledge, has not been addressed in the literature \citep{zhang2019cross,han2022optimal}.

Third, we validate our proposed algorithm through extensive simulation studies and real-data applications. The numerical results show that TPS-PCA can effectively recover both the shared and private components in a computationally-efficient manner. Detailed analyzes of the electricity and temperature dataset \citep{chen2021tensor}, the financial dataset \citep{fan2022structural}, and the activity recognition dataset \citep{barshan2014recognizing} are also performed, which further demonstrate the potential applicability of the proposed personalized subspace learning method.

\subsection{Organization and Notations}
In Section \ref{sec:1}, we formally introduce the statistical problem of decoupling multi-view tensor observations. In Section \ref{sec:2}, we introduce the optimization-based method named Tucker personalized subspace PCA (TPS-PCA). 
Section \ref{sec:3} considers theoretical analyzes for the proposed method, including estimator consistency, minimax optimality, and shared dimension selection consistency. Section \ref{sec:4} validates the applicability of our algorithms from both simulation and real data analysis perspectives. We end our main article with Section \ref{sec:5}. Further discussions of the main text, along with the detailed theoretical proofs of the main theoretical results, can be found in the appendix.

Throughout the article, we use lowercase letters $x,y,v\in\mathbb{R}^{d}$ and uppercase letters $A,B\in\mathbb{R}^{d_1\times d_2}$ to denote (random) vectors and matrices, respectively. Let $v[i]$ be the $i$-th element of $v$. For $q\geq 1$, the $\ell_q$-norm of a column (row) vector $v$ $(v^\top\in\mathbb{R}^{d})$ is denoted by $\|v\|_q$. We abbreviate $\|v\|_2$ as $\|v\|$ if $q=2$. For vectors $v_1$ and $v_2$, we let $v_1^{\top}v_2$
denote the inner product of $v_1$ and $v_2$.

Let $I_d$ and $0_{d_1\times d_2}$ denote the $d$-dimensional identity matrix and the $d_1$ by $d_2$ zero matrix, respectively. We omit the subscript of the zero matrix when no confusion arises.
For any real matrix $A$, we define $A[i,j]$ as the element in the $i$-th row and $j$-th column. We further use $\text{rank}(A)$ and $\text{span}(A)$ to denote the rank of $A$ and the subspace spanned by the column vectors of $A$, respectively. We also let $\text{vec}(A)$ denote the column-wise vectorization of $A$. Let $\sigma_i(A)$ be the $i$-th largest non-zero singular value of $A$. For any symmetric matrix $A$, we denote $\lambda_i(A)$ as the $i$-th largest eigenvalue of $A$. We abbreviate $\lambda_i(A)$ as $\lambda_i$ without introducing ambiguity. We further use 
$\|A\|_{\op}$, $\|A\|_F$, and $\tr(A)$ to denote the operator norm, the Frobenius norm, and the trace of $A$, respectively. Let $\langle A,B\rangle\triangleq\tr(A^\top B)$ and $A\otimes B\in\mathbb{R}^{d_1^2\times d_2^2}$ be the inner product and the Kronecker product between $A$ and $B$, respectively. Define $[A,B]$ as the column-wise concatenation of $A$ and $B$. Let $C_1,c_1,C_2,\cdots$ denote the constants whose values may vary from line to line. 
For any $a_n,b_n$, we let $a_n\lesssim b_n$ and $a_n=O(b_n)$ be $a_n\leq Cb_n$ for sufficiently large $n$, while $a_n\asymp b_n$ means $a_n\lesssim b_n$ and $b_n\lesssim a_n$ simultaneously. Furthermore, we denote $a_n=o(b_n)$ as $|a_n/b_n|\rightarrow0$ for $n\rightarrow\infty$. For any $\epsilon>0$, let $a_n=O_p(b_n)$ denote $\mathbb{P}(|a_n/b_n|> c_{\epsilon})<\epsilon$ for some constant $c_\epsilon$. Let $a_n=o_p(b_n)$ denote $\mathbb{P}(|a_n/b_n|> c)\rightarrow0$ for any constant $c>0$.

We use $\cX\in\mathbb{R}^{d_1\times\cdots\times d_L}$ and $\cX[i_1,\cdots,i_L]$ to denote an order-$L$ tensor and its $(i_1,\cdots,i_L)$-th entry, respectively. For two order-$L$ tensors $\cX$ and $\cT$ with identical shapes, let $[\cX,\cT]_l\in\mathbb{R}^{d_1\times \cdots\times2d_l\times\cdots\times d_L}$ and $[\cX,\cT]_{L+1}\in\mathbb{R}^{d_1\times \cdots\times d_L\times2}$ denote the tensors obtained by concatenating $\cX$ and $\cT$ along the $l$-th mode and the $(L+1)$-th mode, respectively. 
For an order-$L$ tensor $\cX\in\mathbb{R}^{d_1\times\cdots\times d_L}$,
we define its mode-$l$ matricization
$X:=\cM_l(\cX)\in\mathbb{R}^{d_l\times\prod_{m\neq l}d_m}$
by arranging the remaining modes in the cyclic order
$(l+1,\ldots,L,1,\ldots,l-1)$.
More precisely, let $(m_1,\ldots,m_{L-1})$ denote this ordered list.
Then
$$
X[i_l,j]=\cM_l(\cX)[i_l,j]=\cX[i_1,\ldots,i_L],
\ \text{where}\ 
j=1+\sum_{t=1}^{L-1}(i_{m_t}-1)
\prod_{u=t+1}^{L-1}d_{m_u},
$$
The product is interpreted as $1$ if $L-1\leq t+1$ for $t=L-3,L-2$. 
Given $B\in\mathbb{R}^{d\times d_l}$, we define the mode-$l$ product of $\cX$ as
$$
\mathcal{B} = \mathcal{X} \times_l B,\ \text{where}\ 
\mathcal{B}[i_1,\cdots,j,\cdots,i_L]:=\sum_{i_l=1}^{d_l}\mathcal{X}[i_1,\cdots,i_l,\cdots,i_L]B[j,i_l].
$$
For a tensor $\mathcal{X}$, its Tucker decomposition is defined as
$$
\left\{
\begin{matrix}
\mathcal{X}=\mathcal{C}\times_1U_1\times_2U_2\cdots\times_LU_L,\ \cC\in\mathbb{R}^{r_1\times r_2\times\cdots\times r_L},\\
\cX[i_1,i_2,\cdots,i_L]=\sum\limits_{m_1=1}^{r_1}\sum\limits_{m_2=1}^{r_2}\cdots\sum\limits_{m_L=1}^{r_L}\cC[m_1,m_2,\cdots,m_L]U_1[i_1,m_1]U_2[i_2,m_2]\cdots U_L[i_L,m_L].
\end{matrix}
\right.
$$
$\cC$ denotes the core tensor, and each $U_l\in\cO(d_l, r_l):=\{U\in\mathbb{R}^{d_l\times r_l}|U^\top U=I_{r_l}\}$ denotes the loading matrix for the $l$-th mode, where $\cO(d_l, r_l)$ consists of all $d_l\times r_l$ column-orthonormal matrices. We denote $\|\cX\|_F:=(\sum_{i_1=1}^{d_1}\cdots\sum_{i_L=1}^{d_L}\cX^2[i_1,i_2,\cdots,i_L])^{1/2}$ as the tensor Frobenius norm for $\cX$.
Further detailed tutorials for tensor computations and applications can be found in \cite{liu2022tensor}.

\section{Problem Formulation}
\label{sec:1}
To formulate the statistical problem, assume that we have multi-view $(K\geq2)$ tensor observations denoted as 
$\cD_k :=\{\cX_1^{(k)},\cdots,\cX_n^{(k)}\},$
where $k\in[K]$ is the view index. For each $k$, the $i$-th observation $\cX_i^{(k)}\in \reals^{p^{(k)}_1\times\cdots\times p^{(k)}_{L_k}}$ is an order-$L_k$ tensor, and $p^{(k)}_j$ is the dimension of the $j$-th mode. To model the shared-private structure for each $\cX_i^{(k)}$, we impose the following model:
\begin{equation}\label{eq:spm}
    \cX_i^{(k)}=\sum_{m_1=1}^{r_s}v_{i,m_1} \cS_{m_1}^{(k)} + \sum_{m_2=1}^{r_k-r_s}w^{(k)}_{i,m_2} \cP_{m_2}^{(k)} + \cE_i^{(k)},
\end{equation}
where $\cS_{m_1}^{(k)}$ and $\cP_{m_2}^{(k)}$ are shared and private basis tensors, while $v_{i,m_1}$ and $w^{(k)}_{i,m_2}$ are the corresponding factor coefficients for each subject. If the intrinsic tensor structures are ignored and the tensor observations within each view are directly vectorized, \eqref{eq:spm} reduces to vector personalized PCA model \citep{shi2024personalized}.
\begin{equation}
\label{eq:spm_vec}
\begin{aligned}
\underbrace{\left[\vec(\cX_1^{(k)}),\cdots,\vec(\cX_n^{(k)})\right]}_{p^{(k)}_1p^{(k)}_2\cdots p^{(k)}_{L_k} \times n}&= \underbrace{\left[\vec(\cS_{1}^{(k)}),\cdots,\vec(\cS_{r_s}^{(k)})\right]}_{p^{(k)}_1p^{(k)}_2\cdots p^{(k)}_{L_k} \times r_s}V^{\top}\\&\quad+ \underbrace{\left[\vec(\cP_{1}^{(k)}),\cdots,\vec(\cP_{r_k-r_s}^{(k)})\right]}_{p^{(k)}_1p^{(k)}_2\cdots p^{(k)}_{L_k} \times (r_k-r_s)}(W^{(k)})^{\top}\\&\quad + \underbrace{\left[\vec(\cE_1^{(k)}),\cdots,\vec(\cE_n^{(k)})\right]}_{p^{(k)}_1p^{(k)}_2\cdots p^{(k)}_{L_k} \times n},
\end{aligned}
\end{equation}
where $V\in\mathbb{R}^{n\times r_s}$ and $W^{(k)}\in\mathbb{R}^{n\times (r_k-r_s)}$ satisfy  $V^{\top} V = I_{r_s}$, $(W^{(k)})^{\top} W^{(k)} = I_{r_k-r_s}$, and $V^{\top}W^{(k)}=0$ for each $k\in[K]$. Without loss of generosity, we require $0<r_s<\min_k\{r_k\}\leq\max_k\{r_k\}\ll n$. The cases $r_k=r_s$ and $r_s=0$ correspond to degenerate settings, where $W^{(k)}$ and $V$ are set to zero matrices, respectively.
Although it is  theoretically feasible to decouple 
$\span(V)$ and $\span(W^{(k)})$ based on \eqref{eq:spm_vec}, vectorizing $\cX_i^{(k)}$ inevitably results in computational inefficiency. For instance, an order-$3$ tensor with $p_j=100$ for each $j\in[3]$ becomes a vector of dimension $1,000,000$, imposing a heavy burden on numerical computation.

Fortunately, $\cX_i^{(k)}$ typically possesses latent low-rank structures in practice \citep{kolda2009tensor}, where the latent shared and private basis tensors, $\cS_{m_1}^{(k)}$ and $\cP_{m_2}^{(k)}$, exhibit Tucker-form \citep{malik2018low,sun2020lowrank} or CP-form decomposition structures \citep{han2024cp}. Such latent structural information can be exploited to improve statistical performance, whereas direct vectorization, such as \eqref{eq:spm_vec}, breaks these intrinsic structures, causing a loss of statistical efficiency \citep{he2023one}. In these low-rank cases, we can, in general, stack $\{\cS_{m_1}^{(k)}\}_{m_1}$ and $\{\cP_{m_2}^{(k)}\}_{m_2}$ along the $(L_k+1)$-th mode to obtain
\begin{equation}
\label{eq:spm_conca}
\cX^{(k)} 
= \cS^{(k)}\times_{L_k+1} V + \cP^{(k)}\times_{L_k+1} W^{(k)} + \cE^{(k)}
=
[\cS^{(k)},\cP^{(k)}]_{L_k+1}\times_{L_k+1}[V,W^{(k)}] + \cE^{(k)},
\end{equation}
where
$\cS^{(k)}
=[\cS_1^{(k)},\cdots,\cS_{r_s}^{(k)}]_{L_k+1}$ and  $\cP^{(k)}=[\cP_1^{(k)},\cdots,\cP_{r_k-r_s}^{(k)}]_{L_k+1}$ are the concatenated shared and private tensors for view $k$, respectively. 
Owing to the modification from tensor vectorization in \eqref{eq:spm_vec} to tensor concatenation in \eqref{eq:spm_conca}, the low-rank structures of
$\cS_{m_1}^{(k)}$ and $\cP_{m_2}^{(k)}$, regardless of  Tucker-form or CP-form, can be preserved in $[\cS^{(k)},\cP^{(k)}]_{L_k+1}$. As such, we can directly
analyze $[\cS^{(k)},\cP^{(k)}]_{L_k+1}$ without losing latent structural information. Due to the fact that CP decomposition is a special case of the Tucker decomposition \citep{kolda2009tensor}, in this work, we impose a Tucker low-rank structure on $[\cS^{(k)},\cP^{(k)}]_{L_k+1}$ such that
$\left[\cS^{(k)},\cP^{(k)}\right]_{L_k+1}= \cC^{(k)} \times_1U^{(k)}_{1}\cdots \times_{L_k} U^{(k)}_{L_k}$
where $\cC^{(k)}\in \reals^{r^{(k)}_1\times\cdots\times r^{(k)}_{L_k}\times r_k}$ and $U^{(k)}_{j} \in \reals^{p^{(k)}_j\times r^{(k)}_j}$ are the corresponding core tensor and loading matrices. Then, we can reformulate \eqref{eq:spm} as the following equivalent form
\begin{equation}\label{TJIVE model}
\begin{aligned}
\cX^{(k)} 
&=
[\cS^{(k)},\cP^{(k)}]_{L_k+1}\times_{L_k+1}[V,W^{(k)}] + \cE^{(k)}\\
&=\underbrace{\cC^{(k)} \times_1U^{(k)}_{1}\cdots \times_{L_k} U^{(k)}_{L_k}\times_{L_k+1}[V,W^{(k)}]}_{\cT^{(k)}}+\cE^{(k)},
\end{aligned}
\end{equation}
where $\cX^{(k)}=[\cX_1^{(k)},\cdots,\cX_n^{(k)}]_{L_k+1}$, and $\cE^{(k)}$ is the error tensor. 

In essence, \eqref{TJIVE model} is closely related to a broad class of existing methods and models. For instance, the two-view ($K=2$) tensor canonical correlation analysis 
(tensor CCA, see also \cite{chen2021tensor,min2019tensor})
aims to identify canonical tensors $\cS^{(1)}$ and $\cS^{(2)}$ given two tensors $\cX^{(1)}$ and $\cX^{(2)}$  such that
$$
\left\{\hat{\cS}^{(1)}, \hat{\cS}^{(2)}\right\}
=
\argmax_{\cS^{(1)}, \cS^{(2)}}
\text{Corr}\left(\langle \cS^{(1)},\cX^{(1)}\rangle,\langle \cS^{(2)},\cX^{(2)}\rangle\right),
$$
where $\text{Corr}(\cdot,\cdot)$ denotes the correlation between two random variables, and $
\langle \cU,\cX\rangle$ represents the tensor inner product. In the same spirit, we aim to retrieve the shared components from multiple ($K\geq 2$) tensor views, while \eqref{TJIVE model} explicitly models the private components as well. In parallel, the classical joint and individual variation explained 
(JIVE) model \citep{lock2013joint,yang2025estimating} decouples multiple data views into a joint component $(J^{(k)})$ shared across views and 
individual components $(A^{(k)})$ specific to each view:  
$$
X^{(k)}
=
J^{(k)}
+
A^{(k)}
+
E^{(k)}
=
S^{(k)}V^\top
+
P^{(k)}(W^{(k)})^\top
+
E^{(k)},
\  k\in[K].
$$
When $\cX_i^{(k)}$ is in vector form $(L_k=1\ \text{for each}\ k\in[K])$, $\cX^{(k)}$ reduces to a data matrix and \eqref{TJIVE model} coincides with the classical JIVE model. As such, \eqref{TJIVE model} can also be viewed as a tensor-version JIVE (TJIVE) model.

The above discussion highlights the significance of investigating \eqref{TJIVE model} by identifying the 
shared and view-specific subspaces, i.e., $\span(V)$ and $\span(W^{(k)})$. Motivated by these connections, our main statistical task narrows down to estimating $\span(V)$ and $\span(W^{(k)})$ from multi-view tensors $\{\cX^{(k)}\}_k$.

\section{Methodology}
\label{sec:2}
In this section, we present the decomposition method for multi-view tensor observations. The first step is to obtain sufficiently accurate mode-wise subspace estimators $\{\hat{U}^{(k)}_j\}^{k\in[K]}_{j\in[L_k]}$. Given a specific $k\in[K]$, one may consider borrowing subspace information from  other views to estimate $\{U_j^{(k)}\}_{j\in[L_k]}$. However, the tensor mode dimensions from different views are typically different in many applications, except for the mode corresponding to $n$ subjects (see the right panel of Figure \ref{fig:intro_examples_both} for example). Hence, $\{U^{(k)}_j\}^{k\in[K]}_{j\in[L_k]}$ cannot be estimated in a cross-view manner in a general way. Nevertheless, numerous tools are available for obtaining well-estimated $\{\hat{U}^{(k)}_j\}_{j\in[L_k]}$ under mild conditions on $\cX^{(k)}$. For instance, they can be estimated either by directly applying HOOI \citep{de2000best} to each $\cX^{(k)}$, or by leveraging auxiliary knowledge from previous statistical tasks (e.g., via transfer learning) \citep{he2024representational}.  As such, following \cite{zhang2018tensor,han2022optimal,chen2025distributed}, $\{\hat{U}^{(k)}_j\}^{k\in[K]}_{j\in[L_k]}$ can be regarded as given subspace estimators satisfying certain conditions (see Assumption \ref{assum:2} in Section \ref{sec:3}). Then, we can construct
\begin{equation}\label{eq:Mj}
    \tilde{M}^{(k)}:= \cM_{L_k+1}\left(\cX^{(k)}\times_1(\hat{U}^{(k)}_{1})^{\top}\cdots \times_{L_k}(\hat{U}^{(k)}_{L_k})^{\top} \right).
\end{equation}
The projection mindset in \eqref{eq:Mj} is a standard technique in tensor decomposition tasks; see 
\citep{yu2022projected,he2024matrix,barigozzi2026statistical} for further discussion.

If the goal is only to estimate the loading subspace of the subject-mode, one can directly take the leading $r_k$ left singular vectors of $\tilde{M}^{(k)}$, denoted by $\tilde{V}^{(k)}$, and further achieve dimension reduction for $\cX^{(k)}$ along the $(L_k+1)$-th mode. As our main statistical goal is to identify the shared and view-specific components, we need to further retrieve $V$ and $W^{(k)}$ from these multi-view $\tilde{M}^{(k)}$'s. To acquire a reliable baseline estimator, we start by proposing the tensor-version personalized PCA (TP-PCA) problem, adapted from the seminal personalized PCA problem from \cite{shi2024personalized}:
\begin{equation}\label{eq:optimization}
\begin{aligned}    
\left\{\tilde{A},\left\{\tilde{B}^{(k)}\right\}_k\right\}=    
&\argmax_{A,\{B^{(k)}\}_k}\sum_{k=1}^{K} \tr\left[ A^{\top}\tilde{M}^{(k)}(\tilde{M}^{(k)})^{\top} A\right]+ \sum_{k=1}^{K} \tr\left[(B^{(k)})^{\top}\tilde{M}^{(k)}(\tilde{M}^{(k)})^{\top} B^{(k)}\right],\\
    &\text{subject to }\; A^{\top}A= I_{r_s},\; (B^{(k)})^{\top}B^{(k)}=I_{r_k-r_s},\; A^{\top}B^{(k)}=0,\ k\in[K].
    \end{aligned}
\end{equation}

Following the arguments in \cite{shi2024personalized}, we can show that the global optimum of \eqref{eq:optimization} can consistently decouple the shared components from the private ones (see Proposition \ref{prop:1} for details). However, \eqref{eq:optimization} inherits the same computational limitations as vector personalized PCA, where one needs to apply iterative Stiefel gradient descent algorithms for a locally optimal solution. Moreover, as we shall see in Section \ref{sec:3}, due to the joint optimization nature of personalized PCA, the statistical estimation error bounds for $\{\tilde{A},\{\tilde{B}^{(k)}\}_k\}$ are presented in a view-averaged form, thereby lacking view-wise statistical guaranties for $\tilde{A}$ and each $\tilde{B}^{(k)}$ separately. This limitation motivates us to take TP-PCA as a starting point and develop an alternative optimization problem. Specifically, we consider the following \emph{Tucker personalized subspace PCA (TPS-PCA)} problem as an ideal surrogate for \eqref{eq:optimization}. We can solve
\begin{equation}\label{eq:personalized_subspace}
\begin{aligned} 
    \left\{\hat{V},\left\{\hat{W}^{(k)}\right\}_k\right\}=&\argmin_{A,\{B^{(k)}\}_k}\sum_{k=1}^{K}  \left\|AA^\top+B^{(k)}(B^{(k)})^{\top}-\tilde{V}^{(k)}(\tilde{V}^{(k)})^{\top}\right\|_F^2,\\
    &\text{subject to }\; A^{\top}A= I_{r_s},\; (B^{(k)})^{\top}B^{(k)}=I_{r_k-r_s},\; A^{\top}B^{(k)}=0,\ k\in[K].
\end{aligned}
\end{equation}
After some simple algebraic operations, \eqref{eq:personalized_subspace} can be transformed into the following equivalent form:
\begin{equation}\label{eq:personalized_subspace2}
\begin{aligned} \left\{\hat{V},\left\{\hat{W}^{(k)}\right\}_k\right\}=&\argmax_{A,\{B^{(k)}\}_k}\sum_{k=1}^{K}\tr\left[ A^{\top}\tilde{V}^{(k)}(\tilde{V}^{(k)})^{\top} A\right]+ \sum_{k=1}^{K} \tr\left[ (B^{(k)})^{\top}\tilde{V}^{(k)}(\tilde{V}^{(k)})^{\top} B^{(k)}\right],\\
    &\text{subject to }\; A^{\top}A= I_{r_s},\; (B^{(k)})^{\top}B^{(k)}=I_{r_k-r_s},\; A^{\top}B^{(k)}=0,\ k\in[K].
    \end{aligned}
\end{equation}
Essentially, \eqref{eq:personalized_subspace2} replaces $\tilde{M}^{(k)}(\tilde{M}^{(k)})^{\top}$ in \eqref{eq:optimization} with local projection matrices $\tilde{V}^{(k)}(\tilde{V}^{(k)})^{\top}$, where $\tilde{V}^{(k)}$ consists of the leading $r_k$ left singular vectors of $\tilde{M}^{(k)}$. This modification leads to a different procedure for solving \eqref{eq:personalized_subspace2}. Indeed, unlike \eqref{eq:optimization} where only a local optimum can be theoretically found after multiple iterations, the global optimum of \eqref{eq:personalized_subspace} can be achieved in one-step. We formalize this result in the following theorem.
\begin{theorem}[One-step global optimum of \eqref{eq:personalized_subspace}]\label{theorem:go}
    The one-step estimators $\{\hat{V},\{\hat{W}^{(k)}\}_k\}$ formed by Algorithm \ref{alg:decouple} below achieve the global optimum of \eqref{eq:personalized_subspace}.
\end{theorem}

We now present the one-step decomposition procedure. First, we aggregate $\{\tilde{V}^{(k)}(\tilde{V}^{(k)})^\top\}$ and take the leading $r_s$ eigenvectors of the average projection matrix
\begin{equation}\label{eq:apm}
    \tilde{\Sigma}_{0}:=\frac{1}{K}\sum_{k=1}^{K}\tilde{V}^{(k)}(\tilde{V}^{(k)})^{\top},
\end{equation}
to acquire $\hat{V}$. Then, we can obtain the view-specific estimator $\hat{W}^{(k)}$ by taking the leading $r_k-r_s$ left singular subspace of $(I_n-\hat{V}\hat{V}^{\top})\tilde{V}^{(k)}$.
In fact, $\hat{V}\hat{V}^{\top}$ minimizes the mean squared projection metric to all $\tilde{V}^{(k)}(\tilde{V}^{(k)})^{\top}$ and is often referred to as the physical barycenter of $\span(\tilde{V}^{(k)})$. 
Indeed, the average projection matrix has long been used in multiple fields that require the aggregation of subspaces, such as satellite meteorology \citep{crone1995statistical}, transfer learning \citep{li2024knowledge,he2024representational}, multi-source decomposition \citep{feng2018angle,hu2025aggregated,yang2025estimating}, and robust model averaging \citep{liski2016combining,fan2019distributed,he2025distributed,li2025robust}. For clarity, we summarize the two procedures above in Algorithm \ref{alg:decouple}.
\renewcommand{\algorithmicrequire}{\textbf{Require}}
\renewcommand{\algorithmicensure}{\textbf{Ensure}}
\renewcommand{\algorithmicreturn}{\textbf{Return}}
\newcommand{\RETURNVALUE}{\item[\algorithmicreturn]}
\begin{algorithm}
\caption{Decoupling shared and private features across multi-view tensors.}\label{alg:decouple}
\begin{algorithmic}[1]
\REQUIRE ~~\\
    $\{\cX^{(k)}\}_{k=1}^{K}$; data mode estimators $\{\hat{U}^{(k)}_{1},\cdots, \hat{U}^{(k)}_{L_k}\}_{k=1}^{K}$;\\
\ENSURE ~~\\
    \STATE Initialization: acquire $\tilde{M}^{(k)}$ via \eqref{eq:Mj}, take the leading $r_k$ left singular vectors of $\tilde{M}^{(k)}$, denoted as $\tilde{V}^{(k)}$;
    \STATE Shared subspace: acquire the shared subspace estimator $\hat{V}$ by taking the leading $r_s$ eigenvectors of $\tilde{\Sigma}_{0}:=\sum_{k=1}^{K}\tilde{V}^{(k)}(\tilde{V}^{(k)})^{\top}/K$;
    \STATE Private subspaces: after obtaining $\hat{V}$, for $k\in[K]$, acquire $\hat{W}^{(k)}$ by taking the leading $r_k-r_s$ left singular vectors of $(I_n-\hat{V}\hat{V}^{\top})\tilde{V}^{(k)}$;
\RETURNVALUE $\hat{V}, \hat{W}^{(k)}$ for $k\in [K]$.
\end{algorithmic}
\end{algorithm}

To understand Theorem \ref{theorem:go}, note that the maximum of $\tr[(B^{(k)})^{\top}\tilde{V}^{(k)}(\tilde{V}^{(k)})^{\top} B^{(k)}]$ equals $r_k-r_s$ regardless of the choice of $A$. Therefore, the private part of \eqref{eq:personalized_subspace} can be discarded from the joint optimization, which is then reduced to a shared sub-problem, i.e., solving an eigenvalue decomposition of $\tilde\Sigma_0$ in \eqref{eq:apm}. On the other hand, the maximum of $\tr[(B^{(k)})^{\top}\tilde{M}^{(k)}(\tilde{M}^{(k)})^{\top} B^{(k)}]$ is a function of $A$. As such, the private part of \eqref{eq:optimization} is coupled with the shared part, and an iterative algorithm that alternately updates $A$ and $B^{(k)}$ is required. Hereafter, we shall refer to $\{\hat{V},\{\hat{W}^{(k)}\}_k\}$ collectively as the TPS-PCA estimators.

As a by-product of Algorithm \ref{alg:decouple}, $\tilde{\Sigma}_0$ in \eqref{eq:apm} also offers a natural selection criterion for the shared subspace dimension $r_s$, whereas \cite{shi2024personalized} did not provide a direct procedure for shared dimension selection. The statistical insight is that the first $r_s$ eigenvalues of $\tilde{\Sigma}_0$ approximate $1$ under mild model and signal-to-noise assumptions, while the others are strictly smaller than $1$. As such, the spectral signal in $\tilde{\Sigma}_0$ carries direct information about $r_s$. Similar dimension selection ideas
based on aggregated projection matrices have been recently adopted in studies including \cite{li2024knowledge,hu2025aggregated,chen2025distributed}. Likewise, we apply the following truncated eigenvalue criterion to $\tilde{\Sigma}_0$ to select $r_s$.

Specifically, let
$\hat\lambda_1\geq\cdots\geq\hat\lambda_n$
denote the eigenvalues of $\tilde{\Sigma}_0$. We estimate $r_s$ by
retaining the leading eigenvalues whose deviations from $1$ remain below
a prespecified tolerance $\zeta$. Equivalently, we adopt the following
truncation criterion:
\begin{equation}
\label{eq:shared-rank-selection}
\begin{aligned}
\hat r_s
&:=
\max\left\{0,\sum_{r=1}^{R_s}
\mathbb{I}\left\{1-\hat\lambda_r\leq\zeta\right\}\right\}.
\end{aligned}
\end{equation}
Here, $R_s\in(0,n]$ denotes the maximum range for searching for a
suitable estimator of $r_s$, while $\zeta$ serves as a tolerance
threshold for the individual deviations of the leading eigenvalues
$\{\hat\lambda_\ell\}_{l=1}^{R_s}$ from $1$. 
A larger value of $\zeta$ corresponds to a lower truncation threshold
$1-\zeta$ and therefore generally permits a larger estimated shared
dimension, while a smaller value of $\zeta$ leads to a more conservative
selection. An appropriately chosen $\zeta$ ensures that $\hat r_s$ is a
consistent estimator of $r_s$ (see Corollary 
\ref{cor:rank consistency} for further details). Note that 
\eqref{eq:shared-rank-selection} assumes
$\{\{r_j^{(k)}\}_{j\in[L_k]},r_k\}_{k\in[K]}$
to be known in advance. However, these quantities are typically unavailable
and need to be estimated beforehand. Following \cite{chen2025distributed}, 
we can apply the eigenvalue-ratio method \citep{han2022rank} to the left singular
space of $\cM_j(\cX^{(k)})$ for consistently selecting
$r_j^{(k)}$ and $r_k$. One can also consider the
existing rank determination methods, such as the bootstrap-based procedure \citep{yu2025testing}, eigenvalue-ratio-type estimators \citep{xia2015consistently,yu2019robust,wang2022high}, and information criteria \citep{bai2002determining,han2022rank}.

\section{Theory}\label{sec:3}
In this section, we conduct a thorough theoretical analysis of the estimators proposed in Section \ref{sec:2}. We start by imposing some technical assumptions.

\begin{assumption}[Sub-Gaussian noise]\label{assum:1}
    For each $k\in [K]$, we assume that $\cE^{(k)}$ has i.i.d. sub-Gaussian entries with mean zero and variance proxy $e_k^2$ with $e:=\max_{k\in[K]}e_k<+\infty$. Moreover, we assume that $\{\mathcal E^{(k)}\}_{k\in[K]}$ are mutually independent across views.
\end{assumption}

\begin{assumption}[Data mode estimators]\label{assum:2}
For any $U,V\in\mathbb R^{p\times r}$ with orthonormal columns, we define $
\sin\Theta(U,V):=\text{diag}(\sin\theta_1,\cdots,\sin\theta_r)
$ where $\theta_i=\arccos(\sigma_i(U^\top V))$. Assume that we have acquired sufficiently good data mode subspace estimators $\{\hat{U}^{(k)}_j\}^{k\in[K]}_{j\in[L_k]}$, such that
\begin{equation}\label{eq:init_data}
\max_{k\in[K]}\max_{j\in[L_k]}\left[\left\|\sin\Theta\left(\hat{U}^{(k)}_j,U^{(k)}_j\right)\right\|_{\op}\right]\leq \frac{1}{2}.
\end{equation}
\end{assumption}

\begin{assumption}[Misalignment of the private subspaces]\label{assum:3}
    Assume that there exists a positive constant $\theta\in (0,1)$, such that
    $$\left\|\frac{1}{K}\sum_{k=1}^{K} W^{(k)}(W^{(k)})^{\top}\right\|_{\op}\leq 1-\theta.$$
\end{assumption}

\begin{assumption}
[Signal-to-noise regime]
\label{assum:4}
Suppose that $\max_{k\in[K]}\{L_k, r^{(k)}_1,\cdots, r^{(k)}_{L_k}, r_k\}$ remain bounded while $\max_{k\in[K]}\max_{j\in[L_k]}\{p_j^{(k)}\}\lesssim n$. Let $\sigma_{\max}$ and $\sigma_{\min}$ denote the maximal and minimal singular values of $\cM_{L_k+1}(\cC^{(k)})/\sqrt{n}\in \reals^{r_k \times r^{(k)}_{\Pi}}$ across all $k\in[K]$ where $\cC^{(k)}$ is defined in \eqref{TJIVE model} and $r^{(k)}_{\Pi}:=\prod_{j=1}^{L_k}r_j^{(k)}$. Then, we further assume that the signal strength for each tensor satisfies  $\sigma_{\min}>C$ for a positive constant $C$.
\end{assumption}

Assumption \ref{assum:1} is standard in the tensor decomposition literature; see
\cite{xia2019sup,zhou2022optimal} for similar assumptions.
Assumption \ref{assum:2} requires $\span(\hat{U}^{(k)}_j)$ to lie within a fixed neighborhood of its population counterpart, thereby providing valid initializations for estimating the shared and private subspaces. Similar conditions are also imposed in related works such as \cite{zhang2018tensor,han2022optimal} and \cite{chen2025distributed}.
Assumption \ref{assum:3} merely excludes the degenerate case where the private
subspaces contain an additional direction that is common to all views, which is
a standard identifiability requirement in shared-private decoupling models \citep{shi2024personalized,li2024knowledge,hu2025aggregated,yang2025estimating}.
Assumption \ref{assum:4}, adapted from assumptions imposed in Theorem $1$ of \cite{zhang2018tensor}, is the signal-to-noise regime of our theoretical analysis. In this regime, the local subspace estimator $\tilde{V}^{(k)}$ can be accurately estimated (see Lemmas \ref{lem:signal} and \ref{lem:noise} of the supplementary file) to guaranty the theoretical advantage of TPS-PCA over TP-PCA (see Remark \ref{rem:TPvsTPS}). To illustrate the mildness of the signal strength condition, consider the classical pervasiveness condition in the tensor factor
model \citep{chen2024rank} where the signal strength for $\cX^{(k)}=[\cX_1^{(k)},\cdots,\cX_{n}^{(k)}]_{L_k+1}$
satisfies
$
\sigma_{\min}^2
\asymp
\prod_{j=1}^{L_k} p_j^{(k)}.
$
Consequently, the signal requirement
$\sigma_{\min}>C$ naturally follows.

\subsection{View-wise Decoupling Consistency}
Based on the above assumptions, we first provide the upper bounds for the TPS-PCA estimators. In particular, the outputs of TPS-PCA achieve view-wise decoupling consistency and provide sharper rates of convergence compared to those of TP-PCA (tensor-version personalized PCA). As a theoretical baseline, we first present the consistency of the TP-PCA estimators under the assumptions given in the last section. 

\begin{proposition}[Statistical upper bound for TP-PCA]\label{prop:1}
    Let $\{\tilde{A},\{\tilde{B}^{(k)}\}_k\}$ be the solution of the optimization problem \eqref{eq:optimization}. Under Assumptions \ref{assum:1} to \ref{assum:4} and the conditions in Theorem $1$ of \cite{shi2024personalized}, we have
\begin{equation}\label{eq:prop1}
\left\|\tilde{A}\tilde{A}^{\top}-VV^{\top}\right\|^2_F+\frac{1}{K}\sum_{k=1}^{K}\left\|\tilde{B}^{(k)}(\tilde{B}^{(k)})^{\top}-W^{(k)}(W^{(k)})^{\top}\right\|^2_F=O_p\left(\frac{\sigma_{\max}^2}{\theta\sigma_{\min}^4}\right).
\end{equation}
\end{proposition}

When $\sigma^2_{\max}/\sigma^4_{\min}\rightarrow0$, 
the shared component $\tilde{A}$ is guaranteed to be separated consistently from the private components $\{\tilde{B}^{(k)}\}_k$, as stated by Proposition \ref{prop:1}. However, the errors for TP-PCA estimators
fail to provide the decoupling property from a view-wise perspective. The key technical challenge is 
that the convergence error for each view $k$, i.e., $\|\tilde{A}\tilde{A}^{\top}-VV^{\top}\|^2_F+\|\tilde{B}^{(k)}(\tilde{B}^{(k)})^{\top}-W^{(k)}(W^{(k)})^{\top}\|^2_F$, cannot be well-separated from the left-hand side of \eqref{eq:prop1}. In contrast, for the TPS-PCA estimators, we can readily decouple the shared-private error from the view-wise perspective.

\begin{theorem}[Statistical upper bounds for TPS-PCA]
\label{theorem:upper bound}
Let $[\hat{V},\hat{W}^{(k)}]$ be the TPS-PCA estimators for $\cX^{(k)}$. Based on Assumptions \ref{assum:1} to \ref{assum:4}, we have, for any $k\in[K]$, that
\begin{equation}
\label{main rates for V}
\left\|\hat{V}\hat{V}^{\top}-VV^{\top}\right\|^2_F
=O_p\left(\frac{1}{K\theta^2\sigma_{\min}^2}
+
\frac{1}{\theta^2\sigma_{\min}^4}\right),
\end{equation}
\begin{equation}
\label{main rates for Wk}
\left\|\hat{W}^{(k)}\left(\hat{W}^{(k)}\right)^{\top}-W^{(k)}\left(W^{(k)}\right)^{\top}\right\|^2_F
=O_p\left(\frac{1}{\sigma^2_{\min}}
+
\frac{1}{K\theta^2\sigma_{\min}^2}
+
\frac{1}{\theta^2\sigma_{\min}^4}\right).
\end{equation}
\end{theorem}

As $1/\sigma^{2}_{\min}\rightarrow0$, the estimation errors of both $\hat{V}$ and $\hat{W}^{(k)}$ converge 
to $0$. As such, 
Theorem \ref{theorem:upper bound} indicates a view-wise decoupling consistency for TPS-PCA estimators. It arises from the one-step nature of Algorithm \ref{alg:decouple}, where both $\hat{V}$ and $\hat{W}^{(k)}$ have closed-forms and can be theoretically analyzed. For a better comparison with the average error bound in Proposition \ref{prop:1}, we further combine \eqref{main rates for V} and \eqref{main rates for Wk} to obtain the following joint convergence rate of $[\hat{V}, \hat{W}^{(k)}]$.
\begin{corollary}
For any $k\in[K]$, based on the conditions and notations of
Theorem \ref{theorem:upper bound}, we have
\begin{equation}
\label{main rates for upper bound}
\left\|\hat{V}\hat{V}^{\top}-VV^{\top}\right\|^2_F
+
\left\|\hat{W}^{(k)}(\hat{W}^{(k)})^{\top}-W^{(k)}\left(W^{(k)}\right)^{\top}\right\|^2_F
=O_p\left(\frac{1}{\sigma^2_{\min}}
+
\frac{1}{K\theta^2\sigma_{\min}^2}
+
\frac{1}{\theta^2\sigma_{\min}^4}\right).
\end{equation}
\end{corollary}

\begin{remark}
[Rate comparison]
\label{rem:TPvsTPS}
According to Assumption \ref{assum:3}, $\theta$ is of constant order. Since $K>1$,  it follows that 
$1/K\theta^2\sigma^2_{\min}=O(1/\sigma^2_{\min})$.
For the $1/\theta^2\sigma_{\min}^4$ term, it can be a higher-order term compared to $1/\sigma_{\min}^2$ as long as Assumption \ref{assum:4} holds. As such, $1/\sigma_{\min}^2$ is the main term in TPS-PCA. Under a homogeneous setting such that the convergence error for each view $k$ is of the same order for TP-PCA, we can conclude that
$$
\frac{1}{\sigma^2_{\min}}
+
\frac{1}{K\theta^2\sigma_{\min}^2}
+
\frac{1}{\theta^2\sigma_{\min}^4}
\lesssim
\frac{1}{\sigma^2_{\min}}
\lesssim
\frac{\sigma_{\max}^2}{\theta\sigma_{\min}^4},
$$
where the right-hand side is the rate in Proposition \ref{prop:1}. This implies that TPS-PCA estimators not only possess the view-wise decoupling property but also have sharper rates of convergence than those of TP-PCA, even under homogeneous settings.
\end{remark}

\begin{corollary}
[Shared dimension selection consistency]
\label{cor:rank consistency}
For $0\leq r_s\leq R_s\leq n$, assume that
$1/\sigma^2_{\min}K+1/\sigma^4_{\min}=o(1)$, while the fixed truncation tolerance
$\zeta$ in \eqref{eq:shared-rank-selection} satisfies
$0<\zeta<\theta$. Then, under the same conditions in Theorem \ref{theorem:upper bound}, we have
$
\P(\hat r_s=r_s)\to1.
$
\end{corollary}
As a byproduct of Theorem \ref{theorem:upper bound}, the consistency of selecting $r_s$ by using \eqref{eq:shared-rank-selection} can also be established. In fact, the conditions in Corollary \ref{cor:rank consistency} reflect the effectiveness of \eqref{eq:shared-rank-selection}.
According to Assumption \ref{assum:3}, the first $r_s$ population eigenvalues equal $1$, while the rest remain bounded above by $1-\theta$. The requirement $0<\zeta<\theta$ places the truncation threshold $1-\zeta$ strictly between these two spectral blocks, i.e., $1-\theta<1-\zeta<1$, to exactly separate the remaining eigenvalues from the leading $r_s$ ones, thereby ensuring the identification of $r_s$ at the population level. The condition $1/\sigma^2_{\min}K+1/\sigma^4_{\min}=o(1)$ further requires the signal strength to dominate the stochastic error during the estimation procedures to ensure that the eigenvalues of $\widetilde{\Sigma}_0$ concentrate around their population counterparts. Figure \ref{fig:rs_selection_spectrum} in Section \ref{sec:simulation} provides a finite-sample illustration, where a stronger alignment of the private subspaces indicates a smaller population eigenvalue gap between shared and private spectral blocks, while a larger noise level leads to a stronger perturbation of the empirical eigenvalues.

\subsection{View-wise minimax optimality}
In this section,
we leverage the framework of the minimax information lower bound   \citep{cai2013sparse,cai2021transfer} and claim that $\hat{V}$ and $\hat{W}^{(k)}$ given by TPS-PCA achieve the view-wise minimax optimality. We first provide the parameter space $\Theta$. Given a fixed set of hyper-parameters $\{\{p^{(k)}, r^{(k)},L_k\}_{k\in[K]};r_s;n;K;\sigma_{\max};\sigma_{\min}\}$, where $p^{(k)}=(p^{(k)}_1,\cdots, p^{(k)}_{L_k})$ and $r^{(k)}=(r^{(k)}_1,\cdots, r^{(k)}_{L_k},r_k)$, we define
\begin{equation}
\label{Theta}
\begin{aligned}
\Theta:=\biggl\{
&\left\{\cT^{(k)}=\cC^{(k)}\times_1 U_1^{(k)}\times_2
U_2^{(k)}\times_3\cdots\times_{L_k} U_{L_k}^{(k)}\times_{L_k+1}[V,W^{(k)}]\right\}_k\Big|\  \sigma_{\max}\leq C\sigma_{\min}, \\
&\sigma_{\min}\leq\min_{k\in[K]}\{\sigma_{r_k}(\cM_{L_k+1}(\cC^{(k)})/\sqrt{n})\}\leq\cdots\leq\max_{k\in[K]}\{\sigma_{1}(\cM_{L_k+1}(\cC^{(k)})/\sqrt{n})\}\leq\sigma_{\max};  \\
& V^\top V=I_{r_s}, \{\{(U_j^{(k)})^\top U_j^{(k)}=I_{r_j^{(k)}}\}_{j=1}^{L_k}, (W^{(k)})^\top W^{(k)}=I_{r_k-r_s}, (W^{(k)})^\top V=0\}_{k=1}^{K}\biggr\},
\end{aligned}
\end{equation}
as our parameter space $\Theta=\Theta(p^{(k)}, r^{(k)},L_k,r_s,n,K,\sigma_{\max},\sigma_{\min})$. With the parameter space in hand, we can proceed with the discussion of the minimax lower bound for $[\hat{V},\hat{W}^{(k)}]$, where we conclude with the following theorem.
\begin{theorem}
[Information lower bound]
\label{theorem:minimax}
Assume that each $\cE^{(k)}$ has i.i.d. standard normal distribution entries. If $\{\cT^{(k)}\}_k\in \Theta$, we have the following minimax lower bounds based on the same assumptions as in Theorem \ref{theorem:upper bound}:
\begin{equation}
\label{minimax for V}
\inf_{\hat{V}}\sup_{\{\cT^{(k)}\}_k\in\Theta}\E\left\|\hat{V}\hat{V}^{\top}-VV^{\top}\right\|^2_F
\gtrsim\frac{1}{K\sigma_{\min}^2},
\end{equation}
\begin{equation}
\label{minimax for Wk}
\inf_{\hat{W}^{(k)}}\sup_{\{\cT^{(k)}\}_k\in\Theta}\E\left\|\hat{W}^{(k)}(\hat{W}^{(k)})^{\top}-W^{(k)}(W^{(k)})^{\top}\right\|^2_F
\gtrsim\frac{1}{\sigma_{\min}^2}.
\end{equation}
\end{theorem}

\begin{corollary}
Based on the same conditions in Theorem \ref{theorem:minimax}, we have the following joint minimax lower bound for TPS-PCA estimators $[\hat{V},\hat{W}^{(k)}]$
\begin{equation}
\label{minimax}
\inf_{\hat{V},\hat{W}^{(k)}}\sup_{\{\cT^{(k)}\}_{k}\in\Theta}\E\left[\left\|\hat{V}\hat{V}^{\top}-VV^{\top}\right\|^2_F
+
\left\|\hat{W}^{(k)}\left(\hat{W}^{(k)}\right)^{\top}-W^{(k)}\left(W^{(k)}\right)^{\top}\right\|^2_F\right]
\gtrsim \frac{1}{\sigma_{\min}^2},
\end{equation}
\end{corollary}

\begin{remark}
[One-step minimax optimality]
According to Assumptions \ref{assum:3} and \ref{assum:4}, we can treat $\theta$ as constant and $1/\theta^2\sigma_{\min}^4$  as a higher order term of $1/\sigma_{\min}^2$. Thus, we can simplify \eqref{main rates for Wk} as
\begin{equation*}
\left\|\hat{W}^{(k)}\left(\hat{W}^{(k)}\right)^{\top}-W^{(k)}\left(W^{(k)}\right)^{\top}\right\|^2_F
=O_p\left(\frac{1}{\sigma^2_{\min}}\right).
\end{equation*}
Clearly, the above upper bound matches the lower bound in \eqref{minimax for Wk}
without further assumptions. Moreover,
if we allow $K\lesssim\sigma_{\min}^2$ to hold, we can further simplify \eqref{main rates for V} as 
\begin{equation*}
\left\|\hat{V}\hat{V}^{\top}-VV^{\top}\right\|^2_F
=O_p\left(\frac{1}{K\sigma_{\min}^2}\right),
\end{equation*}
which again matches the lower bounds in \eqref{minimax for V}. The view-number requirement $K\lesssim\sigma_{\min}^2$ is standard in average-projection-type estimators and is imposed in related works; see \cite{fan2019distributed,li2024knowledge}. Note that
under the classical pervasiveness condition in the tensor factor model \citep{barigozzi2026tail}, we have $\sigma_{\min}^2\asymp\prod_{j=1}^{L_k} p_j^{(k)}$. As such, this requirement is particularly mild in our multi-view tensor setting, as $K$ is often treated as a finite number rather than a diverging term. Based on the above discussion,  $[\hat{V},\hat{W}^{(k)}]$ given by TPS-PCA can be regarded as the \textit{one-step view-wise optimally decoupling estimators} under moderate assumptions. In contrast, $[\tilde{A}, \tilde{B}^{(k)}]$ from TP-PCA cannot achieve minimax optimality even from an average perspective.
\end{remark}

\section{Numerical experiments}
\label{sec:4}
We apply TPS-PCA to synthetic datasets and three real-world datasets, including electricity–temperature data \citep{chen2021tensor}, firm-level financial data \citep{fan2022structural}, and human activity recognition data \citep{barshan2014recognizing}. 
These experiments are designed to examine four aspects of our proposed TPS-PCA method. The first dataset evaluates the finite-sample performance of TPS-PCA under a simulated setting. The electricity-temperature example focuses on the  interpretability of the shared-private components given by TPS-PCA. The financial and activity datasets further investigate whether the extracted shared-private structures enhance the effectiveness of downstream portfolio analysis and supervised classification tasks. Taken together, these analyses provide empirical evidence for both the structural interpretability and practical utility of TPS-PCA.

\subsection{Simulation data}
\label{sec:simulation}
In this section, we conduct the following simulation studies. Four methods are considered for comparison in this section, including: $(1)$ TPS-PCA (Algorithm \ref{alg:decouple}), $(2)$  TP-PCA (tensor-version personalized PCA), $(3)$ Vectorized TPS-PCA, and $(4)$ Vectorized TP-PCA (vector-version personalized PCA), where ``Vectorized" refers to applying vectorization to every $\cX_i^{(k)}$ by following the order from mode $1$ to mode $L_k$. Here and thereafter, we directly utilize the \texttt{PerPCA} package provided by \cite{shi2024personalized} to numerically solve TP-PCA and vectorized TP-PCA. The detailed data generating and parameter settings are presented as follows. 

We generate $K=9$ (view number) datasets and set the number of samples in each view to $n = 200$. The shape of each tensor observation in each view is set as follows.
\begin{enumerate}
    \item Views $k=1,2,3$: each $\cX_i^{(k)}$ is a $3D$ tensor $(L_k=3)$ with a size of $30 \times 30 \times 30$, i.e., $p_j^{(k)}=30$. 
    \item Views $k=4,5,6$: each $\cX_i^{(k)}$ is $2D$ matrix $(L_k=2)$ with a size of $100 \times 100$, i.e., $p_j^{(k)}=100$;
    \item Views $k=7,8,9$: each $\cX_i^{(k)}$ is a $1D$ vector $(L_k=1)$ with a length of $1000$, i.e., $p_j^{(k)}=1000$.
\end{enumerate}
The shared components $\{\{\cS_{m_1}^{(k)}\}_{m_1=1}^{r_s}\}_{k=1}^K$ and private components $\{\{\cP_{m_2}^{(k)}\}_{m_2=1}^{r_k-r_s}\}_{k=1}^K$  are all formed with predefined low-dimensional structures: we set the intrinsic Tucker ranks $(r_1^{(k)}, r_2^{(k)}, r_3^{(k)})$ as $(2,2,2)$ for $k=1,2,3$, the matrix rank $(r_1^{(k)}, r_2^{(k)})=(2,2)$ for $k=4,5,6$, and the unit vector norm for $k=7,8,9$. For $[V,W^{(k)}]$, the dimensions of $V$ and $W^{(k)}$ are set to be $r_s = 3$ and $r_k-r_s = 2$, respectively. 
We generate $V \in \mathbb{R}^{n \times 3}$ by imposing QR decomposition on an $n$ by $3$ element-wise standard Gaussian matrix $G$, i.e., $V=\text{QR}(G)$. For each $k$, we further obtain $W^{(k)} \in \mathbb{R}^{n \times 2}$ by $W^{(k)}=\text{QR}((I_n-VV^\top)G^{(k)})$, where $G^{(k)}$ is an $n$ by $2$ element-wise standard Gaussian matrix that is independent of $G$.
The noise level is set as $e_k = 0.1$ for all $k\in[K]$, while the signal level for both shared and private parts is set to be $1$.

For the algorithmic level, we set the dimensions of $\hat{V}(\tilde{A})$, $\{\hat{W}^{(k)}(\tilde{B}^{(k)})\}_k$, and $\{\{\hat{U}_j^{(k)}\}\}_k$ as $3$, $2$, and $10$, respectively, for all $k$ and $j$ where data mode estimators $\{\{\hat{U}_j^{(k)}\}_j\}_k$ are acquired by applying HOOI on each $\cX^{(k)}$. The performance of $\hat{V}$ and $\hat{W}^{(k)}$ is evaluated by the scaled projection metric, i.e., $\|\hat{V}\hat{V}^\top-VV^\top\|_F/2r_s$ and $\|\hat{W}^{(k)}(\hat{W}^{(k)})^\top-W^{(k)}(W^{(k)})^\top\|_F/2(r_k-r_s)$. The global iteration round of the TP-PCA algorithm and the local iteration time for the Stiefel manifold gradient descent are set to be $50$ and $5$, while the learning rate is fixed at $0.1$. We follow the initialization steps provided by \cite{shi2024personalized}, where they initialize the shared subspace estimator by aggregating local PCA directions from all views, with each private subspace randomly placed in the orthogonal complement of the shared initialization. The whole simulation for the four methods is repeated $100$ times.
The detailed simulation results for the four methods are reported in Table \ref{table:subspace_comparison}.
\begin{table}[h]
\caption{Comparison of 4 methods for estimating $\widehat{V}$ and $\{\widehat{W}^{(k)}\}_{k=1,4,7}$.}
\vspace*{0.5em}
\label{table:subspace_comparison}
\renewcommand{\arraystretch}{1.15}
\centering
{\scriptsize
\resizebox{\linewidth}{!}{
\begin{tabular}{ccccc}
\toprule[1.1pt]
\diagbox{Estimators}{Methods} & TPS-PCA & TP-PCA & Vectorized TPS-PCA & Vectorized TP-PCA \\
\midrule
$\scriptstyle \widehat{V}$ & \textbf{0.0628(0.0026)} & 0.0688(0.0020) & 0.1173(0.0048) & 0.1505(0.0031) \\
$\scriptstyle \widehat{W}^{(1)}$ & \textbf{0.2530(0.0520)} & 0.6492(0.0706) & 0.6549(0.0889) & 0.9925(0.0049) \\
$\scriptstyle \widehat{W}^{(4)}$ & \textbf{0.1404(0.0049)} & 0.1405(0.0049) & 0.1961(0.0068) & 0.9868(0.0085) \\
$\scriptstyle \widehat{W}^{(7)}$ & \textbf{0.1456(0.0048)} & 0.2928(0.1790) & 0.1463(0.0048) & 0.2832(0.1645) \\
Computational time & \textbf{0.0395(0.0028)} & 1.0246(0.0243) & 0.0852(0.0035) & 8.4372(0.0678) \\
\addlinespace[0.3em]
\bottomrule[1.1pt]
\end{tabular}
}}

\medskip 
{\fontsize{9pt}{9pt}\selectfont
\noindent\parbox{\linewidth}{
\textit{Note.} We report the averaged estimation error and computational time for each method, with standard deviations shown in parentheses. ``Computational time" implies the  runtime in $1$ numerical simulation loop. The CPU we used for computation is an AMD Ryzen 9 7900X3D 12-Core Processor at 4.40 GHz.
}
}
\end{table}

Due to symmetry, we only report three views, including $k=1$ (tensor view), $k=4$ (matrix view), and $k=7$ (vector view), instead of all $9$ view results. 
Table~\ref{table:subspace_comparison} reports the empirical estimation errors and computational time for four methods. According to the ``TPS-PCA” and ``TP-PCA" columns in Table \ref{table:subspace_comparison}, TPS-PCA achieves smaller estimation errors across all reported subspaces and also requires less computational time per iteration compared to the TP-PCA method. These results support our claims in Sections \ref{sec:2} and \ref{sec:3}, where TPS-PCA acts as an ideal surrogate for TP-PCA. From the ``Vectorized TPS-PCA'' and  ``Vectorized TP-PCA'' columns in Table \ref{table:subspace_comparison}, we can further deduce that explicitly exploiting the tensor structure is beneficial for both the statistical accuracy and computational efficiency of TPS-PCA and TP-PCA. The vectorized methods discard the intrinsic multi-linear structure and consequently exhibit larger estimation errors and substantially higher computational costs, especially for private subspace recovery.

We further investigate the finite-sample performance of the shared dimension selection criterion in \eqref{eq:shared-rank-selection}.
The data generation procedure and simulation configurations are similar to those mentioned above. The only differences are stated as follows. We choose the signal level for both shared and private parts as $20$, while we set the noise level $e_k=e$ across all views and choose $e\in\{1.6,1.8,2,2.2,2.4\}$. When generating $W^{(k)}$, assuming that we have acquired $V\in\mathbb{R}^{n\times3}$, we generate an element-wise standard Gaussian matrix $\tilde{G}\in\mathbb{R}^{n\times 2(K+1)}$, and then orthogonalize the columns of $(I_n-VV^\top)\tilde{G}$ to acquire $[D,H^{(1)},\cdots,H^{(K)}]$. For a given $\alpha$, we generate $W^{(k)}=\sqrt{\cos(\alpha)}D+\sqrt{1-\cos(\alpha)}H^{(k)}$. This construction allows $V^\top W^{(k)}=0$ and $(W^{(k)})^\top W^{(k)}=I_2$ for all $k\in[K]$ and   $(W^{(k)})^\top W^{(\ell)}=\cos(\alpha)I_2$ for any $k\neq\ell$. In fact, $\alpha$ represents the principal angle between any two private subspaces, with smaller values corresponding to stronger cross-view alignment. In this simulation, we choose $\alpha\in\{90^\circ,70^\circ,50^\circ,30^\circ\}$. For the parameters in \eqref{eq:shared-rank-selection}, we further set $R_s=5$ and $\zeta=0.10$. The whole simulation process is repeated $50$ times for each $(\alpha,e)$. The frequencies for correctly estimating and underestimating $r_s$ are reported in Table  \ref{table:rs_selection}, showing that the proposed criterion accurately identifies the shared dimension when the private subspaces are separated and the noise levels are mild.

\begin{table}[h]
\caption{The frequencies of correct estimation and underestimation by \eqref{eq:shared-rank-selection} over $50$ replications.}
\vspace*{0.5em}
\label{table:rs_selection}
\renewcommand{\arraystretch}{1.15}
\centering
{\scriptsize
\resizebox{\linewidth}{!}{
\begin{tabular}{cccccc}
\toprule[1.1pt]
\diagbox{Angle $\alpha$}{Noise $e$} 
& 1.6 & 1.8 & 2.0 & 2.2 & 2.4 \\
\midrule
90$^\circ$ 
& 1.000(0.000) 
& 1.000(0.000) 
& 1.000(0.000) 
& 1.000(0.000) 
& 0.980(0.020) \\

70$^\circ$ 
& 1.000(0.000) 
& 1.000(0.000) 
& 1.000(0.000) 
& 1.000(0.000) 
& 0.960(0.040) \\

50$^\circ$ 
& 1.000(0.000) 
& 1.000(0.000) 
& 1.000(0.000) 
& 1.000(0.000) 
& 0.960(0.040) \\

30$^\circ$ 
& 1.000(0.000) 
& 1.000(0.000) 
& 1.000(0.000) 
& 0.980(0.020) 
& 0.960(0.040) \\
\addlinespace[0.3em]
\bottomrule[1.1pt]
\end{tabular}
}}

\medskip 
{\fontsize{9pt}{9pt}\selectfont
\noindent\parbox{\linewidth}{
\textit{Note.} The values outside and inside parentheses denote the frequencies of correct estimation and underestimation of $r_s$, respectively.
}
}
\end{table}

Following \cite{sergazinov2026spectral}, we further investigate how $\tilde{\Sigma}_0$ separates shared and private subspaces by providing
Figure \ref{fig:rs_selection_spectrum} to illustrate the empirical spectral distribution of $\tilde{\Sigma}_0$ under different $(\alpha,e)$ settings. Notably, the spectral patterns in Figure \ref{fig:rs_selection_spectrum} are consistent with the discussion provided in Corollary \ref{cor:rank consistency}. Specifically, decreasing $\alpha$ allows the private subspaces to be more strongly aligned across views, thereby moving the purple private spectral block toward the green shared block and reducing the population-level separation between them. Increasing the noise level $e$, by contrast, widens the empirical spectral blocks and increases the stochasticity of the estimated eigenvalues. These two facts shown in Figure \ref{fig:rs_selection_spectrum} correspond to the two requirements in Corollary \ref{cor:rank consistency}, which demonstrates that \eqref{eq:shared-rank-selection} can provide an accurate estimate of $r_s$ under mild identification and signal-to-noise conditions.

\begin{figure}[h]
  \begin{minipage}[t]{1\linewidth}
    \centering
    \includegraphics[width=0.98\linewidth]{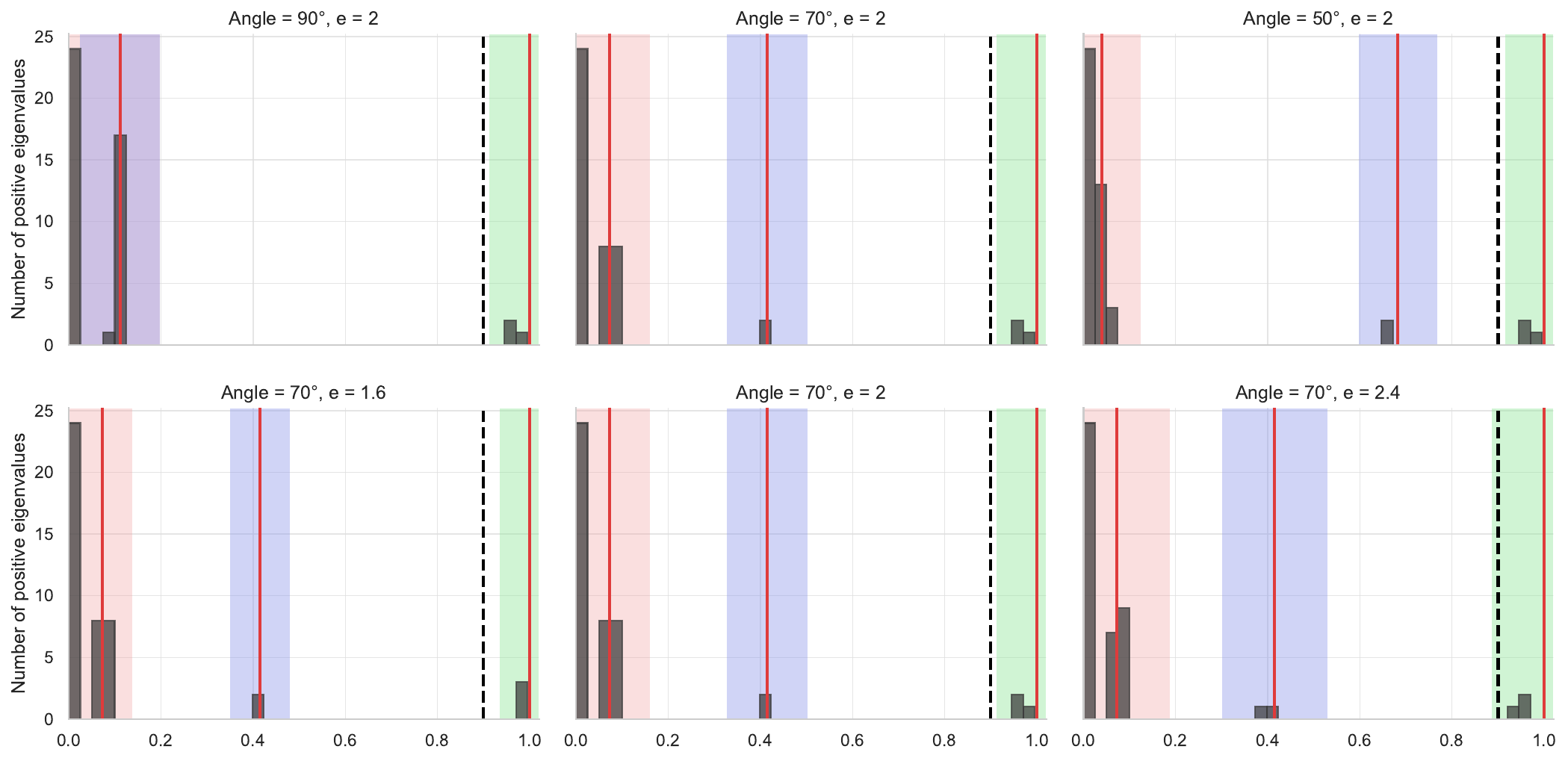}
  \end{minipage}
  \caption{Empirical spectrum of $\widetilde\Sigma_0$ under different $(\alpha,e)$ settings. The black bars give the number of empirical eigenvalues within each interval. The red, purple, and green shaded regions denote the noise, private, and shared perturbed spectral blocks, respectively. The red solid lines mark the corresponding population eigenvalues of $\tilde{\Sigma}_0$, and the black dashed line marks the truncation threshold $1-\zeta=0.90$. The top three sub-figures fix $e=2$ and vary $\alpha$ within $\{90^\circ,70^\circ,50^\circ\}$, while the bottom three sub-figures fix $\alpha=70^\circ$ and change $e$ within $\{1.6,2,2.4\}$. \label{fig:rs_selection_spectrum}}
\end{figure}

\subsection{Electricity-Temperature Dataset}
\label{sec:electricity}
In this section, we consider an electricity-temperature dataset \citep{chen2021tensor}. Specifically, this dataset consists of two views: the Adelaide electricity demand data matrix (view $1$) and the Kent Town temperature data matrix (view $2$). Each view has a matrix size of $48 \times 3556$, where $48$ rows correspond to half-hourly observations within a day, i.e., $48=24\times2$, and $3556$ columns indicate the number of consecutive days, i.e., $3556=508\times 7$, where $508$ represents the number of considered weeks, and $7$ corresponds to the $7$ days within each week. We then reshape each dataset into an order-3 tensor with dimensions $508 \times 48 \times 7$ and align the first mode for the two tensors. We perform dimension reduction along the first mode while preserving the full length of the second and third modes. The dimension parameters $r_k$ can be chosen by applying the existing tensor rank determination approaches \citep{zhang2019cross,han2022rank} for each unfolding matrix, while  $r_s$ can be selected by using \eqref{eq:shared-rank-selection}. In this section, we set $r_1=4$ and $r_s=r_2=3$. The results are reported in  Figure \ref{fig:electricity_TPSPCA}.

\begin{figure}[h]
  \begin{minipage}[t]{1\linewidth}
    \centering
    \includegraphics[width=0.9\linewidth]{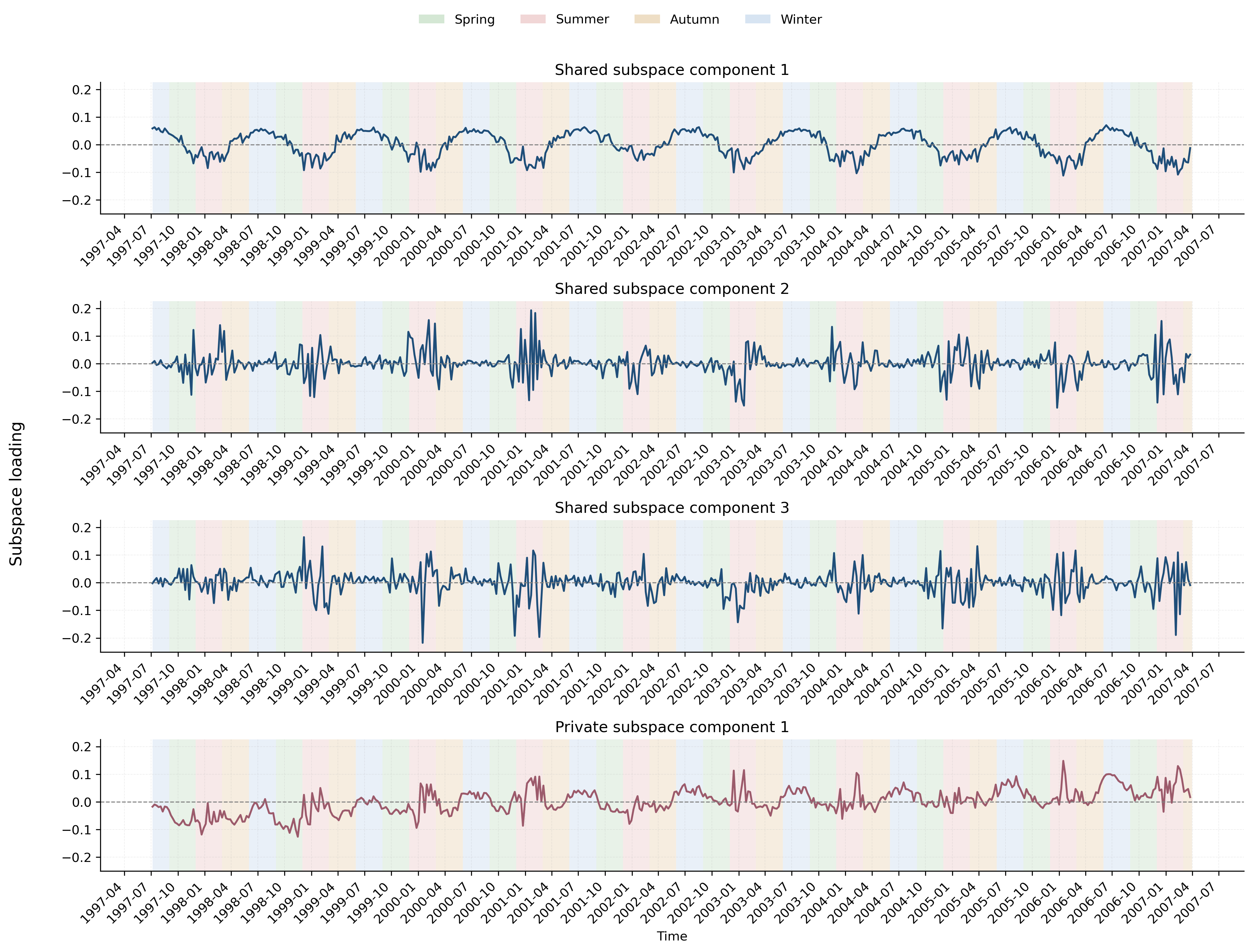}
  \end{minipage}
  \caption{TPS-PCA subspace estimators $[\hat{V},\hat{W}^{(1)}]$ for an electricity-temperature dataset. Four different colors in this figure indicate four seasons. The bottom red line is the private electricity-demand variation that is not explainable by the temperature change. \label{fig:electricity_TPSPCA}}
\end{figure}

Figure \ref{fig:electricity_TPSPCA} shows that the shared components exhibit recurrent annual variation across the observation period. In particular, the first shared component captures a smooth seasonal cycle, whereas the remaining shared components represent sharper temporal variations beyond the dominant seasonal pattern. In contrast, the private electricity-demand component evolves at a substantially lower frequency and gradually shifts from negative values in the early years to positive values in later years. This pattern differs from the recurrent annual variation captured by the shared components and may reflect gradual changes in electricity demand that are not explained by the seasonal variation common to both views. Overall, this example demonstrates that the shared-private decomposition learned by TPS-PCA separates recurrent cross-view variation from slower view-specific dynamics. The corresponding results for TP-PCA are presented in the appendix to save space.

\subsection{Financial Dataset}
\label{sec:financial}

We next consider datasets involving monthly firm-level returns and financial characteristics \citep{fan2022structural}. The processed dataset is collected from the Center for Research in Security Prices (CRSP) and contains $671$ monthly files from February 1965 to December 2020. Each file includes $245$ firms. Each raw monthly panel has $154$ columns, including firm identifiers, group labels, monthly returns, and firm characteristics, etc.
For each month, we use the next-month return vector as the return view and the current-month firm characteristics as the characteristic view. This produces an order-$2$ return tensor $\mathcal{Y}\in\mathbb{R}^{670\times 245},$
and an order-$3$ characteristic tensor $\mathcal{X}\in\mathbb{R}^{670\times 245\times 150}$. See the right panel of Figure \ref{fig:intro_examples_both} for a visualized explanation.

We first align $\cX$ and $\cY$ along the time mode of length $670$ to examine the ability of TPS-PCA to extract latent shared-private features from a temporal perspective. By using Algorithm \ref{alg:decouple}, we can estimate a shared temporal subspace $\widehat{V}$, a return-private temporal subspace $\widehat{W}_Y$, and a characteristic-private temporal subspace
$\widehat{W}_X$. The shared dimension parameter $r_s$ can be selected by utilizing \eqref{eq:shared-rank-selection} again or the method imposed in \cite{chen2025distributed}, while the rest of the Tucker rank parameters can be chosen by utilizing existing methods such as information criteria \citep{bai2002determining} or eigenvalue ratios \citep{lam2012factor,han2024tensor}. We treat each column of $\widehat{V}$ and $\widehat{W}_Y$ as high-frequency time series and plot them in Figure \ref{fig:TPSPCA_financial}.

Figure \ref{fig:TPSPCA_financial} shows that the shared components summarize several distinct forms of common temporal variation between stock returns and firm characteristics. The first component remains relatively stable over the observation period, while the second component exhibits a persistent sign transition from the 1980s to the 1990s, indicating a long-run change in the common temporal structure linking the two views. The remaining shared component is centered near zero and captures higher-frequency fluctuations that are simultaneously present in returns and financial characteristics.

\begin{figure}[h]
  \begin{minipage}[t]{0.9\linewidth}
    \centering
    \includegraphics[width=0.9\linewidth]{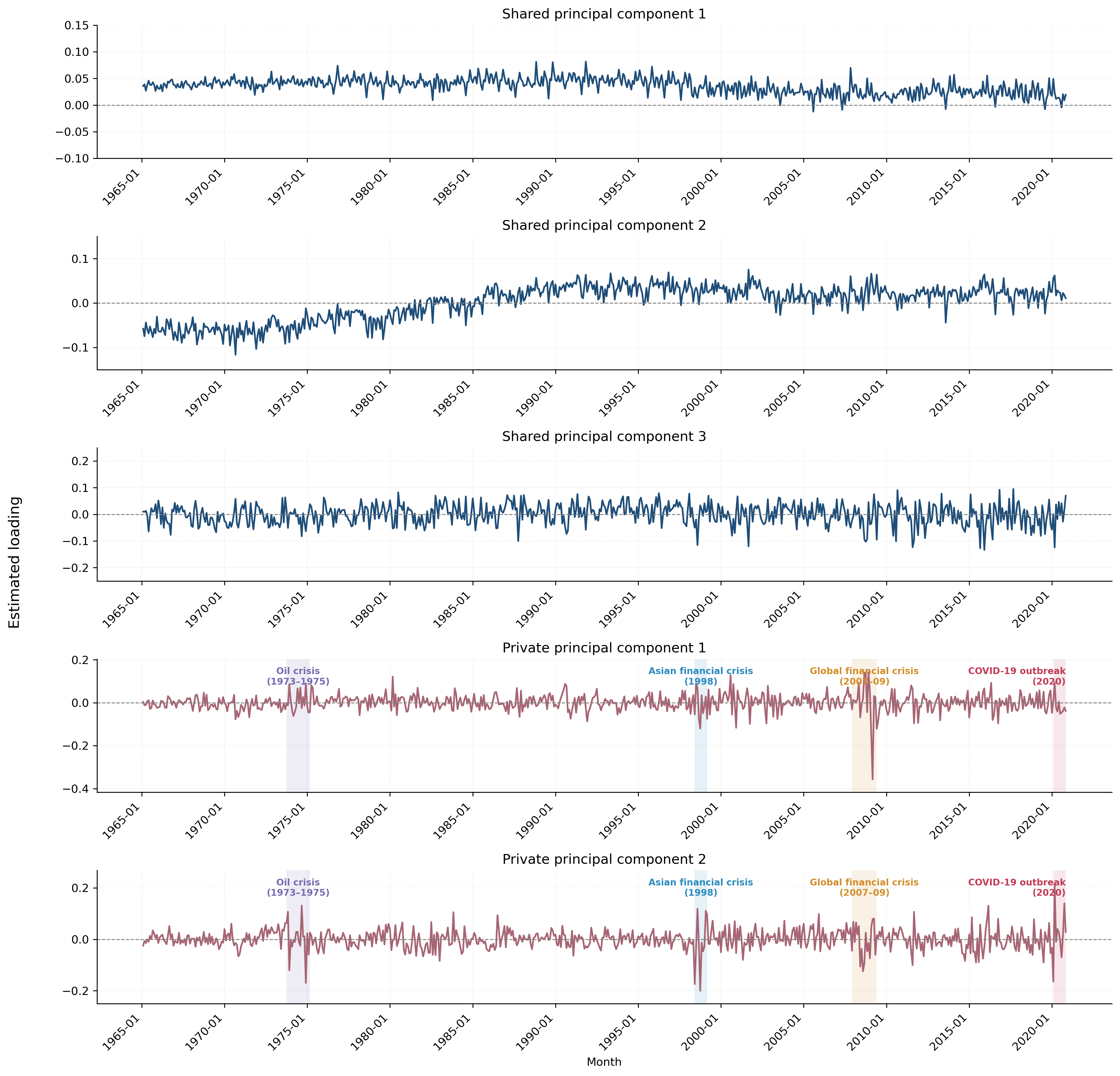}
  \end{minipage}
  \caption{TPS-PCA subspace estimators $[\hat{V},\hat{W}_Y]$
  for a two-view financial dataset containing samples of stock returns and financial characteristics. The bottom red lines are the private return-view components that are not explainable by the firm characteristics.\label{fig:TPSPCA_financial}}
\end{figure}

The private return-view components display temporal behavior that is not explained by the common components. In particular, the return-specific components remain close to zero during relatively stable periods but exhibit pronounced fluctuations around several specific time points, including the mid-1970s, the early 2000s, the global financial crisis period, and the COVID-19 outbreak. These patterns suggest that TPS-PCA isolates return-specific disturbances from the shared structures.

In addition to the identification of shared and private features along the temporal mode, we also use the extracted shared-private subspaces from $\cX$ and $\cY$ to construct out-of-sample portfolios. In particular, we align $\cX$ and $\cY$ along the firm mode with length $245$. By again using Algorithm \ref{alg:decouple}, we can estimate a shared firm-level subspace $\widehat{V}$, a return-private firm-level subspace $\widehat{W}_Y$, and a characteristic-private firm-level subspace
$\widehat{W}_X$. We choose a rolling window length of $240$ months.
At each rolling window, the portfolio weights are obtained from projected returns using specific estimated subspaces. We compare three subspace construction choices: the full TPS-PCA subspace $[\widehat{V},\widehat{W}_X,\widehat{W}_Y]$, the return-related TPS-PCA subspace $[\widehat{V},\widehat{W}_Y]$, and a PCA subspace estimated from $\cY$ only. The transaction fee rate is set to be $0.0002$. 
We report several standard out-of-sample performance measures, including annualized mean return, annualized standard deviation, Sharpe ratio, certainty equivalent return, cumulative return, and maximum drawdown. The results are summarized in Table \ref{table:financial_portfolio}.

\begin{table}[h]
\caption{Out-of-sample portfolio performance under different subspace constructions.}
\vspace*{0.5em}
\label{table:financial_portfolio}
\renewcommand{\arraystretch}{1.15}
\centering
{\small
\resizebox{0.9\linewidth}{!}{
\setlength{\tabcolsep}{24pt}
\begin{tabular}{cccc}
\toprule[1.1pt]
 & $\scriptstyle [\widehat{V},\widehat{W}_X,\widehat{W}_Y]$ & $\scriptstyle [\widehat{V},\widehat{W}_Y]$ & Single PCA \\
\midrule
Mean & \textbf{0.1333} & 0.1329 & 0.1281 \\
Std & 0.1303 & \textbf{0.1295} & 0.1316 \\
SR & 1.0227 & \textbf{1.0258} & 0.9736 \\
CER & \textbf{0.1248} & 0.1245 & 0.1194 \\
CR & \textbf{85.0146} & 84.1756 & 70.3970 \\
MDD (\%) & 41.5630 & \textbf{41.0320} & 42.8115 \\
\addlinespace[0.3em]
\bottomrule[1.1pt]
\end{tabular}
}}

\medskip 
{\fontsize{9pt}{9pt}\selectfont
\noindent\parbox{\linewidth}{
\textit{Note.} Mean and Std denote annualized mean return and annualized standard deviation. SR, CER, CR and MDD represent the Sharpe ratio, certainty equivalent return with risk-aversion parameter equal to $1$, cumulative return, and maximum drawdown, respectively. The out-of-sample portfolios are constructed with a rolling window of $240$ months and transaction fee rate $0.0002$.
}
}
\end{table}

Table \ref{table:financial_portfolio} shows that incorporating shared-private structures improves portfolio performance compared to the return-only PCA benchmark, demonstrating that the joint information in firm characteristics and returns is useful for portfolio construction.
Thus, the above downstream financial applications provide solid evidence that the shared-private decoupling is not only interpretable but also retains useful information relevant to portfolio construction.

\subsection{Activity Recognition Dataset}
\label{sec:Activity}
Finally, we consider the human activity recognition dataset collected by \citet{barshan2014recognizing} and subsequently preprocessed by \citet{wang2018stratified}\footnote{The preprocessed activity recognition dataset is available at \url{https://www.kaggle.com/datasets/jindongwang92/crossposition-activity-recognition}.}. 
The dataset contains observations from eight volunteer subjects, including four females and four males, denoted as $S_1,\cdots,S_8$. Each subject performs $19$ activities, including sitting, standing, ascending, and descending stairs, indexed by $A_1,\cdots,A_{19}$. Activity motion data are collected over a $5$-minute period at a sampling frequency of $12$ times per minute, yielding $5\times12=60$ observations indexed from $n_1$ to $n_{60}$ for each subject-activity pair. During the data collection period, $5$ inertial and magnetic sensor units are attached to the chest, right and left arms, and right and left legs of each subject, respectively. Every sensor unit consists of three sensors: an accelerometer, a gyroscope, and a magnetometer. Each sensor can extract $27$ features. Hence, every sensor unit provides us with  $27\times3=81$ features for each subject $S_i$ while performing the action $A_j$ at time point $n_t$. Combining all extracted features leads us to $5$ tensors, $\{\cX^{(k)}\}_{k=1}^{K=5}$, each of dimension $81\times8\times19\times60$, corresponding to the number of features, subjects, activities, and time points, respectively.

We first conduct an activity pattern recognition task to provide visualization results of the shared and private features for this dataset. We align $\{\cX^{(k)}\}_{k=1}^5$ along the time mode and acquire $\{\{\hat{U}_j^{(k)}\}_{j=1}^{3}\}_{k=1}^5$ by applying HOOI on each $\cX^{(k)}$. Then, we utilize TPS-PCA to estimate $\{\hat{V},\hat{W}^{(k)}\}_{k=1}^5$ and TP-PCA to estimate  $\{\tilde{A},\tilde{B}^{(k)}\}_{k=1}^5$, respectively. Similar to Sections \ref{sec:electricity} and \ref{sec:financial}, $\{\{r_j^{(k)}\}_{j=1}^{3},r_k\}$ can be determined by applying the existing rank selection methods such as \cite{choi2017selecting} and \cite{han2024cp}. The shared dimension $r_s$ can be selected by either applying \eqref{eq:shared-rank-selection} or the method imposed in \cite{li2024knowledge}. In this experiment, we choose $\{\{r_j^{(k)}\}_{j=1}^{3},r_s,r_k-r_s\}$ as $\{2,1,4,1,1\}$ for both TPS-PCA and TP-PCA.  For each $k$, we project $\cX^{(k)}$ by $\{\{\hat{U}_j^{(k)}\}_{j=1}^{2},\hat{V},\hat{W}^{(k)}\}$ and $\{\{\hat{U}_j^{(k)}\}_{j=1}^{2},\tilde{A},\tilde{B}^{(k)}\}$  to obtain $\{\hat{\cS}_{\text{TPS}}^{(k)},
\hat{\cP}_{\text{TPS}}^{(k)}\}$ and $\{\hat{\cS}_{\text{TP}}^{(k)},
\hat{\cP}_{\text{TP}}^{(k)}\}\subseteq\mathbb{R}^{2\times1\times19\times1}$, respectively:
$$
\hat{\cS}_{\text{TPS}}^{(k)}=\cX^{(k)}\times_1(\hat{U}_1^{(k)})^\top\times_2(\hat{U}_2^{(k)})^\top\times_4\hat{V}^\top,\ 
\hat{\cP}_{\text{TPS}}^{(k)}=\cX^{(k)}\times_1(\hat{U}_1^{(k)})^\top\times_2(\hat{U}_2^{(k)})^\top\times_4(\hat{W}^{(k)})^\top,
$$
$$
\hat{\cS}_{\text{TP}}^{(k)}=\cX^{(k)}\times_1(\hat{U}_1^{(k)})^\top\times_2(\hat{U}_2^{(k)})^\top\times_4\tilde{A}^\top,\ 
\hat{\cP}_{\text{TP}}^{(k)}=\cX^{(k)}\times_1(\hat{U}_1^{(k)})^\top\times_2(\hat{U}_2^{(k)})^\top\times_4(\tilde{B}^{(k)})^\top.
$$
We treat $\{\hat{\cS}_{\text{TPS}}^{(k)},
\hat{\cP}_{\text{TPS}}^{(k)}, \hat{\cS}_{\text{TP}}^{(k)},
\hat{\cP}_{\text{TP}}^{(k)}\}_{k=1}^5$ as matrices $\{\hat{S}_{\text{TPS}}^{(k)},
\hat{P}_{\text{TPS}}^{(k)}, \hat{S}_{\text{TP}}^{(k)},
\hat{P}_{\text{TP}}^{(k)}\}_{k=1}^5$ by ignoring the projected subjects mode and time point mode, respectively. Furthermore, we consider the columns of each matrix as points in $\reals^2$ and plot them on $2D$ graphs. For the sake of clear visualization, we only report the graphs corresponding to $\{\cX^{(k)}\}_{k=2}^5$. The activity patterns for TPS-PCA are presented in Figure \ref{fig:TJIVE_activity}, while those for TP-PCA are included in the appendix.

\begin{figure}[h]
  \begin{minipage}[t]{1\linewidth}
    \centering
    \includegraphics[width=0.9\linewidth]{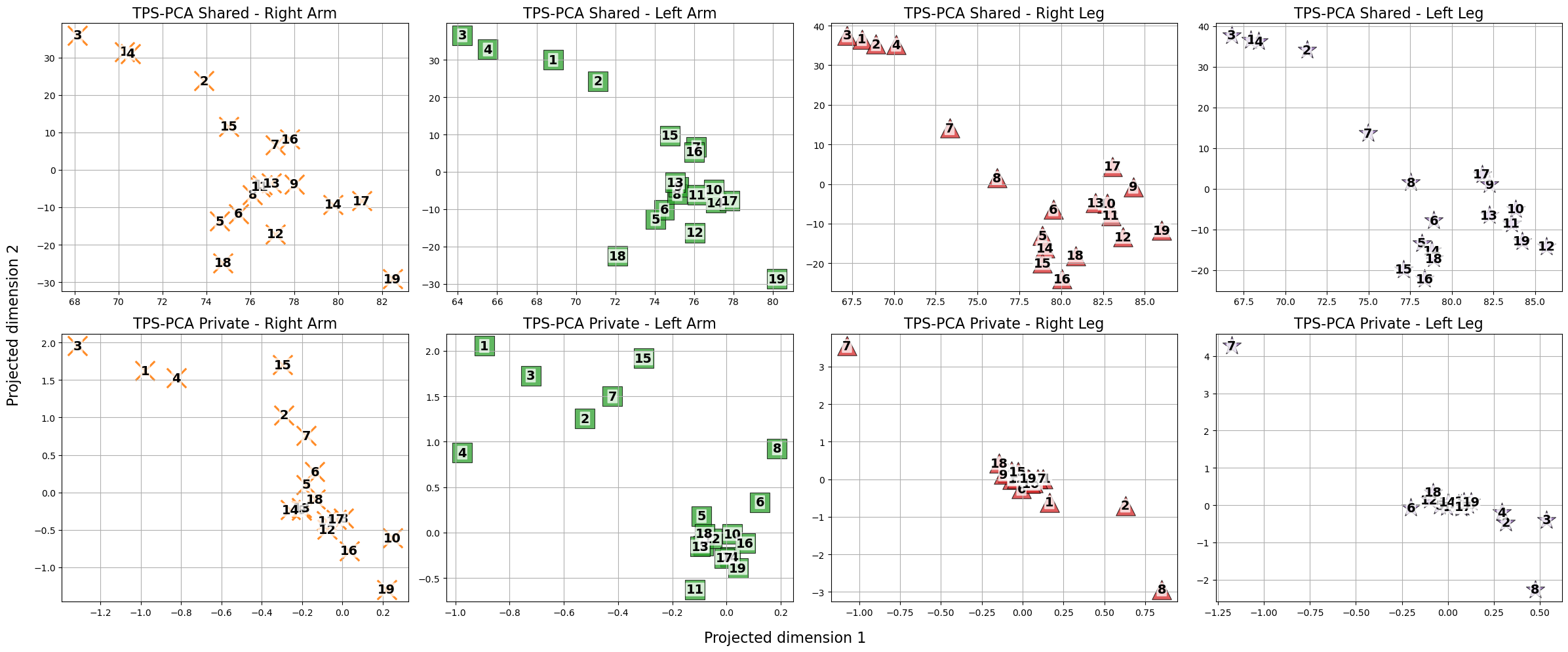}
  \end{minipage}
  \caption{The 2D graphs  produced by TPS-PCA. The numbers $1$ to $19$ represent the corresponding activities of $A_1$ to $A_{19}$, respectively \label{fig:TJIVE_activity}}
\end{figure}

The shared components formed by TPS-PCA, i.e., the first row of Figure \ref{fig:TJIVE_activity}, provide an evident common action pattern. In particular, $A_1$ (sitting), $A_2$ (standing), $A_3$, and $A_4$ (lying on the back and on
the right side) are all separated from the rest of the $15$ moving actions among all $4$ sub-graphs in the first row of Figure \ref{fig:TJIVE_activity}. This means that TPS-PCA provides a well-estimated shared component that distinguishes between moving and static motions. 

For each body part, the private components formed by TPS-PCA, i.e., the second row of Figure \ref{fig:TJIVE_activity}, provide an evident view-specific pattern. In particular, the patterns for the right and left legs of TPS-PCA reflect some similarity, while the right and left arms exhibit two entirely distinct patterns. In practice, leg movements are usually performed by both legs simultaneously. As a result, even the private individual patterns of the left and right legs often exhibit a certain degree of similarity. By contrast, arm movements are more likely to be influenced by a subject’s upper-limb motor habits and, therefore, may display different patterns between the left and right arms. For example, whether a subject is left-handed or right-handed may lead to distinct patterns. 

Based on Figure \ref{fig:TJIVE_activity}, different activities can be classified more accurately by using shared-private structures of $\{\cX^{(k)}\}_{k=1}^5$.
Among all activities, $A_7$ (standing still in an elevator) and $A_8$ (moving around in an elevator) are the most difficult to distinguish \citep{barshan2014recognizing}. We therefore construct two binary classification tasks to examine whether the shared-private structures learned by TPS-PCA improve discrimination for $A_7$  and $A_8$. 

For the first task, we concatenate the observations of the five body parts corresponding to $A_7$ and $A_8$ along the feature mode, yielding
$
\tilde{\cX}\in\mathbb{R}^{405\times8\times2\times60}.
$
Each subject is treated as a view, indexed by $k=1,\cdots,8$, while the first subject $(k=1)$ is selected as the target view, denoted as $\tilde{\cX}^{(1)}\in\mathbb{R}^{405\times 2\times 60}$. For the target-only baseline, we apply HOSVD to the first mode of $\tilde{\cX}^{(1)}$ and estimate the subspace with dimensions of $60$.
Projecting  $\tilde{\cX}^{(1)}$  onto this subspace along the first mode and merging the other two modes gives
$\tilde{X}_{\text{tar}}\in\mathbb{R}^{60\times 120}.$ 
For TPS-PCA, we set  $(r_s,r_1-r_s)=(30,30)$ and apply Algorithm \ref{alg:decouple} on $\{\tilde{\cX}^{(k)}\}_{k=1}^8$ to obtain $[\hat{V},\hat{W}^{(1)}]\in\mathbb{R}^{405\times 60}$. We again project $\tilde{\cX}^{(1)}$ onto $[\hat{V},\hat{W}^{(1)}]$ along the first mode and  obtain 
$\tilde{X}_{\text{TPS}}\in\mathbb{R}^{60\times 120}$.
For both $\tilde{X}_{\text{tar}}$ and $\tilde{X}_{\text{TPS}}$, samples are randomly split into $80\%$ training data and $20\%$ testing data for classifier training and evaluation. We repeat the experiment for $100$ times and report the corresponding results in Table \ref{table:classification_comparison1}.

\begin{table}[htbp]
\caption{The average accuracy of TPS-PCA and target baseline methods for classifying $A_7$ and $A_8$ under the setting of the first classification task.}
\vspace*{0.5em}
\label{table:classification_comparison1}
\centering
\renewcommand{\arraystretch}{1.15}

{\scriptsize
\begin{tabular*}{\linewidth}{@{\extracolsep{\fill}}lccccc}
\toprule[1.1pt]
 & Logistic & LinearSVM & Ridge & Random Forest & kNN \\
\midrule
TPS-PCA & \textbf{0.9567(0.0406)} & \textbf{0.9642(0.0360)} & \textbf{0.9500(0.0423)} & \textbf{0.9454(0.0438)} & \textbf{0.9375(0.0538)} \\
Target & 0.8704(0.0640) & 0.8646(0.0694) & 0.9029(0.0621) & 0.9254(0.0504) & 0.6388(0.0894) \\
\bottomrule[1.1pt]
\end{tabular*}
}

\medskip 
{\fontsize{9pt}{9pt}\selectfont
\noindent\parbox{\linewidth}{
\textit{Note.} The average accuracy for classifying $A_7$ and $A_8$ is reported, with the standard deviation reported in parentheses. ``Logistic'', ``LinearSVM'', ``Ridge'', ``Random Forest'' and ``kNN'' denote $5$ classifiers: logistic regression, linear support vector machine, ridge regression, random forest and k-nearest neighbors, respectively.
}
}
\end{table}

The second task evaluates the classification effectiveness of transferring information across different body parts rather than subjects. We can treat five body parts as views, indexed by $k=1,\cdots,5$, and select the left-leg view $(k=4)$ as the target view. For each view, only the observations from the first subject are retained, resulting in the target-view tensors as
$
\breve{\cX}^{(k)}\in\mathbb{R}^{81\times2\times60},k=1,\cdots,5.
$
The reduced sample matrices from HOSVD and TPS-PCA subspace estimators, i.e., $\breve{X}_{\text{tar}}$ and $\breve{X}_{\text{TPS}}$, are constructed using the same procedure as the first classification task, with the dimension parameters replaced from $(60,30,30)$ to $(20,10,10)$. The corresponding classification results are presented in Table \ref{table:classification_comparison2}.

\begin{table}[htbp]
\caption{The average accuracy of TPS-PCA and target baseline methods for classifying $A_7$ and $A_8$ under the setting of the second classification task.}
\vspace*{0.5em}
\label{table:classification_comparison2}
\centering
\renewcommand{\arraystretch}{1.15}

{\scriptsize
\begin{tabular*}{\linewidth}{@{\extracolsep{\fill}}lccccc}
\toprule[1.1pt]
 & Logistic & LinearSVM & Ridge & Random Forest & kNN \\
\midrule
TPS-PCA & \textbf{0.8833(0.0683)} & 0.8712(0.0613) & \textbf{0.8892(0.0599)} & \textbf{0.9017(0.0554)} & \textbf{0.8504(0.0784)} \\
Target & 0.8504(0.0754) & \textbf{0.8737(0.0694)} & 0.8804(0.0664) & 0.8762(0.0658) & 0.7658(0.0906) \\
\bottomrule[1.1pt]
\end{tabular*}
}
\end{table}

Tables \ref{table:classification_comparison1} and \ref{table:classification_comparison2} show that TPS-PCA outperforms the target-only baseline for all five classifiers in the subject-view task and for four of the five classifiers in the body-part-view task. The largest gains are observed for kNN, for which the average accuracy increases from $0.6388$ to $0.9375$ in the first task and from $0.7658$ to $0.8504$ in the second task. As kNN depends directly on local distances in the learned representation, these pronounced improvements suggest that TPS-PCA produces more discriminative representations for separating $A_7$ and $A_8$. 

The two classification tasks also represent complementary application settings. The first one corresponds to subject-specific activity recognition, where information from other individuals is used to improve representation learning for a target subject while preserving subject-specific variation. This setting can be closely related to personalized rehabilitation monitoring. The second corresponds to body-part-specific sensing, where measurements from multiple sensor locations are used to improve recognition based on a target wearable device, such as a wrist- or ankle-mounted sensor.

In conclusion, the numerical findings above consistently support the central motivation of TPS-PCA, i.e., to separate common and view-specific structures that are interpretable and useful for downstream applications. The benefits of TPS-PCA, therefore, arise not merely from dimensionality reduction but from explicitly modeling the coexistence of shared variation and view-specific heterogeneity.

\section{Discussion}
\label{sec:5}
In this work, we address the problem of identifying common and private structures from multi-view tensor observations. Inspired by personalized PCA, a TPS-PCA approach is proposed to analyze this type of data. The one-step closed-form global optimum and the view-wise minimax decoupling property imply the potential usefulness of TPS-PCA.

For future research, the proposed framework can be extended to generative models, e.g., diffusion models and GANs, for multi-view tensor data, following \citep{chen2023score,guo2026tucker}. In parallel, TPS-PCA can also be integrated with deep neural network methods and pretrained feature extractors to further advance the learning of multi-modal data integration \citep{han2021transformer}.

\section{Acknowledgments}
Kong's work was partially supported by the National Natural Science Foundation (NSF) of China (72342019 and 12431009).
Zhou's work was partially supported by Grant A-8004134-00-00 at the National University of Singapore.

\begin{center}
\textbf{\Large Appendix for ``Optimal Personalized Subspace Learning for Multi-view Tensor Observations"}
\end{center}

The appendix provides additional discussion on selected sections of the main text and presents the detailed proofs of the main theorems, together with auxiliary lemmas and propositions.

\appendix	
\addtocontents{toc}{\protect\setcounter{tocdepth}{2}}
\tableofcontents

\section{Further Details for Real Data Analysis}
\label{supsec:real}
This section provides more results for real data analysis omitted in Section \ref{sec:4}.
\subsection{Electricity-Temperature Datasets}
\label{supsec:electricity}
In addition to TPS-PCA, we further conduct TP-PCA \citep{shi2024personalized} on this electricity-temperature dataset. The data cleaning process and dimension parameters are both identical to those of TPS-PCA. 
The results of TP-PCA are reported in Figure \ref{fig:electricity_PerPCA}.
\begin{figure}[h]
  \begin{minipage}[t]{1\linewidth}
    \centering
    \includegraphics[width=0.8\linewidth]{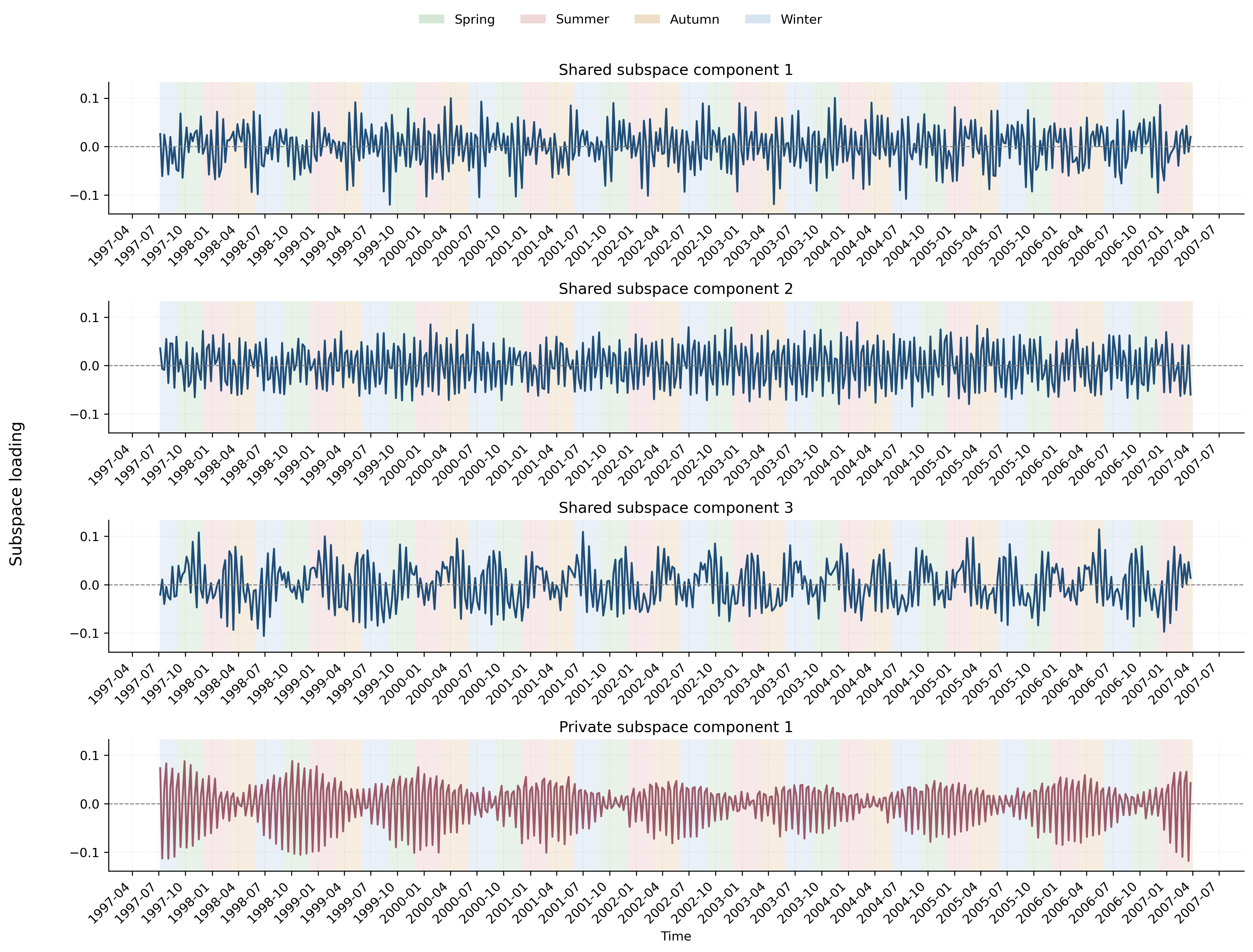}
  \end{minipage}
  \caption{TP-PCA subspace estimators $\{\tilde{A},\tilde{B}^{(1)}\}$ for an electricity-temperature dataset. \label{fig:electricity_PerPCA}}
\end{figure}

Compared to TPS-PCA, the shared and private
components extracted by TP-PCA (Figure \ref{fig:electricity_PerPCA}) exhibit neither clear seasonal nor long-term temporal patterns. Both shared and private components fluctuate rapidly around zero, and their variations are largely dominated by high-frequency noise-like oscillations rather than coherent annual or gradual trends. These comparisons show that TPS-PCA provides a more interpretable shared-private decomposition for the electricity-temperature datasets than TP-PCA.

\subsection{Activity Recognition Datasets}
\label{supsec:activity}
In this section, we conclude the activity patterns for TP-PCA in Figure \ref{fig:PerPCA_activity}, where the round number in Figure \ref{fig:PerPCA_activity} corresponds to the iteration number $R$ in Algorithm $2$ of \cite{shi2024personalized}. Comparing the first row of Figure \ref{fig:PerPCA_activity} to that of Figure \ref{fig:TJIVE_activity}, we can observe similar common action patterns among $4$ different body parts for TPS-PCA and TP-PCA. This means that both methods can deliver well-estimated shared components that distinguish between moving and static motions. However, 
while TPS-PCA achieves accurate pattern identification for each body part with $1$ round of iteration (the second row of Figure \ref{fig:TJIVE_activity}), TP-PCA cannot distinguish any individual pattern in the first few rounds of iteration. This validates the one-step numerical advantage of TPS-PCA compared to TP-PCA. 

\begin{figure}[h]
  \begin{minipage}[t]{1\linewidth}
    \centering
    \includegraphics[width=0.85\linewidth]{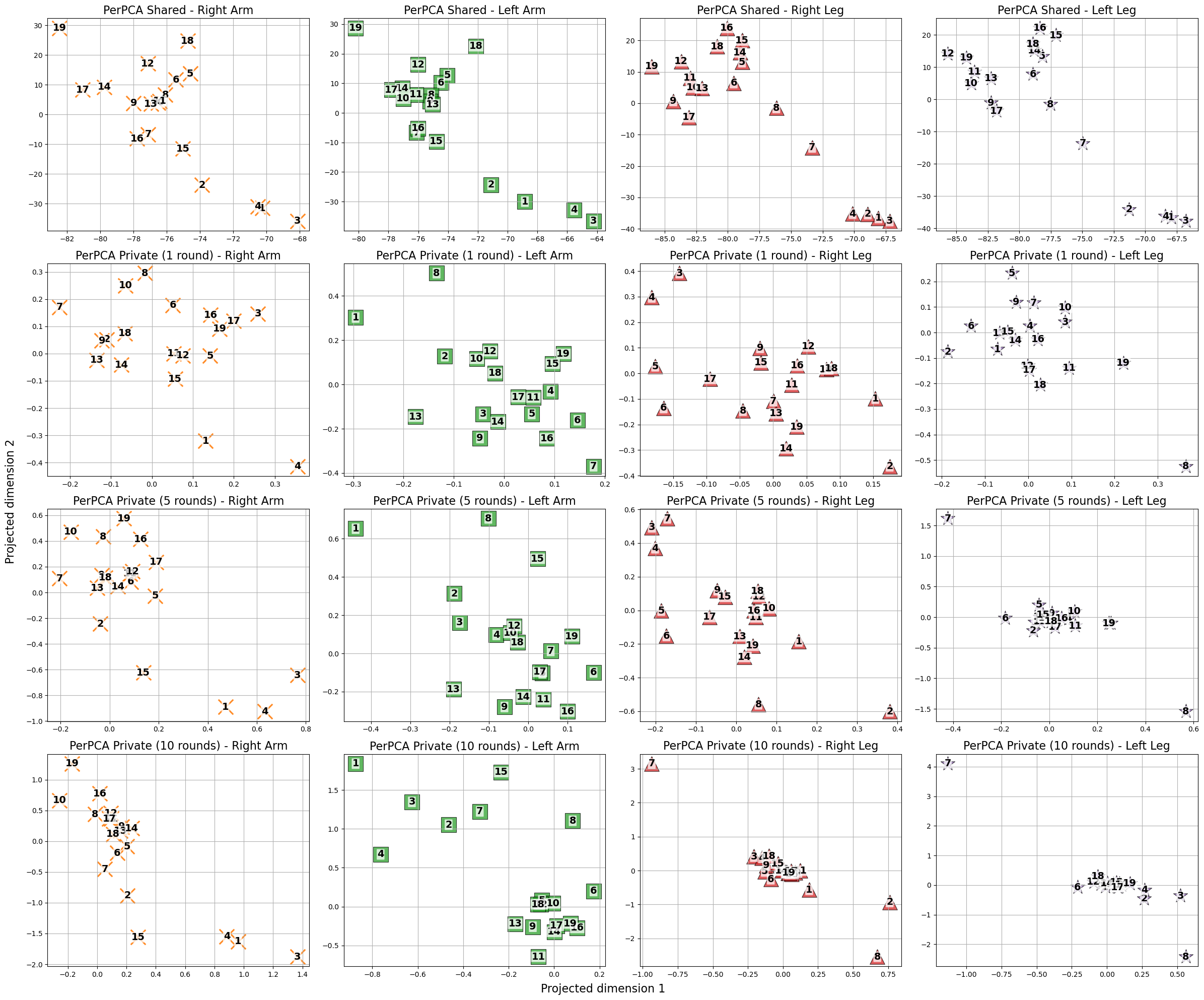}
  \end{minipage}
  \caption{The 2D graphs  produced by TP-PCA. The numbers $1$ to $19$ represent the corresponding activities of $A_1$ to $A_{19}$, respectively \label{fig:PerPCA_activity}}
\end{figure}

\section{Proof of the Theoretical Results}
In this section, we present the proof of the theoretical results in this work. Before we dive into the details, some basic notations are required to be introduced first. Define $\lambda_{\max}$ and $\lambda_{\min}$ as the maximum and minimum singular values of $\cM_{L_k+1}(\cC^{(k)})$ across all $k\in[K]$. Based on the definition of $\sigma_{\max}$ and $\sigma_{\min}$, we know that $\lambda_{\max}=\sqrt{n}\sigma_{\max}$ while $\lambda_{\min}=\sqrt{n}\sigma_{\min}$. For each $k\in[K]$, we 
define $X^{(k)}=\cM_{L_k+1}(\cX^{(k)})$. According Lemma 4 of \cite{zhang2018tensor}, we have
\begin{equation}
\label{tildeMk}
\tilde{M}^{(k)}=\cM_{L_k+1}\left(\cX^{(k)}\times_1(\hat{U}^{(k)}_{1})^{\top}\cdots \times_{L_k}(\hat{U}^{(k)}_{L_k})^{\top} \right)=X^{(k)}\left(\hat{U}^{(k)}_{1}\otimes\cdots\otimes \hat{U}^{(k)}_{L_k}\right),
\end{equation}
Moreover, we denote $T^{(k)}=\cM_{L_k+1}(\cT^{(k)})$ and $Z^{(k)}=\cM_{L_k+1}(\cE^{(k)})$, then we also have
\begin{equation}
\label{Mk}
M^{(k)}=\cM_{L_k+1}\left(\cT^{(k)}\times_1(\hat{U}^{(k)}_{1})^{\top}\cdots \times_{L_k}(\hat{U}^{(k)}_{L_k})^{\top} \right)=T^{(k)}\left(\hat{U}^{(k)}_{1}\otimes\cdots\otimes \hat{U}^{(k)}_{L_k}\right),
\end{equation}
\begin{equation}
\label{Ek}
E^{(k)}=\cM_{L_k+1}\left(\cE^{(k)}\times_1(\hat{U}^{(k)}_{1})^{\top}\cdots \times_{L_k}(\hat{U}^{(k)}_{L_k})^{\top} \right)=Z^{(k)}\left(\hat{U}^{(k)}_{1}\otimes\cdots\otimes \hat{U}^{(k)}_{L_k}\right).
\end{equation}
According to \eqref{tildeMk}-\eqref{Ek}, $E^{(k)}=\tilde{M}^{(k)}-M^{(k)}$.
If we further denote $U_{\otimes}^{(k)}:= U^{(k)}_{1}\otimes\cdots\otimes U^{(k)}_{L_k}$ and $\hat{U}_{\otimes}^{(k)}:=\hat{U}^{(k)}_{1}\otimes\cdots\otimes \hat{U}^{(k)}_{L_k}$. Then we can rewrite \eqref{tildeMk}-\eqref{Ek} as:
$$
\tilde{M}^{(k)}=X^{(k)}\hat{U}_{\otimes}^{(k)},\quad
M^{(k)}=T^{(k)}\hat{U}_{\otimes}^{(k)},\quad
E^{(k)}=Z^{(k)}\hat{U}_{\otimes}^{(k)}.
$$
Based on the above notations, we can proceed to prove the theorems in the main article.
\subsection{Proof of Theorem \ref{theorem:go}}
For any given $A$ such that $A^{\top}A= I_{r_s}$, if $r_k=r_s$, then no private subspace exists for the $k$-th view, and the problem is trivial. If $r_s<r_k$, then the subproblem of $B^{(k)}$ is essentially
\begin{equation}
\label{eq:k-view-private}
\begin{aligned}    
\hat{B}^{(k)}(A)=    &\argmax_{B^{(k)}} \tr\left[ (B^{(k)})^{\top}\tilde{V}^{(k)}(\tilde{V}^{(k)})^{\top} B^{(k)}\right],\\
&\text{subject to }\; (B^{(k)})^{\top}B^{(k)}=I_{r_k-r_s},\; A^{\top}B^{(k)}=0.
\end{aligned}
\end{equation}
The solution to this subproblem is clearly the leading $r_k-r_s$ left singular subspace of $(I_n-AA^{\top})\tilde{V}^{(k)}$, as we have 
$$\hat{B}^{(k)}(A)=   \argmax_{(B^{(k)})^{\top}B^{(k)}=I_{r_k-r_s}} \tr\left[ (B^{(k)})^{\top}(I_n-AA^{\top})\tilde{V}^{(k)}(\tilde{V}^{(k)})^{\top} (I_n-AA^{\top})B^{(k)}\right].$$
Note that as $(I_n-AA^{\top})$ is a projection matrix onto a $n-r_s$ dimensional subspace, say $\span(A^{\perp})$, while $\tilde{V}^{(k)}(\tilde{V}^{(k)})^{\top}$ is a projection matrix onto a $r_k$ dimensional subspace $\span(\tilde{V}^{(k)})$, such that $r_s< r_k$. According to Ky Fan's maximum principle \citep{vu2013minimax,vu2013fantope}, we know that the maximum of \eqref{eq:k-view-private} can be achieved by 
\begin{equation*}
\begin{aligned}
\sum_{i=1}^{r_k-r_s}
\lambda_i\left((I_n-AA^{\top})\tilde{V}^{(k)}(\tilde{V}^{(k)})^{\top} (I_n-AA^{\top})\right).
\end{aligned}
\end{equation*}
Further, the dimension formula for subspaces gives us $$\text{dim}\left[\span(A^{\perp})\cap\span(\tilde{V}^{(k)})\right]\geq \text{dim}\left[\span(A^{\perp})\right]+\text{dim}\left[\span(\tilde{V}^{(k)})\right]-n=r_k-r_s.$$
Therefore, we can choose orthonormal vectors
$
\{x_1,\cdots,x_{r_k-r_s}\}\subseteq
\operatorname{span}(A)^\perp
\cap
\operatorname{span}(\widetilde V^{(k)}).
$
For each \(j=1,\cdots,r_k-r_s\), since
\(x_j\in \operatorname{span}(A)^\perp\), we have
$
(I_n-AA^\top)x_j=x_j.
$
Moreover, since \(x_j\in \operatorname{span}(\widetilde V^{(k)})\), we also have
$
\widetilde V^{(k)}(\widetilde V^{(k)})^\top x_j=x_j.
$
Hence,
\[
\begin{aligned}
&(I_n-AA^\top)\widetilde V^{(k)}
(\widetilde V^{(k)})^\top
(I_n-AA^\top)x_j  
=
(I_n-AA^\top)\widetilde V^{(k)}
(\widetilde V^{(k)})^\top x_j  
=
(I_n-AA^\top)x_j  
=
x_j.
\end{aligned}
\]
Therefore,
$
(I_n-AA^\top)\widetilde V^{(k)}
(\widetilde V^{(k)})^\top
(I_n-AA^\top)
$
has at least \(r_k-r_s\) eigenvalues that equal \(1\), which means that $(I_n-AA^\top)\widetilde V^{(k)}$ owns at least \(r_k-r_s\) singular values that equal \(1\). As $\|(I_n-AA^\top)\widetilde V^{(k)}\|_2\leq1$, the top $r_k-r_s$ singular values of $(I_n-AA^{\top})\tilde{V}^{(k)}$ (which are the cosines of the principal angles) are exactly $1$. Thus, for all $A$ such that $A^{\top}A= I_{r_s}$ and $A^{\top}B^{(k)}=0$, we have 
$$\tr\left[ \hat{B}^{(k)}(A)^{\top}\tilde{V}^{(k)}(\tilde{V}^{(k)})^{\top} \hat{B}^{(k)}(A)\right]=r_k-r_s.$$
That is to say, given any shared $A$, the maximum of the $k$-th personalized subproblem remains fixed and can be discarded from the joint optimization in \eqref{eq:personalized_subspace2}. The proof is completed by noticing that the leading $r_s$ subspace of \eqref{eq:apm} solves the shared problem in \eqref{eq:personalized_subspace2} after discarding the personalized subproblems.

\subsection{Proof of Proposition \ref{prop:1}}
The proof mainly follows from Theorem 1 of \cite{shi2024personalized}, by carefully controlling the signal and noise levels of the matricized projected tensor $\tilde{M}^{(k)}\in \reals^{n\times (r^{(k)}_1\cdots r^{(k)}_{L_k})}$.
Recall from \eqref{Mk} that the left singular subspace of $M^{(k)}$ is precisely $(V,W^{(k)})$ by definition. According to Lemma \ref{lem:signal} and Proposition \ref{prop:curvature}, we have $\sigma_{r_k}(M^{(k)}(M^{(k)})^\top)=\sigma_{r_k}^2(M^{(k)})\gtrsim\lambda_{\min}^2$ and
\begin{equation}
\begin{aligned}
    &\left\langle M^{(k)}\left(M^{(k)}\right)^{\top} ,VV^{\top}+W^{(k)}\left(W^{(k)}\right)^{\top}-\tilde{A}\tilde{A}^{\top}-\tilde{B}^{(k)}(\tilde{B}^{(k)})^{\top}\right\rangle\\
    &\quad \gtrsim \lambda_{\min}^2 \left\| VV^{\top}+W^{(k)}\left(W^{(k)}\right)^{\top}-\tilde{A}\tilde{A}^{\top}-\tilde{B}^{(k)}(\tilde{B}^{(k)})^{\top}\right\|^2_F.
\end{aligned}
\end{equation}
Summing both sides for $k\in[K]$, we have
\begin{equation}\label{eq:opt1}
\begin{aligned}
    &\sum_{k=1}^{K}\left\langle M^{(k)}\left(M^{(k)}\right)^{\top} ,VV^{\top}+W^{(k)}\left(W^{(k)}\right)^{\top}-\tilde{A}\tilde{A}^{\top}-\tilde{B}^{(k)}(\tilde{B}^{(k)})^{\top}\right\rangle\\
    &\quad \gtrsim \lambda_{\min}^2\sum_{k=1}^{K} \left\| VV^{\top}+W^{(k)}\left(W^{(k)}\right)^{\top}-\tilde{A}\tilde{A}^{\top}-\tilde{B}^{(k)}(\tilde{B}^{(k)})^{\top}\right\|^2_F.
\end{aligned}
\end{equation}
Meanwhile, as $(\tilde{A},\tilde{B}^{(k)})$ is the solution of the optimization problem \eqref{eq:optimization}, we have directly
\begin{equation}\label{eq:opt2}
   \sum_{k=1}^{K}\left\langle \tilde{M}^{(k)}\left(\tilde{M}^{(k)}\right)^{\top} ,\tilde{A}\tilde{A}^{\top}+\tilde{B}^{(k)}(\tilde{B}^{(k)})^{\top}-VV^{\top}-W^{(k)}\left(W^{(k)}\right)^{\top}\right\rangle\geq 0. 
\end{equation}
Combining \eqref{eq:opt1} and \eqref{eq:opt2}, we shall have
\begin{equation}\label{eq:opt3}
\begin{aligned}
   &\sum_{k=1}^{K}\left\langle \tilde{M}^{(k)}\left(\tilde{M}^{(k)}\right)^{\top} -M^{(k)}\left(M^{(k)}\right)^{\top},\tilde{A}\tilde{A}^{\top}+\tilde{B}^{(k)}(\tilde{B}^{(k)})^{\top}-VV^{\top}-W^{(k)}\left(W^{(k)}\right)^{\top}\right\rangle\\
   &\quad \gtrsim \lambda_{\min}^2\sum_{k=1}^{K} \left\| \tilde{A}\tilde{A}^{\top}+\tilde{B}^{(k)}(\tilde{B}^{(k)})^{\top}-VV^{\top}-W^{(k)}\left(W^{(k)}\right)^{\top}\right\|^2_F.
   \end{aligned}
\end{equation}

For the left-hand side, we have, by the Cauchy-Schwarz inequality, that 
\begin{equation*}
\begin{aligned}
   &\left\langle \tilde{M}^{(k)}\left(\tilde{M}^{(k)}\right)^{\top} -M^{(k)}\left(M^{(k)}\right)^{\top},\tilde{A}\tilde{A}^{\top}+\tilde{B}^{(k)}(\tilde{B}^{(k)})^{\top}-VV^{\top}-W^{(k)}\left(W^{(k)}\right)^{\top}\right\rangle\\
   &\quad \leq \left\|\tilde{M}^{(k)}\left(\tilde{M}^{(k)}\right)^{\top} -M^{(k)}\left(M^{(k)}\right)^{\top}\right\|_F\left\|\tilde{A}\tilde{A}^{\top}+\tilde{B}^{(k)}(\tilde{B}^{(k)})^{\top}-VV^{\top}-W^{(k)}\left(W^{(k)}\right)^{\top}\right\|_F.
   \end{aligned}
\end{equation*}
By another Cauchy-Schwarz inequality, we have
\begin{equation*}
\begin{aligned}   &\sqrt{\sum_{k=1}^{K}\left\|\tilde{M}^{(k)}\left(\tilde{M}^{(k)}\right)^{\top} -M^{(k)}\left(M^{(k)}\right)^{\top}\right\|^2_F}\sqrt{\sum_{k=1}^{K} \left\| \tilde{A}\tilde{A}^{\top}+\tilde{B}^{(k)}(\tilde{B}^{(k)})^{\top}-VV^{\top}-W^{(k)}\left(W^{(k)}\right)^{\top}\right\|^2_F}\\
&\quad \geq \sum_{k=1}^K
\left\|\tilde{M}^{(k)}\left(\tilde{M}^{(k)}\right)^{\top} -M^{(k)}\left(M^{(k)}\right)^{\top}\right\|_F\left\|\tilde{A}\tilde{A}^{\top}+\tilde{B}^{(k)}(\tilde{B}^{(k)})^{\top}-VV^{\top}-W^{(k)}\left(W^{(k)}\right)^{\top}\right\|_F\\
& \quad \geq \sum_{k=1}^K\left\langle \tilde{M}^{(k)}\left(\tilde{M}^{(k)}\right)^{\top} -M^{(k)}\left(M^{(k)}\right)^{\top},\tilde{A}\tilde{A}^{\top}+\tilde{B}^{(k)}(\tilde{B}^{(k)})^{\top}-VV^{\top}-W^{(k)}\left(W^{(k)}\right)^{\top}\right\rangle\\
   &\quad \overset{\eqref{eq:opt3}}{\gtrsim} \lambda_{\min}^2\sum_{k=1}^{K} \left\| \tilde{A}\tilde{A}^{\top}+\tilde{B}^{(k)}(\tilde{B}^{(k)})^{\top}-VV^{\top}-W^{(k)}\left(W^{(k)}\right)^{\top}\right\|^2_F.
   \end{aligned}
\end{equation*}
As a result, based on $E^{(k)}=\tilde{M}^{(k)}-M^{(k)}$ and Lemmas \ref{lem:signal}, \ref{lem:noise}, we have
\begin{equation*}
\begin{aligned}   &\quad \frac{1}{K}\sum_{k=1}^{K} \left\| \tilde{A}\tilde{A}^{\top}+\tilde{B}^{(k)}(\tilde{B}^{(k)})^{\top}-VV^{\top}-W^{(k)}\left(W^{(k)}\right)^{\top}\right\|^2_F\\
   &\lesssim \frac{1}{K\lambda^4_{\min}} \sum_{k=1}^{K}\left\|\tilde{M}^{(k)}\left(\tilde{M}^{(k)}\right)^{\top} -M^{(k)}\left(M^{(k)}\right)^{\top}\right\|^2_F\\
   &\leq \frac{1}{K\lambda^4_{\min}} \sum_{k=1}^{K}\left[2\left\|M^{(k)}\left(E^{(k)}\right)^{\top}\right\|_F
   +
   \left\|E^{(k)}\left(E^{(k)}\right)^{\top}\right\|_F\right]^2\\
   &\leq \frac{1}{K\lambda^4_{\min}} \sum_{k=1}^{K}\left[2\sqrt{2r_k}\left\|M^{(k)}\right\|_{\op}\left\|E^{(k)}\right\|_{\op}
   +
   2\sqrt{\text{rank}(E^{(k)})}\left\|E^{(k)}\right\|_{\op}^2\right]^2.
   \end{aligned}
\end{equation*}
Note that $E^{(k)}=Z^{(k)}\hat{U}_{\otimes}^{(k)}$. Thus, $\text{rank}(E^{(k)})\leq\min\{\text{rank}(Z^{(k)}), \text{rank}(\hat{U}_{\otimes}^{(k)})\}\leq\text{rank}(\hat{U}_{\otimes}^{(k)})$. Further, as $\hat{U}_{\otimes}^{(k)}:=\hat{U}^{(k)}_{1}\otimes\cdots\otimes \hat{U}^{(k)}_{L_k}$, we have:
$$
\text{rank}(\hat{U}_{\otimes}^{(k)})
=
\text{rank}(\hat{U}^{(k)}_{1})\cdots\text{rank}(\hat{U}^{(k)}_{L_k})=\prod_{j=1}^{L_k} r_j^{(k)}.
$$
According to Assumption \ref{assum:4}, we know that $\text{rank}(\hat{U}_{\otimes}^{(k)})=O(1)$, which further gives us:
\begin{equation*}
    \begin{aligned}
    &\quad \frac{1}{K}\sum_{k=1}^{K} \left\| \tilde{A}\tilde{A}^{\top}+\tilde{B}^{(k)}(\tilde{B}^{(k)})^{\top}-VV^{\top}-W^{(k)}\left(W^{(k)}\right)^{\top}\right\|^2_F\\
   &\lesssim \frac{1}{K\lambda^4_{\min}} \sum_{k=1}^{K}\left[2\sqrt{2r_k}\left\|M^{(k)}\right\|_{\op}\left\|E^{(k)}\right\|_{\op}
   +
   2\sqrt{\text{rank}(\hat{U}_{\otimes}^{(k)})}\left\|E^{(k)}\right\|_{\op}^2\right]^2\\
   &\overset{\text{Lemma \ref{lem:noise}}}{\lesssim}\frac{1}{K\lambda^4_{\min}} \sum_{k=1}^{K}\left[\lambda_{\max}\sqrt{n}+n\right]^2
   \overset{\text{Assumption \ref{assum:4}}}{\lesssim}\frac{\lambda_{\max}^2n}{\lambda_{\min}^4},
   \end{aligned}
\end{equation*}
with probability tending to $1$.
Finally, by the insightful Lemma 13 of \cite{shi2024personalized}, if Assumption \ref{assum:3} holds, we have
\begin{equation*}
\begin{aligned}   &\sum_{k=1}^{K} \left\| \tilde{A}\tilde{A}^{\top}+\tilde{B}^{(k)}(\tilde{B}^{(k)})^{\top}-VV^{\top}-W^{(k)}\left(W^{(k)}\right)^{\top}\right\|^2_F\\
   &\quad \gtrsim \theta \left(K\left\|\tilde{A}\tilde{A}^{\top}-VV^{\top}\right\|^2_F+\sum_{k=1}^K\left\|\tilde{B}^{(k)}(\tilde{B}^{(k)})^{\top}-W^{(k)}\left(W^{(k)}\right)^{\top}\right\|^2_F\right).
   \end{aligned}.
\end{equation*}
The proof is complete by noticing that
\begin{equation*}
\begin{aligned}
&\quad\left\|\tilde{A}\tilde{A}^{\top}-VV^{\top}\right\|^2_F+\frac{1}{K}\sum_{k=1}^K\left\|\tilde{B}^{(k)}(\tilde{B}^{(k)})^{\top}-W^{(k)}\left(W^{(k)}\right)^{\top}\right\|^2_F\\
&\lesssim 
\frac{1}{\theta K}\sum_{k=1}^{K} \left\| \tilde{A}\tilde{A}^{\top}+\tilde{B}^{(k)}(\tilde{B}^{(k)})^{\top}-VV^{\top}-W^{(k)}\left(W^{(k)}\right)^{\top}\right\|^2_F
=O_p\left(\frac{\lambda_{\max}^2n}{\theta\lambda_{\min}^4}\right),
\end{aligned}
\end{equation*}
and $\lambda_{\min}=\sqrt{n}\sigma_{\min}$.
\subsection{Proof of Theorem \ref{theorem:upper bound}}
We first impose some basic notations for simplicity considerations. Denote $P_V=VV^\top$, $\hat{P}_V=\hat{V}\hat{V}^\top$. Further, we let $P_W^{(k)}=W^{(k)}(W^{(k)})^\top$ and $\hat{P}_W^{(k)}=\hat{W}^{(k)}(\hat{W}^{(k)})^\top$. Define $V^{(k)}=[V,W^{(k)}]$. Recall that $X^{(k)}=\cM_{L_k+1}(\cX^{(k)})$, $T^{(k)}=\cM_{L_k+1}(\cT^{(k)})$, then we denote $P^{(k)}=P_V+P_W^{(k)}$ and $\tilde{P}^{(k)}=\tilde{V}^{(k)}(\tilde{V}^{(k)})^\top$ as the projection matrices corresponding to the left $r_k$-dimensional singular subspace of $M^{(k)}$ and $\tilde{M}^{(k)}$, respectively, which gives us the subspace error of $X^{(k)}$ and $\Delta^{(k)}=\tilde{P}^{(k)}-P^{(k)}$. For any projection matrix $P$, we denote its orthogonal complement as $P^\perp$. We divide this section into two parts to prove Theorem \ref{theorem:upper bound}.

\subsubsection{Proof of the first claim in Theorem \ref{theorem:upper bound}}
\label{supp:upper bound of V}
\textbf{\textit{Step 1:}}
Recall that $\Sigma_0$ is defined as
$$
\Sigma_0:=P_V + P_V^\perp\left(\frac{1}{K}\sum_{k=1}^{K}P_W^{(k)}\right)P_V^\perp
\overset{V^\top W^{(k)}=0}{=}
P_V + \frac{1}{K}\sum_{k=1}^{K}P_W^{(k)}
.
$$
According to Assumption \ref{assum:3}, we know that the leading $r_s$-dimensional eigenspace of $\Sigma_0$ is exactly  $P_V$. Then, by utilizing Proposition \ref{prop:curvature}, we have:
\begin{equation}
\label{eq:the3_curv_results}
\left\|P_V-\hat{P}_V\right\|_F^2\leq\frac{2}{\delta_{r_s}(\Sigma_0)}\langle\Sigma_0,P_V-\hat{P}_V\rangle,
\end{equation}
where $\delta_{r_s}(\Sigma_0)=\lambda_{r_s}(\Sigma_0)-\lambda_{r_s+1}(\Sigma_0)$ is the $r_s$-th eigenvalue gap of $\Sigma_0$.
Based on Algorithm \ref{alg:decouple}, we know that $\hat{P}_V$ is the projection matrix that corresponds to the leading $r_s$-dimensional eigenspace of $\tilde{\Sigma}_0$. Accordingly, we have:
$$
\langle\tilde{\Sigma}_0,\hat{P}_V\rangle\geq\langle\tilde{\Sigma}_0,P_V\rangle,
$$
which implies
$$
\langle\tilde{\Sigma}_0,\hat{P}_V-P_V\rangle\geq0,\ \text{and}\ \langle-\tilde{\Sigma}_0,P_V-\hat{P}_V\rangle\geq0.
$$
Plugging the above result back into \eqref{eq:the3_curv_results} gives us:
$$
\left\|P_V-\hat{P}_V\right\|_F^2\leq
\frac{2}{\delta_{r_s}(\Sigma_0)}\langle\Sigma_0,P_V-\hat{P}_V\rangle
\leq
\frac{2}{\delta_{r_s}(\Sigma_0)}\langle\Sigma_0-\tilde{\Sigma}_0,P_V-\hat{P}_V\rangle.
$$
By utilizing the Cauchy-Schwarz inequality, we have:
$$
\left\|P_V-\hat{P}_V\right\|_F^2
\leq
\frac{2}{\delta_{r_s}(\Sigma_0)}\langle\tilde{\Sigma}_0-\Sigma_0,\hat{P}_V-P_V\rangle
\leq
\frac{2}{\delta_{r_s}(\Sigma_0)}
\left\|\tilde{\Sigma}_0-\Sigma_0\right\|_F
\left\|\hat{P}_V-P_V\right\|_F,
$$
which gives us:
$$
\left\|P_V-\hat{P}_V\right\|_F
\leq
\frac{2}{\delta_{r_s}(\Sigma_0)}
\left\|\tilde{\Sigma}_0-\Sigma_0\right\|_F.
$$
For $\delta_{r_s}(\Sigma_0)$, we can utilize Proposition \ref{prop:weyl} to acquire:
\begin{equation*}
\begin{aligned}
&\quad\delta_{r_s}(\Sigma_0)
=
\lambda_{r_s}\left(P_V + P_V^\perp\left(\frac{1}{K}\sum_{k=1}^{K}P_W^{(k)}\right)P_V^\perp\right)
-\lambda_{r_s+1}\left(P_V + P_V^\perp\left(\frac{1}{K}\sum_{k=1}^{K}P_W^{(k)}\right)P_V^\perp\right)\\
&\geq
\lambda_{r_s}\left(P_V\right) + \lambda_{n}\left(P_V^\perp\left(\frac{1}{K}\sum_{k=1}^{K}P_W^{(k)}\right)P_V^\perp\right)
-\lambda_{r_s+1}\left(P_V\right) - \lambda_1\left(P_V^\perp\left(\frac{1}{K}\sum_{k=1}^{K}P_W^{(k)}\right)P_V^\perp\right)
\end{aligned}
\end{equation*}
\begin{equation*}
\begin{aligned}
&\geq
\lambda_{r_s}\left(P_V\right) -\lambda_{r_s+1}\left(P_V\right)- \lambda_1\left(P_V^\perp\left(\frac{1}{K}\sum_{k=1}^{K}P_W^{(k)}\right)P_V^\perp\right)\\
&=1-\lambda_1\left(P_V^\perp\left(\frac{1}{K}\sum_{k=1}^{K}P_W^{(k)}\right)P_V^\perp\right)
=1-\left\|P_V^\perp\left(\frac{1}{K}\sum_{k=1}^{K}P_W^{(k)}\right)P_V^\perp\right\|_{\op}\\
&\geq
1-\left\|\frac{1}{K}\sum_{k=1}^{K}P_W^{(k)}\right\|_{\op}
\overset{\text{Assumption \ref{assum:3}}}{\geq} \theta,
\end{aligned}
\end{equation*}
where the last second row holds as $P_V$ is a $r_s$-dimensional projection matrix, i.e., the leading $r_s$ eigenvalues of $P_V$ equal $1$, while the others $0$. Based on the above discussion, we have: 
$$
\left\|\hat{P}_V-P_V\right\|_F=\left\|P_V-\hat{P}_V\right\|_F
\leq
\frac{2}{\theta}
\left\|\tilde{\Sigma}_0-\Sigma_0\right\|_F.
$$
Thus, in order to bound $\|\hat{P}_V-P_V\|_F$, it remains to establish the bound of $\|\tilde{\Sigma}_0-\Sigma_0\|_F$.
$\\$
\textbf{\textit{Step 2:}}
In this step, we decompose $\|\tilde{\Sigma}_0-\Sigma_0\|_F$. Recall from \eqref{eq:apm}, we have: 
$$
\tilde{\Sigma}_{0}=\frac{1}{K}\sum_{k=1}^{K}\tilde{V}^{(k)}(\tilde{V}^{(k)})^{\top}=\frac{1}{K}\sum_{k=1}^{K}\tilde{P}^{(k)}.
$$
Thus, we can rewrite $\tilde{\Sigma}_0$ as:
\begin{equation}
\label{eq:the3_tildeSigma_decompose}
\begin{aligned}
K\tilde{\Sigma}_0
&=\sum_{k=1}^{K}\tilde{P}^{(k)}
=\sum_{k=1}^{K}\left(P^{(k)}+\tilde{P}^{(k)}-P^{(k)}\right)
=\sum_{k=1}^{K}\left(P^{(k)}+\Delta^{(k)}\right)\\
&=\sum_{k=1}^{K}\left(P_V+P_W^{(k)}+\Delta^{(k)}\right)
=KP_V+\sum_{k=1}^{K}P_W^{(k)}+\sum_{k=1}^{K}\Delta^{(k)}\\
&=KP_V+\left(P_V+P_V^\perp\right)\left(\sum_{k=1}^{K}P_W^{(k)}\right)\left(P_V+P_V^\perp\right)+\sum_{k=1}^{K}\Delta^{(k)}\\
&=KP_V+P_V\left(\sum_{k=1}^{K}P_W^{(k)}\right)P_V+P_V^\perp\left(\sum_{k=1}^{K}P_W^{(k)}\right)P_V\\
&\quad+P_V\left(\sum_{k=1}^{K}P_W^{(k)}\right)P_V^\perp+P_V^\perp\left(\sum_{k=1}^{K}P_W^{(k)}\right)P_V^\perp+\sum_{k=1}^{K}\Delta^{(k)}.
\end{aligned}
\end{equation}
By utilizing the fact that $P_VP_{W}^{(k)}=0_{n\times n}$ for any $k\in[K]$, we have:
$$
K\tilde{\Sigma}_0
=\underbrace{K\left[
P_V+P_V^\perp\left(\frac{1}{K}\sum_{k=1}^{K}P_W^{(k)}\right)P_V^\perp\right]}_{K\Sigma_0}
+
\sum_{k=1}^K\Delta^{(k)},
$$
which gives us:
$$
\left\|
\tilde{\Sigma}_0-\Sigma_0\right\|_F=\frac{1}{K}\left\|\sum_{k=1}^{K}\Delta^{(k)}\right\|_F.
$$
Further, splitting each $\Delta^{(k)}$ gives us:
\begin{equation*}
\begin{aligned}
\frac{1}{K}\left\|\sum_{k=1}^{K}\Delta^{(k)}\right\|_F
&\leq\frac{1}{K}\left\|\sum_{k=1}^{K}\left(\tilde{P}^{(k)}-\E \tilde{P}^{(k)}\right)\right\|_F
+
\frac{1}{K}\left\|\sum_{k=1}^{K}\left(\E \tilde{P}^{(k)}-P^{(k)}\right)\right\|_F\\
&=\frac{1}{K}\left\|\sum_{k=1}^{K}\tilde{\Delta}^{(k)}\right\|_F
+
\frac{1}{K}\left\|\sum_{k=1}^{K}\left(\E \tilde{P}^{(k)}-P^{(k)}\right)\right\|_F,
\end{aligned}
\end{equation*}
where we re-write $\tilde{P}^{(k)}-\E P^{(k)}$ as $\tilde{\Delta}^{(k)}$. Then, we have:
$$
\left\|\left\|
\tilde{\Sigma}_0-\Sigma_0\right\|_F\right\|_{\psi_1}
\leq
\underbrace{\frac{1}{K}\left\|\left\|\sum_{k=1}^{K}\tilde{\Delta}^{(k)}\right\|_F\right\|_{\psi_1}}_{(i)}
+
\underbrace{
\frac{1}{K}\left\|\left\|\sum_{k=1}^{K}\left(\E \tilde{P}^{(k)}-P^{(k)}\right)\right\|_F\right\|_{\psi_1}}_{(ii)}.
$$
Accordingly, to bound $\|\tilde{\Sigma}_0-\Sigma_0\|_F$, we only need to discuss terms $(i)$ and $(ii)$, respectively.
$\\$
\textbf{\textit{Step 3:}}
For term $(i)$, we first analyze the local term $\tilde{\Delta}^{(k)}$. According to the definition of the $\psi_1$ norm, we have:
\begin{equation*}
\begin{aligned}
\left\|\left\|\tilde{\Delta}^{(k)}\right\|_F\right\|_{\psi_1}
&=
\left\|\left\|\tilde{P}^{(k)}-\E \tilde{P}^{(k)}\right\|_F\right\|_{\psi_1}
\leq
\left\|\left\|\tilde{P}^{(k)}-P^{(k)}\right\|_F\right\|_{\psi_1}+
\left\|\left\|\E \tilde{P}^{(k)}-P^{(k)}\right\|_F\right\|_{\psi_1}\\
&\overset{\text{Jensen}}{\leq}
\left\|\left\|\tilde{P}^{(k)}-P^{(k)}\right\|_F\right\|_{\psi_1}+
\left\|\E\left\|\tilde{P}^{(k)}- P^{(k)}\right\|_F\right\|_{\psi_1}\\
&\overset{\psi_1 \text{norm}}{\leq}
\left\|\left\|\tilde{P}^{(k)}- P^{(k)}\right\|_F\right\|_{\psi_1}
=
\left\|\left\|\Delta^{(k)}\right\|_F\right\|_{\psi_1}.
\end{aligned}
\end{equation*}
Recall from  \eqref{tildeMk}-\eqref{Ek}, $\tilde{P}^{(k)}$ and $P^{(k)}$ correspond to the left leading $r_k$-dimensional singular subspace of $\tilde{M}^{(k)}$ and $M^{(k)}$, respectively. Thus, we can  acquire:
$$
\left\|\left\|\tilde{\Delta}^{(k)}\right\|_F\right\|_{\psi_1}
\leq
\left\|\left\|\Delta^{(k)}\right\|_F\right\|_{\psi_1}
\overset{\text{Proposition \ref{prop:wedin}}}{\lesssim}
\left\|\frac{\left\|E^{(k)}\right\|_{\op}}{\sigma_{r_k}(M^{(k)})}\right\|_{\psi_1}:=\left\|\varepsilon^{(k)}\right\|_{\psi_1},
$$
where $\varepsilon^{(k)}:=\left\|E^{(k)}\right\|_{\op}/\sigma_{r_k}(M^{(k)})$.
By using Lemmas \ref{lem:signal}, \ref{lem:noise} and Proposition \ref{prop:subexp}, we have:
\begin{equation*}
\begin{aligned}
(i)&=\frac{1}{K}\left\|\left\|\sum_{k=1}^{K}\tilde{\Delta}^{(k)}\right\|_F\right\|_{\psi_1}
\overset{\text{Proposition \ref{prop:subexp}}}{\lesssim}
\frac{1}{K}\sqrt{\sum_{k=1}^{K}\left\|\left\|\tilde{\Delta}^{(k)}\right\|_F\right\|_{\psi_1}^2}
\lesssim
\frac{1}{K}\sqrt{\sum_{k=1}^{K}\left\|\varepsilon^{(k)}\right\|_{\psi_1}^2}\\
&\overset{\text{Lemma \ref{lem:signal}}}{\lesssim}
\frac{1}{K\lambda_{\min}}\sqrt{\sum_{k=1}^{K}\left\|\left\|E^{(k)}\right\|_{\op}\right\|_{\psi_1}^2}
\overset{\text{Lemma \ref{lem:noise}}}{\lesssim}
\frac{1}{K\lambda_{\min}}\sqrt{\sum_{k=1}^{K}n}
=
\sqrt{
\frac{n}{K\lambda_{\min}^2}
}.
\end{aligned}
\end{equation*}
$\\$

\textbf{\textit{Step 4:}}
For term $(ii)$, let $Q^{(k)}=\Delta^{(k)}-G^{(k)}=\tilde{P}^{(k)}-P^{(k)}-G^{(k)}$, where
\begin{equation*}
G^{(k)}=(P^{(k)})^\perp E^{(k)}(M^{(k)})^\dagger
+\bigl((M^{(k)})^\dagger\bigr)^\top(E^{(k)})^\top(P^{(k)})^\perp.
\end{equation*}
Here $(M^{(k)})^\dagger\in\mathbb{R}^{r_\Pi^{(k)}\times n}$ is the Moore--Penrose inverse of $M^{(k)}$, so both terms in $G^{(k)}$ are $n\times n$ matrices. This is the first-order perturbation of the left singular-subspace projector of the rank-$r_k$ matrix $M^{(k)}$. In particular,
\begin{equation*}
\|G^{(k)}\|_F
\leq 2\|E^{(k)}\|_{\op}\|(M^{(k)})^\dagger\|_F
\leq 2\sqrt{r_k}\,\varepsilon^{(k)}.
\end{equation*}

To bound the remainder, apply Proposition \ref{prop:taylor} to the leading $r_k$ eigenspaces of the symmetric dilation
\begin{equation*}
\begin{pmatrix}0&M^{(k)}\\(M^{(k)})^\top&0\end{pmatrix}
\quad\text{with perturbation}\quad
\begin{pmatrix}0&E^{(k)}\\(E^{(k)})^\top&0\end{pmatrix},
\end{equation*}
and to their negatives. Both eigengaps are at least $\sigma_{r_k}(M^{(k)})$, and the perturbation norm equals $\|E^{(k)}\|_{\op}$. Adding the two projector expansions and taking the upper-left block gives the expansion of $\tilde P^{(k)}-P^{(k)}$ above: the first-order terms between the positive and negative eigenspaces cancel, and the terms involving the zero eigenspace give $G^{(k)}$. The two remainder bounds in Proposition \ref{prop:taylor} therefore yield
\begin{equation*}
\|Q^{(k)}\|_F\leq 48\sqrt{r_k}\,(\varepsilon^{(k)})^2,
\qquad \varepsilon^{(k)}\leq 1/10.
\end{equation*}

With the data-mode estimators $\{\hat U_j^{(k)}\}_j$ held fixed, $M^{(k)}$ is fixed and $E^{(k)}=Z^{(k)}\hat U_\otimes^{(k)}$ has mean zero by Assumption \ref{assum:1}. Since $G^{(k)}$ is linear in $E^{(k)}$, it follows that $\E G^{(k)}=0_{n\times n}$.
Let $\Omega=[0,1/10]$ and $\Omega^c=(1/10,+\infty)$, with $\mathbb I_\Omega$ and $\mathbb I_{\Omega^c}$ denoting the corresponding indicator functions. We have
\begin{equation*}
\begin{aligned}
Q^{(k)}
&=Q^{(k)}\mathbb{I}_{\Omega}(\varepsilon^{(k)})
+\left(\tilde{P}^{(k)}-P^{(k)}\right)\mathbb{I}_{\Omega^c}(\varepsilon^{(k)})
-G^{(k)}\mathbb{I}_{\Omega^c}(\varepsilon^{(k)}).
\end{aligned}
\end{equation*}
Consequently, the triangle inequality and Jensen's inequality give
\begin{equation*}
\begin{aligned}
\left\|\E\left(\tilde{P}^{(k)}-P^{(k)}\right)\right\|_F
&=\|\E Q^{(k)}\|_F\leq\E\|Q^{(k)}\|_F\\
&\leq\underbrace{\E\left[\|Q^{(k)}\|_F\mathbb I_\Omega(\varepsilon^{(k)})\right]}_{(iii)}
+\underbrace{\E\left[\|\tilde P^{(k)}-P^{(k)}\|_F\mathbb I_{\Omega^c}(\varepsilon^{(k)})\right]}_{(iv)}\\
&\quad+\underbrace{\E\left[\|G^{(k)}\|_F\mathbb I_{\Omega^c}(\varepsilon^{(k)})\right]}_{(v)}.
\end{aligned}
\end{equation*}
For term $(iii)$, the remainder bound above implies
\begin{equation*}
(iii)\leq 48\sqrt{r_k}\,\E\left[(\varepsilon^{(k)})^2\mathbb I_\Omega(\varepsilon^{(k)})\right]
\leq 48\sqrt{r_k}\,\E(\varepsilon^{(k)})^2\lesssim\E(\varepsilon^{(k)})^2.
\end{equation*}
For term $(iv)$, both $\tilde P^{(k)}$ and $P^{(k)}$ are rank-$r_k$ orthogonal projectors, so $\|\tilde P^{(k)}-P^{(k)}\|_F\leq\sqrt{2r_k}\varepsilon^{(k)}$.
$$
(iv)=\E\left(\left\|\tilde{P}^{(k)}-P^{(k)}\right\|_F\mathbb{I}_{\Omega^c}(\varepsilon^{(k)})\right)
\overset{\text{Proposition \ref{prop:wedin}}}{\lesssim}
\E\left(\varepsilon^{(k)}\mathbb{I}_{\Omega^c}(\varepsilon^{(k)})\right)
\leq
10\E(\varepsilon^{(k)})^2
\lesssim
\E(\varepsilon^{(k)})^2.
$$
For term $(v)$, the bound on $G^{(k)}$ and the inequality
$\varepsilon^{(k)}\mathbb I_{\Omega^c}(\varepsilon^{(k)})\leq10(\varepsilon^{(k)})^2$ give
\begin{equation*}
(v)\leq2\sqrt{r_k}\,\E\left[\varepsilon^{(k)}\mathbb I_{\Omega^c}(\varepsilon^{(k)})\right]
\leq20\sqrt{r_k}\,\E(\varepsilon^{(k)})^2\lesssim\E(\varepsilon^{(k)})^2.
\end{equation*}
By Assumption \ref{assum:4}, $r_k$ is bounded. Lemmas \ref{lem:signal} and \ref{lem:noise} now yield
\begin{equation*}
\E\|Q^{(k)}\|_F
\lesssim\E(\varepsilon^{(k)})^2
\lesssim\|\varepsilon^{(k)}\|_{\psi_1}^2
\lesssim\frac{1}{\lambda_{\min}^2}
\left\|\|E^{(k)}\|_{\op}\right\|_{\psi_1}^2
\lesssim\frac{n}{\lambda_{\min}^2}.
\end{equation*}
Finally, since $\E\tilde P^{(k)}-P^{(k)}$ is deterministic, the triangle inequality gives
\begin{equation*}
(ii)=\frac{1}{K}\left\|\left\|\sum_{k=1}^K\left(\E\tilde P^{(k)}-P^{(k)}\right)\right\|_F\right\|_{\psi_1}
\lesssim\frac{1}{K}\sum_{k=1}^K\E\|Q^{(k)}\|_F
\lesssim\frac{n}{\lambda_{\min}^2}.
\end{equation*}

$\\$
\textbf{\textit{Step 5:}}
Combining the results in \textbf{\textit{Step 3}} and \textbf{\textit{Step 4}} gives us:
\begin{equation*}
\begin{aligned}
\left\|\left\|
\tilde{\Sigma}_0-\Sigma_0\right\|_F\right\|_{\psi_1}
&\leq
\underbrace{\frac{1}{K}\left\|\left\|\sum_{k=1}^{K}\tilde{\Delta}^{(k)}\right\|_F\right\|_{\psi_1}}_{(i)}
+
\underbrace{
\frac{1}{K}\left\|\left\|\sum_{k=1}^{K}\left(\E \tilde{P}^{(k)}-P^{(k)}\right)\right\|_F\right\|_{\psi_1}}_{(ii)}\\
&\lesssim\sqrt{\frac{n}{K\lambda_{\min}^2}}+\frac{n}{\lambda_{\min}^2},
\end{aligned}
\end{equation*}
which gives us:
\begin{equation*}
\begin{aligned}
\left\|\left\|\hat{P}_V-P_V\right\|_F\right\|_{\psi_1}
\leq
\frac{2}{\theta}
\left\|\left\|\tilde{\Sigma}_0-\Sigma_0\right\|_F\right\|_{\psi_1}
\lesssim
\sqrt{\frac{n}{K\theta^2\lambda_{\min}^2}}+\frac{n}{\theta\lambda_{\min}^2}.
\end{aligned}
\end{equation*}
According to the definition of the sub-exponential norm, we have:
$$
\left\|\hat{P}_V-P_V\right\|_F=O_p\left(\sqrt{\frac{n}{K\theta^2\lambda_{\min}^2}}+\frac{n}{\theta\lambda_{\min}^2}
\right).
$$
By utilizing the fact that $(a+b)^2\lesssim a^2+b^2$ for any bounded $a,b$, we can finally deduce that:
$$
\left\|\hat{P}_V-P_V\right\|_F^2=O_p\left\{\left(\sqrt{\frac{n}{K\theta^2\lambda_{\min}^2}}+\frac{n}{\theta\lambda_{\min}^2}
\right)^2\right\}
=O_p\left(\frac{n}{K\theta^2\lambda_{\min}^2}+\frac{n^2}{\theta^2\lambda_{\min}^4}\right).
$$
The first claim of Theorem \ref{theorem:upper bound} is completed by realizing $\lambda_{\min}=\sqrt{n}\sigma_{\min}$.

\subsubsection{Proof of the second claim in Theorem \ref{theorem:upper bound}}
\label{supp:upper bound of Wk}
We change the superscript $k$ in \eqref{main rates for Wk} to $l$ to avoid confusion. Define $\hat{\Gamma}^{(l)}
:=
(I_n-\hat{P}_V)\tilde{P}^{(l)}(I_n-\hat{P}_V)$ and $\Gamma^{(l)}
:=
(I_n-P_V)P^{(l)}(I_n-P_V)$.
According to Algorithm \ref{alg:decouple}, we know that $\hat{P}_W^{(l)}$ and $P_W^{(l)}$ are the projection matrices corresponding to the left leading $(r_l-r_s)$-dimensional singular subspace of $(I_n-\hat{P}_V)\tilde{V}^{(l)}$ and $(I_n-P_V)V^{(l)}$, respectively. Thus, we have:
\begin{equation*}
\begin{aligned}
&\quad\left\|P_W^{(l)}-\hat{P}_W^{(l)}\right\|_F
\overset{\text{Proposition \ref{prop:wedin}}}{\lesssim}
\frac{\left\|\hat{\Gamma}^{(l)}-\Gamma^{(l)}\right\|_{\op}}{\sigma_{r_l-r_s}(\Gamma^{(l)})}
\lesssim
\left\|\hat{\Gamma}^{(l)}-\Gamma^{(l)}\right\|_{\op}\\
&:=
\left\|(I_n-\hat{P}_V)\tilde{P}^{(l)}(I_n-\hat{P}_V)-(I_n-P_V)P^{(l)}(I_n-P_V)\right\|_{\op}.
\end{aligned}
\end{equation*}
We decompose $\hat{\Gamma}^{(l)}$ as follows:
\begin{equation*}
\begin{aligned}
&\quad\hat{\Gamma}^{(l)}
=
(I_n-\hat{P}_V)\tilde{P}^{(l)}(I_n-\hat{P}_V)
=
(I_n-\hat{P}_V)(P^{(l)}+\Delta^{(l)})(I_n-\hat{P}_V)\\
&=(I_n-\hat{P}_V)P^{(l)}(I_n-\hat{P}_V)+(I_n-\hat{P}_V)\Delta^{(l)}(I_n-\hat{P}_V)
=
\hat{P}_V^\perp P^{(l)}\hat{P}_V^\perp +
\hat{P}_V^\perp \Delta^{(l)}\hat{P}_V^\perp \\
&=(\hat{P}_V^\perp-P_V^\perp+P_V^\perp) P^{(l)}(\hat{P}_V^\perp-P_V^\perp+P_V^\perp) +
\hat{P}_V^\perp \Delta^{(l)}\hat{P}_V^\perp.
\end{aligned}
\end{equation*}
For the first term above, we further decompose it as:
\begin{equation*}
\begin{aligned}
&\quad\quad\quad\quad\quad\quad\quad\quad\quad(\hat{P}_V^\perp-P_V^\perp+P_V^\perp) P^{(l)}(\hat{P}_V^\perp-P_V^\perp+P_V^\perp)\\
&=(\hat{P}_V^\perp-P_V^\perp) P^{(l)}(\hat{P}_V^\perp-P_V^\perp)
+
P_V^\perp P^{(l)}(\hat{P}_V^\perp-P_V^\perp)+
(\hat{P}_V^\perp-P_V^\perp) P^{(l)}P_V^\perp+
\underbrace{P_V^\perp P^{(l)}P_V^\perp}_{\Gamma^{(l)}}.
\end{aligned}
\end{equation*}
Accordingly, we have:
$$
\hat{\Gamma}^{(l)}-\Gamma^{(l)}
=
(\hat{P}_V^\perp-P_V^\perp) P^{(l)}(\hat{P}_V^\perp-P_V^\perp)
+
P_V^\perp P^{(l)}(\hat{P}_V^\perp-P_V^\perp)
+
(\hat{P}_V^\perp-P_V^\perp) P^{(l)}P_V^\perp
+
\hat{P}_V^\perp \Delta^{(l)}\hat{P}_V^\perp,
$$
which gives us:
\begin{equation*}
\begin{aligned}
\left\|\hat{\Gamma}^{(l)}-\Gamma^{(l)}\right\|_{\op}
&\leq
\left\|(\hat{P}_V^\perp-P_V^\perp) P^{(l)}(\hat{P}_V^\perp-P_V^\perp)\right\|_{\op}
+
2\left\|P_V^\perp P^{(l)}(\hat{P}_V^\perp-P_V^\perp)\right\|_{\op}
+
\left\|\hat{P}_V^\perp \Delta^{(l)}\hat{P}_V^\perp\right\|_{\op}\\
&\lesssim
\left\|\hat{P}_V^\perp-P_V^\perp\right\|_{\op}
+
\left\|\hat{P}_V^\perp-P_V^\perp\right\|_{\op}^2
+
\left\|\Delta^{(l)}\right\|_{\op}\\
&\leq
\left\|\hat{P}_V-P_V\right\|_F
+
\left\|\hat{P}_V-P_V\right\|_F^2
+
\left\|\Delta^{(l)}\right\|_F\\
&\overset{\text{Section \ref{supp:upper bound of V}}}{=}\left\|\Delta^{(l)}\right\|_F+
O_p\left(\sqrt{\frac{n}{K\theta^2\lambda_{\min}^2}}+\frac{n}{\theta\lambda_{\min}^2}\right).
\end{aligned}
\end{equation*}
By utilizing the same discussion as \textbf{\textit{Step 3}} in Section \ref{supp:upper bound of V}, we have 
$$
\left\|\left\|\Delta^{(l)}\right\|_F\right\|_{\psi_1}
\lesssim
\left\|\varepsilon^{(l)}\right\|_{\psi_1}
\lesssim
\left\|\frac{\left\|E^{(l)}\right\|_{\op}}{\lambda_{\min}}\right\|_{\psi_1}
\overset{\text{Lemma \ref{lem:noise}}}{\lesssim}
\frac{\sqrt{n}}{\lambda_{\min}}.
$$
Accordingly, we have:
$$
\left\|P_W^{(l)}-\hat{P}_W^{(l)}\right\|_F
\lesssim
\left\|\hat{\Gamma}^{(l)}-\Gamma^{(l)}\right\|_{\op}
=O_p
\left(
\frac{\sqrt{n}}{\lambda_{\min}}+
\sqrt{\frac{n}{K\theta^2\lambda_{\min}^2}}+\frac{n}{\theta\lambda_{\min}^2}
\right),
$$
which further gives us:
$$
\left\|P_W^{(l)}-\hat{P}_W^{(l)}\right\|_F^2
=O_p
\left(
\frac{n}{\lambda_{\min}^2}+
\frac{n}{K\theta^2\lambda_{\min}^2}+\frac{n^2}{\theta^2\lambda_{\min}^4}
\right).
$$
Until now, we have completed the proof of Theorem \ref{theorem:upper bound} by again using $\lambda_{\min}=\sqrt{n}\sigma_{\min}$.

\subsection{Proof of Theorem \ref{theorem:minimax}}
To prove Theorem \ref{theorem:minimax}, 
we can relax the left hand side of \eqref{minimax for V} and \eqref{minimax for Wk} into
\begin{equation*}
\inf_{\hat{V}}\sup_{\{\cT^{(k)}\}_k\in\Theta}\E\left\|\hat{V}\hat{V}^{\top}-VV^{\top}\right\|^2_F
\gtrsim
\inf_{\hat{V}}\sup_{\{\cT^{(k)}\}_k\in\Theta''}\E\left\|\hat{V}\hat{V}^{\top}-VV^{\top}\right\|^2_F
\end{equation*}
\begin{equation*}
\inf_{\hat{W}^{(l)}}\sup_{\{\cT^{(k)}\}_k\in\Theta}\E\left\|\hat{W}^{(l)}(\hat{W}^{(l)})^{\top}-W^{(l)}(W^{(l)})^{\top}\right\|^2_F
\gtrsim
\inf_{\hat{W}^{(l)}}\sup_{\{\cT^{(k)}\}_k\in\Theta'}\E\left\|\hat{W}^{(l)}(\hat{W}^{(l)})^{\top}-W^{(l)}(W^{(l)})^{\top}\right\|^2_F,
\end{equation*}
as $\Theta'$ and $\Theta''$ are both subsets of $\Theta$ defined in the following two subsections. 

\subsubsection{Proof of the second claim in Theorem \ref{theorem:minimax}}
We first provide the specific definition of $\Theta'$. For a fixed $l$ and a pre-known set of hyper-parameters $\{\{p^{(k)}, r^{(k)},L_k\}_{k\in[K]};r_s; n;K;\sigma_{\max};\sigma_{\min}\}$, we define  $\Theta'\subseteq\Theta$ as:
\begin{equation}
\label{Theta'}
\begin{aligned}
\Theta':=\biggl\{
&\left\{\cT^{(k)}=\cC^{(k)}\times_1 U_1^{(k)}\times_2
U_2^{(k)}\times_3\cdots\times_{L_k} U_{L_k}^{(k)}\times_{L_k+1}[V,W^{(k)}]\right\}_k
\Big|\ \sigma_{\max}\leq C\sigma_{\min}, \\
&\sigma_{\min}\leq\min_{k\in[K]}\{\sigma_{r_k}(\cM_{L_k+1}(\cC^{(k)})/\sqrt{n})\}\leq\cdots\leq\max_{k\in[K]}\{\sigma_{1}(\cM_{L_k+1}(\cC^{(k)})/\sqrt{n})\}\leq\sigma_{\max};\\
& V^\top V=I_{r_s}, \{\{(U_j^{(k)})^\top U_j^{(k)}=I_{r_j^{(k)}}\}_{j=1}^{L_k}, (W^{(k)})^\top W^{(k)}=I_{r_k-r_s}, (W^{(k)})^\top V=0\}_{k=1}^K;\\
& V,\{\cC^{(k)}, \{U_j^{(k)}\}_{j=1}^{L_k}\}_{k=1}^K, \{W^{(k)}\}_{k\in[K]\backslash\{l\}}\ \text{are fixed parameters}
\biggr\}.
\end{aligned}
\end{equation}
According to \eqref{Theta'} and the definitions of $\lambda_{\max}$ and $\lambda_{\min}$, we have:
$$
\lambda_{\min}\leq\min_{k\in[K]}(\sigma_{r_k}(\cM_{L_k+1}(\cC^{(k)})))\leq\cdots\leq\max_{k\in[K]}(\sigma_{1}(\cM_{L_k+1}(\cC^{(k)})))\leq\lambda_{\max}\leq C\lambda_{\min}
$$
To prove \eqref{minimax for Wk}, we only need to obtain
\begin{equation*}
\inf_{\hat{W}^{(l)}}\sup_{\{\cT^{(k)}\}_k\in\Theta'}\E\left\|\hat{W}^{(l)}(\hat{W}^{(l)})^{\top}-W^{(l)}(W^{(l)})^{\top}\right\|^2_F
\gtrsim
\frac{1}{\sigma_{\min}^2}.
\end{equation*}
$\\$
\textbf{\textit{Local packing:}}

Note that for any $\{\cT^{(k)}\}_k\in\Theta'$, the only changing parameter is $W^{(l)}$. Thus, we consider the metric space $(\Theta',\rho')$ where $\rho'$ is defined as 
$$
\rho'\left(\{\cT_k\}_k,\{\acute{\cT}_k\}_k\right)
:=
\left\|W^{(l)}\left(W^{(l)}\right)^\top-\acute{W}^{(l)}\left(\acute{W}^{(l)}\right)^\top\right\|_F.
$$
$P_W^{(l)}=W^{(l)}(W^{(l)})^\top$ and $P_{\acute{W}}^{(l)}=\acute{W}^{(l)}(\acute{W}^{(l)})^\top$ are the corresponding changing parameters for $\{\cT_k\}_k$ and $\{\acute{\cT}_k\}_k$, respectively. If we denote the orthogonal complement of $V$ as $V^\perp\in\mathbb{R}^{n\times(n-r_s)}$, we can deduce from $W^{(l)}\perp V$ that there exists a $\Xi^{(l)}\in\mathbb{R}^{(n-r_s)\times(r_l-r_s)}$ such that the corresponding unchanged $W^{(l)}$ satisfies $W^{(l)}=V^\perp\Xi^{(l)}$. As $(V^\perp)^\top V^\perp=I_{n-r_s}$ holds, we have:
$$
\left(\Xi^{(l)}\right)^\top\Xi^{(l)}
=\left(\Xi^{(l)}\right)^\top\left(V^\perp\right)^\top V^\perp\Xi^{(l)}=\left(W^{(l)}\right)^\top W^{(l)}=I_{r_l-r_s}.
$$
Accordingly, we have $\Xi^{(l)}\in\cO(n-r_s,r_l-r_s)=\{U\in\mathbb{R}^{(n-r_s)\times(r_l-r_s)}: U^\top U=I_{r_l-r_s}\}$. In contrast, for any $\Xi^{(l)}$ such that  $\Xi^{(l)}\in\cO(n-r_s,r_l-r_s)$, if we let $W^{(l)}$ be $V^\perp\Xi^{(l)}$, then we can deduce that $W^{(l)}\in\Theta'$, or equivalently the corresponding $\{\cT^{(k)}\}_k\in\Theta'$, holds. Define $\Phi:\cO(n-r_s,r_l-r_s)\rightarrow\Theta'$ as:
$$
\Phi(\Xi^{(l)})=V^\perp\Xi^{(l)}=W^{(l)}.
$$
Moreover, for any two distinct $\{\cT^{(k)}\}_k$, $\{\acute{\cT}^{(k)}\}_k\in\Theta'$, we have:
\begin{equation*}
\begin{aligned}
\rho_{\cO}\left(\Xi^{(l)},\acute{\Xi}^{(l)}\right)
&=
\left\|\Xi^{(l)}\left(\Xi^{(l)}\right)^\top
-
\acute{\Xi}^{(l)}\left(\acute{\Xi}^{(l)}\right)^\top\right\|_F\\
&=
\left\|V^\perp\left[\Xi^{(l)}\left(\Xi^{(l)}\right)^\top
-
\acute{\Xi}^{(l)}\left(\acute{\Xi}^{(l)}\right)^\top\right]\left(V^\perp\right)^\top\right\|_F\\
&=
\left\|\left(V^\perp\Xi^{(l)}\right)\left(V^\perp\Xi^{(l)}\right)^\top
-
\left(V^\perp\acute{\Xi}^{(l)}\right)\left(V^\perp\acute{\Xi}^{(l)}\right)^\top\right\|_F\\
&=
\left\|W^{(l)}\left(W^{(l)}\right)^\top-\acute{W}^{(l)}\left(\acute{W}^{(l)}\right)^\top\right\|_F\\
&=\rho'\left(\{\cT_k\}_k,\{\acute{\cT}_k\}_k\right),
\end{aligned}
\end{equation*}
where $\rho_{\cO}$ is the metric for $\cO(p,r)$ with $p>r$.
Accordingly, $\Phi$ is an isometric isomorphism map between $(\Theta',\rho')$ and $(\cO(n-r_s,r_l-r_s),\rho_\cO)$, i.e., $(\Theta',\rho')\cong(\cO(n-r_s,r_l-r_s),\rho_\cO)$. Therefore, we can simply transform our attention from $\Theta'$ to $\cO(n-r_s,r_l-r_s)$ and construct a local packing set of $\cO(n-r_s,r_l-r_s)$.

We first consider $\cO(n-r_s-(r_l-r_s), r_l-r_s)=\cO(n-r_l,r_l-r_s)$. For any $\epsilon>0$, $0\leq r_s< r_l< n$, and $n-r_l> r_l-r_s$ (guaranteed by Assumption \ref{assum:4}), we can utilize Proposition \ref{prop:packing} to acquire the following result by taking $q=2$:
\[
\left( \frac{c}{\epsilon} \right)^{(r_l-r_s)(n +r_s - 2r_l)}
\leq 
D(\cO(n-r_l,r_l-r_s), \rho_{\cO}, \epsilon \sqrt{r_l-r_s}) 
\leq 
\left( \frac{C}{\epsilon} \right)^{(r_l-r_s)(n +r_s - 2r_l)},
\]
where $D(T,d,\epsilon)$, the $\epsilon$-packing number of a given metric space $(T,d)$, is defined in Proposition \ref{prop:packing} formally. Based on the definition of $\epsilon$-packing number, there exists a  $\dot{\cO}(n-r_l,r_l-r_s)\subseteq\cO(n-r_l,r_l-r_s)$ such that $|\dot{\cO}(n-r_l,r_l-r_s)|\geq\left( c/\epsilon \right)^{(r_l-r_s)(n +r_s - 2r_l)}$, while for any two distinct $\dot{\Xi}_1^{(l)}$, $\dot{\Xi}_2^{(l)}\in\dot{\cO}(n-r_l,r_l-r_s)$, we have:
$$
\rho_{\cO}\left(\dot{\Xi}_1^{(l)},\dot{\Xi}_2^{(l)}\right)
=
\left\|\dot{\Xi}_1^{(l)}\left(\dot{\Xi}_1^{(l)}\right)^\top
-
\dot{\Xi}_2^{(l)}\left(\dot{\Xi}_2^{(l)}\right)^\top\right\|_F
\geq
\sqrt{r_l-r_s}\epsilon.
$$
By setting $\epsilon=c/2$, we have:
$$
\left|\dot{\cO}(n-r_l,r_l-r_s)\right|
\geq
2^{(r_l-r_s)(n +r_s - 2r_l)},
\ \text{and}\ \rho_{\cO}\left(\dot{\Xi}_1^{(l)},\dot{\Xi}_2^{(l)}\right)
\geq
\frac{c\sqrt{r_l-r_s}}{2}.
$$
For a fixed $\delta>0$ where the value is to be determined later and every $\dot{\Xi}^{(l)}\in\dot{\cO}(n-r_l,r_l-r_s)$, we define:
$$\ddot{\Xi}^{(l)} = \begin{pmatrix} \sqrt{1 - \delta} I_{r_l-r_s} \\ \sqrt{\delta} \, \dot{\Xi}^{(l)} \end{pmatrix} \in \mathbb{R}^{(n-r_s) \times (r_l-r_s)}.$$
Then we have:
$(\ddot{\Xi}^{(l)})^\top \ddot{\Xi}^{(l)}=(1-\delta)I_{r_l-r_s}+\delta(\dot{\Xi}^{(l)})^\top \dot{\Xi}^{(l)}=I_{r_l-r_s}$, which implies that $\ddot{\Xi}^{(l)}\in\cO(n-r_s,r_l-r_s)$ also holds. If we collect all $\ddot{\Xi}^{(l)}$ as $\ddot{\cO}(n-r_s,r_l-r_s)$, then we have
$$
\left|\ddot{\cO}(n-r_s,r_l-r_s)\right|
=
\left|\dot{\cO}(n-r_l,r_l-r_s)\right|
\geq
2^{(r_l-r_s)(n +r_s - 2r_l)},
$$
as the changing parameters for  $\ddot{\cO}(n-r_s,r_l-r_s)$ and $\dot{\cO}(n-r_l,r_l-r_s)$ are both $\dot{\Xi}^{(l)}$. Accordingly, for any two distinct $\ddot{\Xi}_1^{(l)}$, $\ddot{\Xi}_2^{(l)}\in\ddot{\cO}(n-r_s,r_l-r_s)$, we have:
$$
\ddot{\Xi}_1^{(l)}\ddot{\Xi}_1^{(l)\top} 
= \begin{pmatrix} \sqrt{1-\delta} I_{r_l-r_s} \\ \sqrt{\delta} \dot{\Xi}_1^{(l)} \end{pmatrix} 
\left( \sqrt{1-\delta} I_{r_l-r_s},\ \sqrt{\delta} \dot{\Xi}_1^{(l)\top} \right) = \begin{pmatrix} (1-\delta) I_{r_l-r_s} & \sqrt{\delta(1-\delta)} \dot{\Xi}_1^{(l)\top} \\ \sqrt{\delta(1-\delta)} \dot{\Xi}_1^{(l)} & \delta \dot{\Xi}_1^{(l)}\dot{\Xi}_1^{(l)\top} \end{pmatrix},
$$
and
$$
\ddot{\Xi}_2^{(l)}\ddot{\Xi}_2^{(l)\top} 
= \begin{pmatrix} \sqrt{1-\delta} I_{r_l-r_s} \\ \sqrt{\delta} \dot{\Xi}_2^{(l)} \end{pmatrix} 
\left( \sqrt{1-\delta} I_{r_l-r_s},\ \sqrt{\delta} \dot{\Xi}_2^{(l)\top} \right) = \begin{pmatrix} (1-\delta) I_{r_l-r_s} & \sqrt{\delta(1-\delta)} \dot{\Xi}_2^{(l)\top} \\ \sqrt{\delta(1-\delta)} \dot{\Xi}_2^{(l)} & \delta \dot{\Xi}_2^{(l)}\dot{\Xi}_2^{(l)\top} \end{pmatrix}.
$$
Further, we have:
$$
\left\|\ddot{\Xi}_1^{(l)}\ddot{\Xi}_1^{(l)\top} 
-\ddot{\Xi}_2^{(l)}\ddot{\Xi}_2^{(l)\top} \right\|_F
=
\left\|
\begin{pmatrix} 
0 & \sqrt{\delta(1-\delta)} \left(\dot{\Xi}_1^{(l)}-\dot{\Xi}_2^{(l)}\right)^\top \\ \sqrt{\delta(1-\delta)} \left(\dot{\Xi}_1^{(l)}-\dot{\Xi}_2^{(l)}\right) & \delta \left(\dot{\Xi}_1^{(l)}\left(\dot{\Xi}_1^{(l)}\right)^\top
-
\dot{\Xi}_2^{(l)}\left(\dot{\Xi}_2^{(l)}\right)^\top\right) \end{pmatrix}
\right\|_F.
$$
Based on the definition of $\|\cdot\|_F$ and the triangle inequality, we have:
$$
\rho_{\cO}\left(\ddot{\Xi}_1^{(l)},\ddot{\Xi}_2^{(l)}\right)
\geq
\sqrt{\delta(1-\delta)} \left\|\dot{\Xi}_1^{(l)}-\dot{\Xi}_2^{(l)}\right\|_F
\geq
\frac{\sqrt{\delta(1-\delta)}}{2}
\left\|\dot{\Xi}_1^{(l)}\left(\dot{\Xi}_1^{(l)}\right)^\top
-
\dot{\Xi}_2^{(l)}\left(\dot{\Xi}_2^{(l)}\right)^\top\right\|_F.
$$
Based on the cardinality lower bound of $\dot{\cO}(n-r_l,r_l-r_s)$, we can further deduce that:
\begin{equation}
\label{eq:the5_Wk_rho(Xi1-Xi2)}
\rho_{\cO}\left(\ddot{\Xi}_1^{(l)},\ddot{\Xi}_2^{(l)}\right)
\geq
\frac{\sqrt{\delta(1-\delta)}}{2}
\left\|\dot{\Xi}_1^{(l)}\left(\dot{\Xi}_1^{(l)}\right)^\top
-
\dot{\Xi}_2^{(l)}\left(\dot{\Xi}_2^{(l)}\right)^\top\right\|_F
\geq
\frac{\tilde{c}\sqrt{\delta(1-\delta)}\sqrt{r_l-r_s}}{2}.
\end{equation}
Besides, we also have:
\begin{equation}
\label{eq:the5_Wk_Xi1-Xi2_F}
\left\|\ddot{\Xi}_1^{(l)}
-\ddot{\Xi}_2^{(l)}\right\|_F
\leq
\sqrt{\delta}
\left\|\dot{\Xi}_1^{(l)}
-\dot{\Xi}_2^{(l)}\right\|_F
\leq
\sqrt{2\delta(r_l-r_s)}.
\end{equation}
The last inequality holds by realizing that $\dot{\Xi}_1^{(l)}$ and $\dot{\Xi}_2^{(l)}$ are both $(r_l-r_s)$-dimensional column orthonormal matrices.
Recall the definition of $\Phi$, we can build a subset of $\Theta'$,  $\ddot{\Theta}'\subseteq\Theta'$, such that:
$$
\ddot{\Theta}'
:=
\left\{
\left\{
\ddot{\cT}_i^{(k)}
\right\}_{k\in[K]}
\right\}_{i\in[m]}, |\ddot{\Theta}'|=m.
$$
For each $i\in[m]$, we have:
$$
\ddot{\cT}_i^{(k)}=\cC^{(k)}\times_1 U_1^{(k)}\times_2
U_2^{(k)}\times_3\cdots\times_{L_k} U_{L_k}^{(k)}\times_{L_k+1}[V,W^{(k)}]=\cT^{(k)}\ \text{for any}\ k\in[K]\backslash\{l\},
$$
\begin{equation*}
\begin{aligned}
\ddot{\cT}_i^{(l)}
&=\cC^{(l)}\times_1 U_1^{(l)}\times_2
U_2^{(l)}\times_3\cdots\times_{L_l} U_{L_l}^{(l)}\times_{L_l+1}[V,\ddot{W}_i^{(l)}]\\
&:=
\cC^{(l)}\times_1 U_1^{(l)}\times_2
U_2^{(l)}\times_3\cdots\times_{L_l} U_{L_l}^{(l)}\times_{L_l+1}[V,V^\perp\ddot{\Xi}_i^{(l)}],
\end{aligned}
\end{equation*}
where $\{\ddot{\Xi}_i^{(l)}\}_i$ consists of $\ddot{\cO}(n-r_s,r_l-r_s)$. Thus, according to the definition of  $\Theta'$ \eqref{Theta'}, we know that $\{
\ddot{\cT}_i^{(k)}
\}_{k\in[K]}$ and $\ddot{\Xi}_i^{(l)}$ hold a one-to-one correspondence as 
$i$ varies. Thus, based on the lower bound for $|\ddot{\cO}(n-r_s,r_l-r_s)|$, we have that $|\ddot{\Theta}'|=m\geq2^{(r_l-r_s)(n +r_s - 2r_l)}$. Further, according to the isomorphism property of $\Phi$, we know that \eqref{eq:the5_Wk_rho(Xi1-Xi2)} and \eqref{eq:the5_Wk_Xi1-Xi2_F} also hold for $(\ddot{\Theta}',\rho')$.
Thus, $\ddot{\Theta}'$ is the ideal local packing set of $\Theta'$ for our subsequent discussion.

$\\$
\textbf{\textit{Kullback-Leibler (KL) distance:}}

Denote $\{
\ddot{\cT}_1^{(k)}
\}_{k\in[K]}$ and $\{
\ddot{\cT}_2^{(k)}
\}_{k\in[K]}\in\ddot{\Theta}'$ as two sets of different parameters. Moreover, we define $\{
\ddot{\cX}_1^{(k)}
\}_{k\in[K]}$ and $\{
\ddot{\cX}_2^{(k)}
\}_{k\in[K]}$ as two sets of tensors generated by $\{
\ddot{\cT}_1^{(k)}
\}_{k\in[K]}$ and $\{
\ddot{\cT}_2^{(k)}
\}_{k\in[K]}$, respectively, i.e.,
$$
\ddot{\cX}_1^{(k)}=\ddot{\cT}_1^{(k)}+\cE^{(k)},\ \text{and}\ \ddot{\cX}_2^{(k)}=\ddot{\cT}_2^{(k)}+\tilde{\cE}^{(k)}\ \text{for each}\ k\in[K],
$$
where $\{\cE^{(k)},\tilde{\cE}^{(k)}\}_{k\in[K]}$ are all i.i.d. noise tensors with each element following a standard normal distribution.
Furthermore, we write $\P_1$ and $\P_2$ as the population distributions of $\{
\ddot{\cX}_1^{(k)}
\}_{k\in[K]}$ and $\{
\ddot{\cX}_2^{(k)}
\}_{k\in[K]}$. We also write $\P_1^{(k)}$ and $\P_2^{(k)}$ as the population distributions for each $
\ddot{\cX}_1^{(k)}$ and $\ddot{\cX}_2^{(k)}$, respectively. By the Gaussian likelihood formula, we have:
$$
\vec\left(\ddot{\cX}_i^{(k)}\right)\sim
N\left(\vec\left(\ddot{\cT}_i^{(k)}\right),I_{np_1^{(k)}p_2^{(k)}\cdots p_{L_k}^{(k)}}\right),\ \text{for}\ i=1,2.
$$
Then, we can acquire the following KL-distance equality:
$$
KL\left(\P_1||\P_2\right)
\underset{K \text{tensors}}{\overset{\text{independency of} \{\cX^{(k)}\}_k}{=}}
\sum_{k=1}^K
KL\left(\P_1^{(k)}||\P_2^{(k)}\right)
\underset{\text{for}\  k\in[K]\backslash \{l\}}{\overset{\ddot{\cT}_1^{(k)}=\ddot{\cT}_2^{(k)}}{=}}
KL\left(\P_1^{(l)}||\P_2^{(l)}\right).
$$
For the $l$-th tensor, we have:
\begin{equation}
\label{eq:the5_Wk_KL}
\begin{aligned}
KL\left(\P_1^{(l)}||\P_2^{(l)}\right)
&\underset{\text{in \cite{zhang2018tensor}}}
{\overset{\text{the bottom of page 34}}{=}}
\frac{1}{2}
\left\|\ddot{\cT}_1^{(l)}-\ddot{\cT}_2^{(l)}\right\|_F^2\\
&=
\frac{1}{2}
\left\|\cC^{(l)}\times_1 U_1^{(l)}\times_2
U_2^{(l)}\times_3\cdots\times_{L_l} U_{L_l}^{(l)}\times_{L_l+1}[0_{n\times r_s},V^\perp(\ddot{\Xi}_1^{(l)}-\ddot{\Xi}_2^{(l)})]\right\|_F^2\\
&=
\frac{1}{2}
\left\|
[0_{n\times r_s},V^\perp(\ddot{\Xi}_1^{(l)}-\ddot{\Xi}_2^{(l)})]\cdot\cM_{L_l+1}(\cC^{(l)})\cdot \left(U_{\otimes}^{(l)}\right)^\top
\right\|_F^2\\
&\leq \tilde{C}_1\lambda_{\max}^2\left\|V^\perp\left(\ddot{\Xi}_1^{(l)}-\ddot{\Xi}_2^{(l)}\right)\right\|_F^2
\leq
\tilde{C}_2\lambda_{\min}^2\left\|\left(\ddot{\Xi}_1^{(l)}-\ddot{\Xi}_2^{(l)}\right)\right\|_F^2\\
&\overset{\text{\eqref{eq:the5_Wk_Xi1-Xi2_F}}}{\leq}
2\tilde{C}_2\lambda_{\min}^2\delta(r_l-r_s),
\end{aligned}
\end{equation}
where $U_{\otimes}^{(l)}$ is defined below \eqref{Ek}. The first inequality on the second-to-last line of \eqref{eq:the5_Wk_KL} holds from the definition of $\lambda_{\max}$, while the second inequality holds due to the signal constraint in $\Theta'$ from \eqref{Theta'}. 

$\\$
\textbf{\textit{Fano's Lemma:}}

Combining \eqref{eq:the5_Wk_rho(Xi1-Xi2)}, \eqref{eq:the5_Wk_Xi1-Xi2_F}, and \eqref{eq:the5_Wk_KL} with  the well-known generalized Fano’s lemma, we yield the following lower bound:
\begin{equation*}
\begin{aligned}
&\quad
\inf\limits_{\hat{W}^{(l)}}\sup\limits_{\{\cT^{(k)}\}_k\in\Theta'}\mathbb{P}\left\{\left\|\hat{W}^{(l)}(\hat{W}^{(l)})^{\top}-W^{(l)}(W^{(l)})^{\top}\right\|_F\geq \frac{c}{2} \sqrt{(r_l-r_s) \delta(1-\delta)}\right\}\\
&=\inf\limits_{\hat{W}^{(l)}}\sup\limits_{\{\cT^{(k)}\}_k\in\Theta'}\mathbb{P}\left\{\left\|\hat{W}^{(l)}(\hat{W}^{(l)})^{\top}-W^{(l)}(W^{(l)})^{\top}\right\|_F^2\geq \frac{c^2}{4} (r_l-r_s) \delta(1-\delta)\right\}\\
&\overset{\text{Fano}}{\geq}
1-\frac{\max_{i\neq j}KL\left(\P_i||\P_j\right)+\log2}{\log m}
\geq
1-\frac{\tilde{C}\lambda_{\min}^2\delta(r_l-r_s)+\log2}{(r_l-r_s)(n +r_s - 2r_l)\log 2}.
\end{aligned}
\end{equation*}
By utilizing the Markov inequality, we have:
\begin{equation*}
\begin{aligned}
&\quad\inf\limits_{\hat{W}^{(l)}}\sup\limits_{\{\cT^{(k)}\}_k\in\Theta'}\E\left\|\hat{W}^{(k)}(\hat{W}^{(k)})^{\top}-W^{(k)}(W^{(k)})^{\top}\right\|_F^2\\
&\gtrsim
(r_l-r_s) \delta(1-\delta)
\left(1-\frac{\tilde{C}\lambda_{\min}^2\delta(r_l-r_s)+\log2}{(r_l-r_s)(n +r_s - 2r_l)\log 2}\right).
\end{aligned}
\end{equation*}
By setting $\delta=c_1(n +r_s - 2r_l)/\lambda_{\min}^2$ for a sufficiently small constant $c_1$, we have:
\begin{equation*}
\begin{aligned}
&\quad1-\frac{\tilde{C}\lambda_{\min}^2\delta(r_l-r_s)+\log2}{(r_l-r_s)(n +r_s - 2r_l)\log 2}
=
1-\frac{\tilde{C}c_1(r_l-r_s)(n +r_s - 2r_l)+\log2}{(r_l-r_s)(n +r_s - 2r_l)\log 2}\\
&=
1-\frac{c_1\tilde{C}}{\log 2}-\frac{1}{(r_l-r_s)(n +r_s - 2r_l)}.
\end{aligned}
\end{equation*}
As long as $n$ is sufficiently large while  $4\tilde{C}c_1\leq\log 2$, we have 
$$
\frac{c_1\tilde{C}}{\log 2}+\frac{1}{(r_l-r_s)(n +r_s - 2r_l)}
\leq
\frac{1}{2},
$$
and further:
$$
1-\frac{\tilde{C}\lambda_{\min}^2\delta(r_l-r_s)+\log2}{(r_l-r_s)(n +r_s - 2r_l)\log 2}\geq \frac{1}{2}.
$$
Similarly, as long as $n/\lambda_{\min}^2$ is appropriately small, we have $\delta\leq1/2$, which further gives us $1-\delta\geq1/2$. Combining the above results gives us:
$$
\inf\limits_{\hat{W}^{(l)}}\sup\limits_{\{\cT^{(k)}\}_k\in\Theta'}\E\left\|\hat{W}^{(l)}(\hat{W}^{(l)})^{\top}-W^{(l)}(W^{(l)})^{\top}\right\|_F^2
\gtrsim
\frac{(r_l-r_s)(n +r_s - 2r_l)}{\lambda_{\min}^2}.
$$
By utilizing Assumption \ref{assum:4} and  $\lambda_{\min}=\sqrt{n}\sigma_{\min}$, we can abbreviate the right hand side of the above inequality as $1/\sigma_{\min}^2$, which completes the proof of \eqref{minimax for Wk}.

\subsubsection{Proof of the first claim in Theorem \ref{theorem:minimax}}
Similar to the second claim, we provide the specific definition for $\Theta''$. For a pre-known set of hyper-parameters $\{\{p^{(k)}, r^{(k)},L_k\}_{k\in[K]};r_s;n;K;\sigma_{\max};\sigma_{\min}\}$, we define our parameter space $\Theta''\subseteq\Theta$ as:
\begin{equation}
\label{Theta''}
\begin{aligned}
\Theta'':=\biggl\{
&\left\{\cT^{(k)}=\cC^{(k)}\times_1 U_1^{(k)}\times_2
U_2^{(k)}\times_3\cdots\times_{L_k} U_{L_k}^{(k)}\times_{L_k+1}[V,W^{(k)}]\right\}_k\Big|\ \sigma_{\max}\leq C\sigma_{\min}, \\
&\sigma_{\min}\leq\min_{k\in[K]}\{\sigma_{r_k}(\cM_{L_k+1}(\cC^{(k)})/\sqrt{n})\}\leq\cdots\leq\max_{k\in[K]}\{\sigma_{1}(\cM_{L_k+1}(\cC^{(k)})/\sqrt{n})\}\leq\sigma_{\max};\\
& V^\top V=I_{r_s}, \{\{(U_j^{(k)})^\top U_j^{(k)}=I_{r_j^{(k)}}\}_{j=1}^{L_k}, (W^{(k)})^\top W^{(k)}=I_{r_k-r_s}, (W^{(k)})^\top V=0\}_{k=1}^{K};\\
& \{\cC^{(k)}, \{U_j^{(k)}\}_{j=1}^{L_k}, W^{(k)}\}_{k=1}^{K}\ \text{are fixed parameters}
\biggr\}.
\end{aligned}
\end{equation}
To prove \eqref{minimax for V}, we only need to obtain
\begin{equation*}
\inf_{\hat{V}}\sup_{\{\cT^{(k)}\}_k\in\Theta''}\E\left\|\hat{V}\hat{V}^{\top}-VV^{\top}\right\|^2_F
\gtrsim
\frac{1}{K\sigma_{\min}^2}.
\end{equation*}
$\\$
\textbf{\textit{Discussion for $\Theta''$:}}

Define $r_p=\max_{k\in[K]}\{r_k-r_s\}$.
Note that for any $\{\cT^{(k)}\}_k\in\Theta''$, the only changing parameter is $V$. Thus, we consider the metric space $(\Theta'',\rho'')$ where $\rho''$ is defined as 
$$
\rho''\left(\{\cT_k\}_k,\{\breve{\cT}_k\}_k\right)
:=
\left\|VV^\top-\breve{V}\breve{V}^\top\right\|_F,
$$
where $V$ and $\breve{V}$ are the corresponding changing parameters for $\{\cT_k\}_k$ and $\{\breve{\cT}_k\}_k$ respectively. 
We concatenate $\{W^{(k)}\}_{k\in[K]}$ by column to acquire 
$$
W:=\left[W^{(k)}, k\in[K]\right]\in\mathbb{R}^{n\times r'},
$$
where $r'=\sum_{k\in[K]} (r_k-r_s)$.
Define the vector space spanned by the columns of $W$ as $\span(W)$ where $\text{dim}(\span(W))=r_p'$ and $r_p\leq r_p'\leq r'$ hold. Assumption \ref{assum:4} guaranties $r'< n$ and ensures that $n-r'$ remains strictly positive throughout the subsequent discussion.
We can also deduce from \eqref{Theta''} that $W\perp V$. If we denote the orthogonal complement of $W$ as $W^\perp\in\mathbb{R}^{n\times(n-r_p')}$, then there exists a $\Xi\in\mathbb{R}^{(n-r_p')\times r_s}$ such that the corresponding $V$ satisfies $V=W^\perp\Xi$. As $(W^\perp)^\top W^\perp=I_{n-r_p'}$ holds, we have:
$$
\Xi^\top\Xi
=\Xi^\top \left(W^\perp\right)^\top W^\perp\Xi
=
V^\top V=I_{r_s}.
$$
Accordingly, we have $\Xi\in\cO(n-r_p',r_s)=\{U\in\mathbb{R}^{(n-r_p')\times r_s}: U^\top U=I_{r_s}\}$. In contrast, for any $\Xi$ such that  $\Xi\in\cO(n-r_p',r_s)$, if we let $V$ be $V=W^\perp\Xi$, then we can deduce that $V\in\Theta''$, or equivalently the corresponding $\{\cT^{(k)}\}_k\in\Theta''$, holds. Define $\Phi:\cO(n-r_p',r_s)\rightarrow\Theta''$ as:
$$
\Phi(\Xi)=W^\perp\Xi=V.
$$
Moreover, for any two distinct $\{\cT^{(k)}\}_k$, $\{\breve{\cT}^{(k)}\}_k\in\Theta''$, we have:
\begin{equation*}
\begin{aligned}
\rho_{\cO}\left(\Xi,\breve{\Xi}\right)
&=
\left\|\Xi\Xi^\top
-
\breve{\Xi}\breve{\Xi}^\top\right\|_F
=
\left\|W^\perp\left[\Xi\Xi^\top
-
\breve{\Xi}\breve{\Xi}^\top\right]\left(W^\perp\right)^\top\right\|_F\\
&=
\left\|\left(W^\perp\Xi\right)\left(W^\perp\Xi\right)^\top
-
\left(W^\perp\breve{\Xi}\right)\left(W^\perp\breve{\Xi}\right)^\top\right\|_F\\
&=
\left\|VV^\top-\breve{V}\breve{V}^\top\right\|_F\\
&=\rho''\left(\{\cT_k\}_k,\{\breve{\cT}_k\}_k\right).
\end{aligned}
\end{equation*}
Accordingly, $\Phi$ is an isometric isomorphism map between $(\Theta'',\rho'')$ and $(\cO(n-r_p',r_s),\rho_\cO)$, i.e., $(\Theta'',\rho'')\cong(\cO(n-r_p',r_s),\rho_\cO)$. Therefore, we can again transform our attention from $\Theta''$ to $\cO(n-r_p',r_s)$ and construct a local packing set of $\cO(n-r_p',r_s)$.
Now, we partition $\Theta''$ into disjoint subsets $\{\Theta''_r\}_{r\in[r']\backslash[r_p]}$:
\begin{equation}
\label{Theta_r''}
\begin{aligned}
\Theta_r'':=\biggl\{
&\left\{\cT^{(k)}=\cC^{(k)}\times_1 U_1^{(k)}\times_2
U_2^{(k)}\times_3\cdots\times_{L_k} U_{L_k}^{(k)}\times_{L_k+1}[V,W^{(k)}]\right\}_k\Big|\ \sigma_{\max}\leq C\sigma_{\min}, \\
&\sigma_{\min}\leq\min_{k\in[K]}\{\sigma_{r_k}(\cM_{L_k+1}(\cC^{(k)})/\sqrt{n})\}\leq\cdots\leq\max_{k\in[K]}\{\sigma_{1}(\cM_{L_k+1}(\cC^{(k)})/\sqrt{n})\}\leq\sigma_{\max}\\
& V^\top V=I_{r_s}, \{\{(U_j^{(k)})^\top U_j^{(k)}=I_{r_j^{(k)}}\}_{j=1}^{L_k}, (W^{(k)})^\top W^{(k)}=I_{r_k-r_s}, (W^{(k)})^\top V=0\}_{k=1}^{K};\\
& \dim(W)=r_p'=r; \{\cC^{(k)}, \{U_j^{(k)}\}_{j=1}^{L_k}, W^{(k)}\}_{k=1}^{K}\ \text{are fixed parameters}
\biggr\}\subseteq\Theta''\subseteq\Theta.
\end{aligned}
\end{equation}
For each hyper-parameter $r$, we can also build an isometric isomorphism map $\Phi_r$ between $(\Theta_r'',\rho'')$ and $(\cO(n-r,r_s),\rho_\cO)$ based on the same discussion as above, i.e., $(\Theta_r'',\rho'')\cong(\cO(n-r,r_s),\rho_\cO)$ also holds. If we intend to build a desired lower bound for $\Theta''$,  we can delve into each $\Theta_r''$ and use $\Phi_r$ to bridge $\Theta_r''$ and $\cO(n-r,r_s)$. By  constructing the corresponding $\epsilon$-local packing set of $\cO(n-r,r_s)$ and utilizing the isometric isomorphism $\Phi_r$, we can directly transplant the geometric properties of the local packing set of $\cO(n-r,r_s)$ into  $\Theta_r''$. Then, the lower bound of $\Theta''$ can be simply obtained by further lower bounding  $\Theta_r''$, as the following holds:
\begin{equation*}
\inf_{\hat{V}}\sup_{\{\cT^{(k)}\}_k\in\Theta''}\E\left\|\hat{V}\hat{V}^{\top}-VV^{\top}\right\|^2_F
\gtrsim
\inf_{\hat{V}}\sup_{\{\cT^{(k)}\}_k\in\Theta_{r}''}\E\left\|\hat{V}\hat{V}^{\top}-VV^{\top}\right\|^2_F.
\end{equation*}
The subsequent discussion narrows down to each $\Theta_r''$ and selects the corresponding $r$ that owns the largest lower bound.

$\\$
\textbf{\textit{Local packing for $\Theta_r''$:}}

Fix  $r\in[r']\backslash[r_p-1]$ and let the $\epsilon$ value be identical among each $\Theta_r''$. Similarly, we consider $\cO(n-r-r_s, r_s)$ first. For any $\epsilon>0$, $0\leq \min\{r_s,r\}\leq \max\{r_s,r\}\leq n$, and $n-r\geq r_s$, we can again utilize Proposition \ref{prop:packing} to acquire the following result by taking $q=2$:
\[
\left( \frac{c}{\epsilon} \right)^{r_s(n - r - 2r_s)}
\leq 
D(\cO(n-r-r_s,r_s), \rho_{\cO}, \epsilon \sqrt{r_s}) 
\leq 
\left( \frac{C}{\epsilon} \right)^{r_s(n - r - 2r_s)}.
\]
Based on the definition of $\epsilon$-packing number, there exists a  $\dot{\cO}(n-r-r_s,r_s)\subseteq\cO(n-r-r_s,r_s)$ such that $|\dot{\cO}(n-r-r_s,r_s)|\geq\left( c/\epsilon \right)^{r_s(n - r - 2r_s)}$ while for any two distinct $\dot{\Xi}_{r,1}$, $\dot{\Xi}_{r,2}\in\dot{\cO}(n-r-r_s,r_s)$, we have:
$$
\rho_{\cO}\left(\dot{\Xi}_{r,1},\dot{\Xi}_{r,2}\right)
=
\left\|\dot{\Xi}_{r,1}\dot{\Xi}_{r,1}^\top
-
\dot{\Xi}_{r,2}\dot{\Xi}_{r,2}^\top\right\|_F
\geq
\sqrt{r_s}\epsilon.
$$
By setting $\epsilon=c/2$, which is identical for every $\Theta_r''$, we have:
$$
\left|\dot{\cO}(n-r-r_s,r_s)\right|
\geq
2^{r_s(n - r - 2r_s)},
\ \text{and}\ \rho_{\cO}\left(\dot{\Xi}_{r,1},\dot{\Xi}_{r,2}\right)
\geq
\frac{c\sqrt{r_s}}{2}.
$$
For a fixed $\delta_r>0$ and every $\dot{\Xi}_r\in\dot{\cO}(n-r-r_s,r_s)$, we define:
$$\ddot{\Xi}_r = \begin{pmatrix} \sqrt{1 - \delta_r} I_{r_s} \\ \sqrt{\delta_r} \, \dot{\Xi}_r \end{pmatrix} \in \mathbb{R}^{(n-r) \times r_s}.$$
Then we have:
$\ddot{\Xi}_r^\top \ddot{\Xi}_r=(1-\delta_r)I_{r_s}+\delta_r\dot{\Xi}_r^\top \dot{\Xi}_r=I_{r_s}$, which implies that $\ddot{\Xi}_r\in\cO(n-r,r_s)$ also holds. If we collect all $\ddot{\Xi}_r$ as $\ddot{\cO}(n-r,r_s)$, then we have
$$
\left|\ddot{\cO}(n-r,r_s)\right|
=
\left|\dot{\cO}(n-r-r_s,r_s)\right|
\geq
2^{r_s(n - r - 2r_s)}.
$$
Accordingly, for any two distinct $\ddot{\Xi}_{r,1}$, $\ddot{\Xi}_{r,2}\in\ddot{\cO}(n-r,r_s)$, we have:
$$
\ddot{\Xi}_{r,1}\ddot{\Xi}_{r,1}^\top 
= \begin{pmatrix} \sqrt{1-\delta_r} I_{r_s} \\ \sqrt{\delta_r} \dot{\Xi}_{r,1} \end{pmatrix} 
\left( \sqrt{1-\delta_r} I_{r_s},\ \sqrt{\delta_r} \dot{\Xi}_{r,1}^\top \right) = \begin{pmatrix} (1-\delta_r) I_{r_s} & \sqrt{\delta_r(1-\delta_r)} \dot{\Xi}_{r,1}^\top \\ \sqrt{\delta_r(1-\delta_r)} \dot{\Xi}_{r,1} & \delta_r \dot{\Xi}_{r,1}\dot{\Xi}_{r,1}^\top
\end{pmatrix}.
$$
Similar discussion for $\ddot{\Xi}_{r,2}$ leads us to:
$$
\left\|\ddot{\Xi}_{r,1}\ddot{\Xi}_{r,1}^\top 
-\ddot{\Xi}_{r,2}\ddot{\Xi}_{r,2}^\top \right\|_F
=
\left\|
\begin{pmatrix} 
0 & \sqrt{\delta_r(1-\delta_r)} \left(\dot{\Xi}_{r,1}-\dot{\Xi}_{r,2}\right)^\top \\ \sqrt{\delta_r(1-\delta_r)} \left(\dot{\Xi}_{r,1}-\dot{\Xi}_{r,2}\right) & \delta_r \left(\dot{\Xi}_{r,1}\dot{\Xi}_{r,1}^\top
-
\dot{\Xi}_{r,2}\dot{\Xi}_{r,2}^\top\right) \end{pmatrix}
\right\|_F.
$$
Based on the definition of $\|\cdot\|_F$ and the triangle inequality, we have:
$$
\rho_{\cO}\left(\ddot{\Xi}_{r,1},\ddot{\Xi}_{r,2}\right)
\geq
\sqrt{\delta_r(1-\delta_r)} \left\|\dot{\Xi}_{r,1}-\dot{\Xi}_{r,2}\right\|_F
\geq
\frac{\sqrt{\delta_r(1-\delta_r)}}{2}
\left\|\dot{\Xi}_{r,1}\dot{\Xi}_{r,1}^\top
-
\dot{\Xi}_{r,2}\dot{\Xi}_{r,2}^\top\right\|_F.
$$
Based on the cardinality lower bound of $\dot{\cO}(n-r-r_s,r_s)$, we can further deduce that:
\begin{equation}
\label{eq:the5_V_rho(Xi1-Xi2)}
\rho_{\cO}\left(\ddot{\Xi}_{r,1},\ddot{\Xi}_{r,2}\right)
\geq
\frac{\sqrt{\delta_r(1-\delta_r)}}{2}
\left\|\dot{\Xi}_{r,1}\dot{\Xi}_{r,1}^\top
-
\dot{\Xi}_{r,2}\dot{\Xi}_{r,2}^\top\right\|_F
\geq
\frac{\tilde{c}\sqrt{r_s\delta_r(1-\delta_r)}}{2}.
\end{equation}
Besides, we also have:
\begin{equation}
\label{eq:the5_V_Xi1-Xi2_F}
\left\|\ddot{\Xi}_{r,1}
-\ddot{\Xi}_{r,2}\right\|_F
\leq
\sqrt{\delta_r}
\left\|\dot{\Xi}_{r,1}
-\dot{\Xi}_{r,2}\right\|_F
\leq
\sqrt{2\delta_r r_s}.
\end{equation}
Now we return to the discussion of $\Theta_r''$.
Based on $\Phi_r$, we can build a subset of $\Theta_r''$ and $\ddot{\Theta}_r''\subseteq\Theta_r''$, such that:
$$
\ddot{\Theta}_r''
:=
\left\{
\left\{
\ddot{\cT}_{r,i}^{(k)}
\right\}_{k\in[K]}
\right\}_{i\in[m_r]}, |\ddot{\Theta}_r''|=m_r.
$$
For each $k\in[K]$ and $i\in[m_r]$, we have:
\begin{equation*}
\begin{aligned}
\ddot{\cT}_{r,i}^{(k)}
&=\cC^{(k)}\times_1 U_1^{(k)}\times_2
U_2^{(k)}\times_3\cdots\times_{L_k} U_{L_k}^{(k)}\times_{L_k+1}[\ddot{V}_{r,i},W^{(k)}]\\
&:=
\cC^{(k)}\times_1 U_1^{(k)}\times_2
U_2^{(k)}\times_3\cdots\times_{L_k} U_{L_k}^{(k)}\times_{L_k+1}[W^\perp \ddot{\Xi}_{r,i},W^{(k)}],
\end{aligned}
\end{equation*}
where $\{\ddot{\Xi}_{r,i}\}_i$ consists $\ddot{\cO}(n-r,r_s)$. Thus, according to the definition in $\Theta_r''$ \eqref{Theta_r''}, we know that $\{
\ddot{\cT}_{r,i}^{(k)}
\}_{k\in[K]}$ and $\ddot{\Xi}_{r,i}$ hold a one-to-one correspondence as 
$i$ varies. Accordingly, based on the lower bound for $|\ddot{\cO}(n-r,r_s)|$, we have that $|\ddot{\Theta}_r''|=m_r\geq2^{r_s(n -r - 2r_s)}$ and each of the two distinct $\{
\ddot{\cT}_{r,1}^{(k)}
\}_{k\in[K]}$, $\{
\ddot{\cT}_{r,2}^{(k)}
\}_{k\in[K]}\in\ddot{\Theta}_r''$ satisfy \eqref{eq:the5_V_rho(Xi1-Xi2)} and \eqref{eq:the5_V_Xi1-Xi2_F}. Till now, we have constructed the ideal local packing set, $\ddot{\Theta}_r''$, for $\Theta_r''$.

$\\$
\textbf{\textit{Kullback-Leibler (KL) distance for $\Theta_r''$:}}

Recall that $\{
\ddot{\cT}_{r,1}^{(k)}
\}_{k\in[K]}$, $\{
\ddot{\cT}_{r,2}^{(k)}
\}_{k\in[K]}\in\ddot{\Theta}_r''$ are two sets of different parameters in $\ddot{\Theta}_r''$. Moreover, we define $\{
\ddot{\cX}_{r,1}^{(k)}
\}_{k\in[K]}$ and $\{
\ddot{\cX}_{r,2}^{(k)}
\}_{k\in[K]}$ as two sets of tensors generated by $\{
\ddot{\cT}_{r,1}^{(k)}
\}_{k\in[K]}$, $\{
\ddot{\cT}_{r,2}^{(k)}
\}_{k\in[K]}$, respectively, i.e.,
$$
\ddot{\cX}_{r,1}^{(k)}=\ddot{\cT}_{r,1}^{(k)}+\cE^{(k)},\ \text{and}\ \ddot{\cX}_{r,2}^{(k)}=\ddot{\cT}_{r,2}^{(k)}+\tilde{\cE}^{(k)}\ \text{for each}\ k\in[K],
$$
where $\{\cE^{(k)},\tilde{\cE}^{(k)}\}_{k\in[K]}$ are all i.i.d. noise tensors with each element following a standard normal distribution.
Further, we write $\P_{r,1}$ and $\P_{r,2}$ as the population distributions of $\{
\ddot{\cX}_{r,1}^{(k)}
\}_{k\in[K]}$ and $\{
\ddot{\cX}_{r,2}^{(k)}
\}_{k\in[K]}$. We also write $\P_{r,1}^{(k)}$ and $\P_{r,2}^{(k)}$ as the population distributions for each $
\ddot{\cX}_{r,1}^{(k)}$ and $\ddot{\cX}_{r,2}^{(k)}$, respectively. Clearly, based on the simple normal distribution property, we have:
$$
\vec\left(\ddot{\cX}_{r,i}^{(k)}\right)\sim
N\left(\vec\left(\ddot{\cT}_{r,i}^{(k)}\right),I_{np_1^{(k)}p_2^{(k)}\cdots p_{L_k}^{(k)}}\right),\ \text{for}\ i=1,2.
$$
Then, the KL-distance of $\P_{r,1}$ and $\P_{r,2}$ can be deduced by using the independence between $K$ tensors as follows:
$$
KL\left(\P_{r,1}||\P_{r,2}\right)
=
\sum_{k=1}^K
KL\left(\P_{r,1}^{(k)}||\P_{r,2}^{(k)}\right).
$$
For each $k\in[K]$, we know that:
\begin{equation}
\label{eq:the5_V_KL}
\begin{aligned}
KL\left(\P_{r,1}^{(k)}||\P_{r,2}^{(k)}\right)
&\underset{\text{in \cite{zhang2018tensor}}}
{\overset{\text{the bottom of page 34}}{=}}
\frac{1}{2}
\left\|\ddot{\cT}_{r,1}^{(k)}-\ddot{\cT}_{r,2}^{(k)}\right\|_F^2\\
&=
\frac{1}{2}
\left\|\cC^{(k)}\times_1 U_1^{(k)}\times_2
\cdots\times_{L_k} U_{L_k}^{(k)}\times_{L_k+1}[W^\perp(\ddot{\Xi}_{r,1}-\ddot{\Xi}_{r,2}),0_{n\times(r_k-r_s)}]\right\|_F^2\\
&=
\frac{1}{2}
\left\|
[W^\perp(\ddot{\Xi}_{r,1}-\ddot{\Xi}_{r,2}),0_{n\times(r_k-r_s)}]\cdot\cM_{L_k+1}(\cC^{(k)})\cdot \left(U_{\otimes}^{(k)}\right)^\top
\right\|_F^2\\
&\leq \tilde{C}_1\lambda_{\max}^2\left\|W^\perp\left(\ddot{\Xi}_{r,1}-\ddot{\Xi}_{r,2}\right)\right\|_F^2\\
&\leq
\tilde{C}_2\lambda_{\min}^2\left\|\ddot{\Xi}_{r,1}-\ddot{\Xi}_{r,2}\right\|_F^2\\
&\overset{\eqref{eq:the5_V_Xi1-Xi2_F}}{\leq}
2\tilde{C}_2\lambda_{\min}^2\delta_r r_s,
\end{aligned}
\end{equation}
where $U_{\otimes}^{(k)}$ is defined below \eqref{Ek} and the first inequality holds by the definition of $\lambda_{\max}$. The second inequality holds due to the signal constraint in $\Theta''$ from \eqref{Theta''}. Further, we have:
$$
KL\left(\P_{r,1}||\P_{r,2}\right)
\leq
2\tilde{C}_2K\lambda_{\min}^2\delta_r r_s
$$

$\\$
\textbf{\textit{Lower bound for $\Theta_r''$:}}

Combining the above discussion with  the well known generalized Fano’s lemma, we yield the following lower bound:
\begin{equation*}
\begin{aligned}
&\quad
\inf\limits_{\hat{V}}\sup\limits_{\{\cT^{(k)}\}_k\in\Theta_r''}\mathbb{P}\left\{\left\|\hat{V}\hat{V}^{\top}-VV^{\top}\right\|_F\geq \frac{c}{2} \sqrt{r_s \delta_r(1-\delta_r)}\right\}\\
&=\inf\limits_{\hat{V}}\sup\limits_{\{\cT^{(k)}\}_k\in\Theta_r''}\mathbb{P}\left\{\left\|\hat{V}\hat{V}^{\top}-VV^{\top}\right\|_F^2\geq \frac{c^2}{4} r_s \delta_r(1-\delta_r)\right\}\\
&\overset{\text{Fano}}{\geq}
1-\frac{\max_{i\neq j}KL\left(\P_{r,i}||\P_{r,j}\right)+\log2}{\log m_r}
\geq
1-\frac{\tilde{C}\lambda_{\min}^2K\delta_r r_s+\log2}{r_s(n - r - 2r_s)\log 2}.
\end{aligned}
\end{equation*}
By utilizing the Markov inequality, we have:
\begin{equation*}
\begin{aligned}
&\quad\inf_{\hat{V}}\sup_{\{\cT^{(k)}\}_k\in\Theta_r''}\E\left\|\hat{V}\hat{V}^{\top}-VV^{\top}\right\|^2_F
\gtrsim
r_s \delta_r(1-\delta_r)
\left(1-\frac{\tilde{C}K\lambda_{\min}^2\delta_r r_s+\log2}{r_s(n - r - 2r_s)\log 2}\right).
\end{aligned}
\end{equation*}
By setting $\delta_r=c_1(n -r - 2r_s)/K\lambda_{\min}^2$ for a sufficiently small constant $c_1$, we have:
\begin{equation*}
\begin{aligned}
&\quad1-\frac{\tilde{C}K\lambda_{\min}^2\delta_r r_s+\log2}{r_s(n -r - 2r_s)\log 2}
=
1-\frac{\tilde{C}c_1r_s(n -r - 2r_s)+\log2}{r_s(n -r - 2r_s)\log 2}=
1-\frac{c_1\tilde{C}}{\log 2}-\frac{1}{r_s(n -r - 2r_s)}.
\end{aligned}
\end{equation*}
According to the similar discussion as $\hat{W}^{(l)}$, we have:
\begin{equation*}
\inf_{\hat{V}}\sup_{\{\cT^{(k)}\}_k\in\Theta_r''}\E\left\|\hat{V}\hat{V}^{\top}-VV^{\top}\right\|^2_F
\gtrsim
\frac{r_s(n-r-2r_s)}{K\lambda_{\min}^2}.
\end{equation*}

$\\$
\textbf{\textit{Lower bound for $\Theta''$ and $\Theta$:}}

Again, by noting the fact that $\Theta_r''\subseteq\Theta''$ for each $r\in[r']\backslash[r_p-1]$, we have
\begin{equation*}
\inf_{\hat{V}}\sup_{\{\cT^{(k)}\}_k\in\Theta''}\E\left\|\hat{V}\hat{V}^{\top}-VV^{\top}\right\|^2_F
\geq
\inf_{\hat{V}}\sup_{\{\cT^{(k)}\}_k\in\Theta_r''}\E\left\|\hat{V}\hat{V}^{\top}-VV^{\top}\right\|^2_F.
\end{equation*}
Accordingly, we can acquire the following result by taking the maximum among $r$:
\begin{equation*}
\begin{aligned}
\inf_{\hat{V}}\sup_{\{\cT^{(k)}\}_k\in\Theta''}\E\left\|\hat{V}\hat{V}^{\top}-VV^{\top}\right\|^2_F
&\geq
\max_{r\in[r']\backslash[r_p]}\left\{\inf_{\hat{V}}\sup_{\{\cT^{(k)}\}_k\in\Theta_r''}\E\left\|\hat{V}\hat{V}^{\top}-VV^{\top}\right\|^2_F\right\}\\
&\gtrsim\frac{r_s(n-r_p-2r_s)}{K\lambda_{\min}^2}
\end{aligned}
\end{equation*}
where $r_p=\max\{r_k-r_s\}$. By realizing the $\Theta''\subseteq\Theta$, i.e., the following inequality holds
\begin{equation*}
\inf_{\hat{V}}\sup_{\{\cT^{(k)}\}_k\in\Theta}\E\left\|\hat{V}\hat{V}^{\top}-VV^{\top}\right\|^2_F
\geq
\inf_{\hat{V}}\sup_{\{\cT^{(k)}\}_k\in\Theta''}\E\left\|\hat{V}\hat{V}^{\top}-VV^{\top}\right\|^2_F
\gtrsim
\frac{r_s(n-r_p-2r_s)}{K\lambda_{\min}^2}.
\end{equation*}
We can complete the proof of the first claim of Theorem \ref{theorem:minimax} by again using $\lambda_{\min}=\sqrt{n}\sigma_{\min}$ to abbreviate the right hand side of the above inequality as $1/K\sigma_{\min}^2$.

\subsection{Proof of Corollary \ref{cor:rank consistency}}
\label{supp:rank consistency}
Since $V^\top W^{(k)}=0$ for every $k\in[K]$, we have
$P_VP_W^{(k)}
=
P_W^{(k)}P_V
=
0_{n\times n}.$
Consequently,
\[
\begin{aligned}
\Sigma_0
=
P_V+
P_V^\perp
\left(
\frac{1}{K}\sum_{k=1}^K P_W^{(k)}
\right)
P_V^\perp 
=
P_V+\frac{1}{K}\sum_{k=1}^K P_W^{(k)}
=
\frac{1}{K}\sum_{k=1}^K P^{(k)},
\end{aligned}
\]
where $\Sigma_0$ is defined in Step $1$ of the proof of Theorem
\ref{theorem:upper bound}. It follows immediately that
$
\Sigma_0V=V.
$
Thus, every vector in $\span(V)$ is an eigenvector of $\Sigma_0$
corresponding to the eigenvalue $1$.

For every unit vector $x$ orthogonal to $\span(V)$, we further have
\[
\begin{aligned}
x^\top\Sigma_0x
=
x^\top
\left(
\frac{1}{K}\sum_{k=1}^K P_W^{(k)}
\right)x 
\overset{\text{Assumption \ref{assum:3}}}{\leq}
1-\theta.
\end{aligned}
\]
Therefore,
$
\lambda_1(\Sigma_0)
=
\cdots
=
\lambda_{r_s}(\Sigma_0)
=
1,
$
and
$
\lambda_{r_s+1}(\Sigma_0)\leq1-\theta
$
whenever $r_s<n$,

Define $\omega_n$ by
$\omega_n
:=
\|\tilde{\Sigma}_0-\Sigma_0\|_{\op}.$
According to Step $2$ of the proof of Theorem
\ref{theorem:upper bound}, we have
\[
\begin{aligned}
\tilde{\Sigma}_0-\Sigma_0
=
\frac{1}{K}\sum_{k=1}^K
\left(\tilde P^{(k)}-P^{(k)}\right)
=
\frac{1}{K}\sum_{k=1}^K\Delta^{(k)},
\end{aligned}
\]
and hence
\[
\omega_n
=
\left\|
\frac{1}{K}\sum_{k=1}^K\Delta^{(k)}
\right\|_{\op}.
\]
Clearly we have
$0\preceq\tilde{\Sigma}_0\preceq I_n.$

We first consider the case $r_s\geq1$. For every
$\ell=1,\cdots,r_s$, Weyl's inequality gives us
$\hat\lambda_\ell
\geq
\lambda_\ell(\Sigma_0)-\omega_n 
=
1-\omega_n.$
Consequently, on the event $\{\omega_n<\zeta\}$, we have
$1-\hat\lambda_\ell<\zeta,
\ 
\ell=1,\cdots,r_s,$
so that all the first $r_s$ eigenvalues of $\tilde{\Sigma}_0$ are retained by the truncated
criterion in \eqref{eq:shared-rank-selection}. 

Another application of Weyl's inequality yields
\[
\begin{aligned}
\hat\lambda_{r_s+1}
\leq
\lambda_{r_s+1}(\Sigma_0)+\omega_n 
\leq
1-\theta+\omega_n.
\end{aligned}
\]
Therefore, on the event $\{\omega_n<\theta-\zeta\}$, we have
$
\zeta<1-\hat\lambda_{r_s+1}.
$
Thus, the $(r_s+1)$th eigenvalue is excluded by the truncated
criterion.
Combining the preceding two conclusions, whenever
$\omega_n
<
\min\{\zeta,\theta-\zeta\},$
we have
$\hat r_s=r_s.$

It remains to show that $\{\omega_n<\min\{\zeta,\theta-\zeta\}\}$ holds with probability tending to
one. By utilizing the result in Step $5$ of the proof of Corollary
\ref{theorem:upper bound}, we have
\[
\begin{aligned}
\|\omega_n\|_{\psi_1}
=
\left\|
\left\|
\tilde{\Sigma}_0-\Sigma_0
\right\|_{\op}
\right\|_{\psi_1}
\leq
\left\|
\left\|
\tilde{\Sigma}_0-\Sigma_0
\right\|_F
\right\|_{\psi_1}
\lesssim
\sqrt{\frac{n}{K\lambda_{\min}^2}}
+
\frac{n}{\lambda_{\min}^2}.
\end{aligned}
\]
According to the condition of Corollary \ref{cor:rank consistency}, we can deduce that $1/K\sigma_{\min}^2+1/\sigma_{\min}^4=o(1)$. Accordingly, we have 
$n/K\lambda^2_{\min}+n^2/\lambda^4_{\min}=o(1)$ by  $\lambda_{\min}=\sqrt{n}\sigma_{\min}$. As such, 
$
\|\omega_n\|_{\psi_1}=o(1)
$ holds. By the definition of the $\psi_1$ norm, we have $\omega_n=o_p(1).$
As $\zeta$ is fixed and satisfies $0<\zeta<\theta$, we have $\min\{\zeta,\theta-\zeta\}>0.$
Consequently,
\[
\P\left(
\omega_n<
\min\{\zeta,\theta-\zeta\}
\right)
\to1.
\]
It follows that, for $r_s\geq1$, we have
$\P(\hat r_s=r_s)\to1.$

We finally consider the case $r_s=0$. In this case,
$P_V=0_{n\times n}$ and
\[
\Sigma_0
=
\frac{1}{K}\sum_{k=1}^K P_W^{(k)}
=
\frac{1}{K}\sum_{k=1}^K P^{(k)}.
\]
By Assumption \ref{assum:3},
$
\lambda_1(\Sigma_0)\leq1-\theta.
$
Weyl's inequality  gives
$\hat\lambda_1
\leq
\lambda_1(\Sigma_0)+\omega_n
\leq
1-\theta+\omega_n.$
Hence,
$1-\hat\lambda_1
\geq
\theta-\omega_n.$
On the event $\{\omega_n<\theta-\zeta\}$, it follows that
$\hat\lambda_1<1-\zeta.$
Thus, none of the eigenvalues in the search range is retained by
\eqref{eq:shared-rank-selection}, and consequently
$\hat r_s=0.$
Since $\omega_n=o_p(1)$ and $\theta-\zeta>0$, we have
\[
\P(\omega_n<\theta-\zeta)\to1.
\]
Therefore,
$\P(\hat r_s=0)\to1.$
Combining the cases $r_s\geq1$ and $r_s=0$ proves that
\[
\P(\hat r_s=r_s)\to1,
\]
which completes the proof of Corollary 
\ref{cor:rank consistency}.

\section{Technical Lemmas}
\label{tec lemma}
\begin{lemma}[Signal strength]\label{lem:signal}
Under \eqref{Mk} and Assumption \ref{assum:2}, we have $\sigma_{r_k}(M^{(k)})\gtrsim \lambda_{\min}$ holding for each $k\in[K]$.
\end{lemma}

\begin{proof}
Throughout the whole proof of Lemma \ref{lem:signal}, we only discuss the $k$-th local machine.
Based on the definitions of $T^{(k)}=\cM_{L_k+1}(\cT^{(k)})$ and $\hat{U}_{\otimes}^{(k)}:=\hat{U}^{(k)}_{1}\otimes\cdots\otimes \hat{U}^{(k)}_{L_k}$, we have the following facts:
\begin{equation*}
\begin{aligned}
&\quad\sigma_{r_k}(M^{(k)})
\overset{\eqref{Mk}}{=}\sigma_{r_k}\left(T^{(k)}\left(\hat{U}^{(k)}_{1}\otimes\cdots\otimes \hat{U}^{(k)}_{L_k}\right)\right)
=
\sigma_{r_k}\left(\cM_{L_k+1}(\cT^{(k)})\hat{U}_{\otimes}^{(k)}\right)\\
&\overset{\eqref{TJIVE model}}{=}\sigma_{r_k}\left[\cM_{L_k+1}\left(\cC^{(k)} \times_1U^{(k)}_{1}\cdots \times_{L_k} U^{(k)}_{L_k}\times_{L_k+1}[V,W^{(k)}]\right)\hat{U}_{\otimes}^{(k)}\right].
\end{aligned}
\end{equation*}
Then, $U_{\otimes}^{(k)}:= U^{(k)}_{1}\otimes\cdots\otimes U^{(k)}_{L_k}$. The column orthonormal property for each $U_j^{(k)}$ give us:
\begin{equation*}
\begin{aligned}
&\quad
\left[
\cM_{L_k+1}\left(
\cC^{(k)}
\times_1 U_1^{(k)}
\cdots
\times_{L_k} U_{L_k}^{(k)}
\times_{L_k+1}[V,W^{(k)}]
\right)
\right]\hat U_{\otimes}^{(k)}
\\
&=
[V,W^{(k)}]\,
\cM_{L_k+1}(\cC^{(k)})
(U_{\otimes}^{(k)})^\top
\hat U_{\otimes}^{(k)}
\\
&=
[V,W^{(k)}]\,
\cM_{L_k+1}(\cC^{(k)})
(U_{\otimes}^{(k)})^\top
U_{\otimes}^{(k)}
(U_{\otimes}^{(k)})^\top
\hat U_{\otimes}^{(k)}
\\
&=
\cM_{L_k+1}(\cT^{(k)})
U_{\otimes}^{(k)}
(U_{\otimes}^{(k)})^\top
\hat U_{\otimes}^{(k)}
\\
&=
T^{(k)}
U_{\otimes}^{(k)}
(U_{\otimes}^{(k)})^\top
\hat U_{\otimes}^{(k)}.
\end{aligned}
\end{equation*}
Accordingly, we have:
$$
\sigma_{r_k}(M^{(k)})=\sigma_{r_k}(T^{(k)}U_{\otimes}^{(k)}(U_{\otimes}^{(k)})^\top\hat{U}_{\otimes}^{(k)}).
$$
Based on the above equality, we shall have
\begin{equation*}
\begin{aligned}
\sigma_{r_k}&(M^{(k)})=\sigma_{r_k}\left(T^{(k)}U_{\otimes}^{(k)}(U_{\otimes}^{(k)})^{\top}\hat{U}_{\otimes}^{(k)}\right)
\geq \sigma_{r_k}\left(T^{(k)}U_{\otimes}^{(k)}\right)\cdot \sigma_{\min}\left[(U_{\otimes}^{(k)})^{\top}\hat{U}_{\otimes}^{(k)}\right]\\
&\overset{\text{Proposition \ref{prop:weyl}}}{\geq}
\sigma_{r_k}\left(T^{(k)}\right)\cdot
\sigma_{\min}\left(U_{\otimes}^{(k)}\right)\cdot \sigma_{\min}\left[(U_{\otimes}^{(k)})^{\top}\hat{U}_{\otimes}^{(k)}\right]\\
&\underset{U_{\otimes}^{(k)}\ \text{and}\ \hat{U}_{\otimes}^{(k)}}
{\overset{\text{definitions of}}{=}}
\sigma_{r_k}\left(T^{(k)}\right)\cdot \sigma_{\min}\left\{\left[(U^{(k)}_{1})^{\top}\otimes\cdots\otimes (U^{(k)}_{L_k})^{\top}\right]\left[\hat{U}^{(k)}_{1}\otimes\cdots\otimes \hat{U}^{(k)}_{L_k}\right]\right\}\\
&\underset{\text{ \cite{zhang2018tensor}}}
{\overset{\text{Lemma 4 of}}{=}}
\sigma_{r_k}\left(T^{(k)}\right)\cdot \sigma_{\min}\left\{\left[(U^{(k)}_{1})^{\top}\hat{U}^{(k)}_{1}\otimes\cdots\otimes (U^{(k)}_{L_k})^{\top}\hat{U}^{(k)}_{L_k}\right]\right\}
\\
&\underset{\text{ \cite{zhang2018tensor}}}
{\overset{\text{Lemma 4 of}}{=}}
\sigma_{r_k}\left(T^{(k)}\right)\cdot   \sigma_{\min}\left((U^{(k)}_{1})^{\top}\hat{U}^{(k)}_{1}\right)\cdots \sigma_{\min}\left((U^{(k)}_{L_k})^{\top}\hat{U}^{(k)}_{L_k}\right)\\
&\underset{\text{ \cite{cai2018rate}}}
{\overset{\text{Lemma 1 of}}{=}}\sigma_{r_k}\left(T^{(k)}\right)\cdot   \sqrt{1-\left\|\sin\Theta(U^{(k)}_{1},\hat{U}^{(k)}_{1})\right\|_{\op}^2}\cdots \sqrt{1-\left\|\sin\Theta(U^{(k)}_{L_k},\hat{U}^{(k)}_{L_k})\right\|_{\op}^2}
\\
&\overset{\text{Assumption \ref{assum:2}}}{\gtrsim} \lambda_{\min}\left(\frac{\sqrt{3}}{2}\right)^{L_k}
\gtrsim\lambda_{\min},
\end{aligned}
\end{equation*}
where the last inequality holds due to the finiteness of $\{L_k\}_k$. Till now, we have completed the proof of Lemma \ref{lem:signal}.
\end{proof}

\begin{lemma}[Noise level]\label{lem:noise}
For each $k\in[K]$, recall that $r^{(k)}_{\Pi}:=\prod_{i=1}^{L_k}r_i^{(k)}$. Under Assumption \ref{assum:1} and \ref{assum:4}, there exist constants $c,C>0$ such that the following concentration holds for any $t>0$:
$$
\P\left\{\left\|E^{(k)}\right\|_{\op}\geq C\left(\sqrt{n}+\sqrt{r^{(k)}_{\Pi}}+\sqrt{\sum_{j=1}^{L_k}p_j^{(k)}r_j^{(k)}}+t\right)\right\}\leq2e^{-ct^2},
$$
which further implies:
$$\left\|\left\|E^{(k)}\right\|_{\op}\right\|_{\psi_1} \leq \left\|\left\|E^{(k)}\right\|_{\op}\right\|_{\psi_2}=O\left(\sqrt{n}+\sqrt{\sum_{j=1}^{L_k}p_j^{(k)}}\right)=O(\sqrt{n}).$$
\end{lemma}
\begin{proof}
To prove the first claim of Lemma \ref{lem:noise}, we can imitate the proving procedure of Lemma $5$ in \cite{zhang2018tensor}. 

$\\$
\textbf{\textit{Step 1:}}
Recall from \eqref{Ek}, for each $k\in[K]$, we have:
\begin{equation*}
E^{(k)}=Z^{(k)}\left(\hat{U}^{(k)}_{1}\otimes\cdots\otimes \hat{U}^{(k)}_{L_k}\right).
\end{equation*}
Further, for each $j\in[L_k]$, we denote 
$\cQ_j^{(k)}:=\{Q_j^{(k)} \in \mathbb{R}^{p_j^{(k)} \times r_j^{(k)}}| \|Q_j^{(k)}\|_{\op} \leq 1\}$. Define $A^{(k)}$ as:
$$A^{(k)} = \max_{\substack{ \{Q_j^{(k)}\in\cQ_j^{(k)}\}_{j\in[L_k]} }} \left\| Z^{(k)}\left(Q_1^{(k)}\otimes\cdots\otimes Q_{L_k}^{(k)}\right) \right\|_{\op}.$$
Clearly, $\{\|\hat{U}^{(k)}_{j}\|_{\op}\leq1\}_{j\in[L_k]}$ holds, and we have $\|E^{(k)}\|_{\op}\leq A^{(k)}$. Thus, in order to prove Lemma \ref{lem:noise}, all we need to do is consider the concentration for $A^{(k)}$. By applying Lemma $7$ of \cite{zhang2018tensor}, we can find an $\epsilon$-net for every $\cQ_j^{(k)}$ such that 
$$
\tilde{\cQ}_j^{(k)}:=\{Q_{j,1}^{(k)},Q_{j,2}^{(k)},\cdots,Q_{j,m_j^{(k)}}^{(k)}\}\subseteq\cQ_j^{(k)},\ \text{and}\ m_j^{(k)}\leq \left(\frac{4+\epsilon}{\epsilon}\right)^{p_j^{(k)}r_j^{(k)}}.
$$
Then, for every $j\in[L_k]$, we extract one $Q_{j,q_j}^{(k)}$ from $\tilde{\cQ}_j^{(k)}$ and define:
$$
E^{(k)}_{(q_1,q_2,\cdots,q_{L_k})}:= Z^{(k)}\left(Q_{1,q_1}^{(k)}\otimes Q_{2,q_2}^{(k)}\otimes\cdots\otimes Q_{L_k,q_{L_k}}^{(k)}\right).
$$
If we denote the $i$-th row of $Z^{(k)}$ as $(z^{(k)}_i)^\top$, then the $i$-th row of $E^{(k)}_{(q_1,q_2,\cdots,q_{L_k})}$ simply equals $(z^{(k)}_i)^\top(Q_{1,q_1}^{(k)}\otimes Q_{2,q_2}^{(k)}\otimes\cdots\otimes Q_{L_k,q_{L_k}}^{(k)})$. As each element of $\cE^{(k)}$ is mutually independent to the rest of the elements and $Z^{(k)}=\cM_{L_k+1}(\cE^{(k)})$, we know that  $E^{(k)}_{(q_1,q_2,\cdots,q_{L_k})}$ also has independent rows. Besides, for any $u\in\mathbb{R}^{r_{\Pi}^{(k)}}$ where $\|u\|_2=1$, we have:
$$
\left\|\langle (z^{(k)}_i)^\top\left(Q_{1,q_1}^{(k)}\otimes Q_{2,q_2}^{(k)}\otimes\cdots\otimes Q_{L_k,q_{L_k}}^{(k)}\right),u\rangle\right\|_{\psi_2}
=
\left\|\langle z^{(k)}_i,\left(Q_{1,q_1}^{(k)}\otimes Q_{2,q_2}^{(k)}\otimes\cdots\otimes Q_{L_k,q_{L_k}}^{(k)}\right)u\rangle\right\|_{\psi_2}.
$$
According to Assumption \ref{assum:1}, we have, by utilizing Lemma 5.24 of \cite{vershynin2010introduction}, that:
$$
\left\|\langle (z^{(k)}_i)^\top\left(Q_{1,q_1}^{(k)}\otimes Q_{2,q_2}^{(k)}\otimes\cdots\otimes Q_{L_k,q_{L_k}}^{(k)}\right),u\rangle\right\|_{\psi_2}
\leq
e^2\left\|\left(Q_{1,q_1}^{(k)}\otimes Q_{2,q_2}^{(k)}\otimes\cdots\otimes Q_{L_k,q_{L_k}}^{(k)}\right)u\right\|.
$$
Accordingly, each row of $E^{(k)}_{(q_1,q_2,\cdots,q_{L_k})}$ is an independent, centered, $e^2$-sub Gaussian vector. Thus, by utilizing Theorem 5.39 of \cite{vershynin2010introduction}, there exists a constant $c_1$ such that for any $t>0$, we have:
$$
\P\left\{\left\|E^{(k)}_{(q_1,q_2,\cdots,q_{L_k})}\right\|_{\op}\geq\sqrt{n}+\sqrt{r_{\Pi}^{(k)}}+t
\right\}\leq2\exp\left(-c_1t^2\right).
$$
Note that the above concentration holds for any $q_j\in[m_j^{(k)}]$ and $j\in[L_k]$. Based on the above result, we obtain the following union bound:
\begin{equation*}
\begin{aligned}
&\quad\P\left\{\max_{\substack{
1\leq j\leq L_k\\
1\leq q_j\leq m_j^{(k)}
}}\left\|E^{(k)}_{(q_1,q_2,\cdots,q_{L_k})}\right\|_{\op}\geq\sqrt{n}+\sqrt{r_{\Pi}^{(k)}}+t
\right\}\\
&\leq
\sum_{q_1=1}^{m_1^{(k)}}\cdots\sum_{q_{L_k}=1}^{m_{L_k}^{(k)}}
\P\left\{\left\|E^{(k)}_{(q_1,q_2,\cdots,q_{L_k})}\right\|_{\op}\geq\sqrt{n}+\sqrt{r_{\Pi}^{(k)}}+t
\right\}\\
&\leq
\sum_{q_1=1}^{m_1^{(k)}}\cdots\sum_{q_{L_k}=1}^{m_{L_k}^{(k)}}2\exp\left(-c_1t^2\right)
=
\left[2\exp\left(-c_1t^2\right)\right]\prod\limits_{j=1}^{L_k}m_j^{(k)}\\
&\leq\left[2\exp\left(-c_1t^2\right)\right]\left[\prod\limits_{j=1}^{L_k}
\left(\frac{4+\epsilon}{\epsilon}\right)^{p_j^{(k)}r_j^{(k)}}\right]\\
&=2\exp\left\{\left(\sum_{j=1}^{L_k}p_j^{(k)}r_j^{(k)}\right)\left[\log\left(\frac{4+\epsilon}{\epsilon}\right)\right]-c_1t^2\right\}.
\end{aligned}
\end{equation*}

$\\$
\textbf{\textit{Step 2:}}
Assume that 
$$
\left\{Q_j^{(k),*}
\right\}_{j=1}^{L_k}= \argmax_{\substack{ \{Q_j^{(k)}\in\cQ_j^{(k)}\}_{j\in[L_k]} }} \left\| Z^{(k)}\left(Q_1^{(k)}\otimes\cdots\otimes Q_{L_k}^{(k)}\right) \right\|_{\op}.
$$
As $Q_j^{(k),*}\in\cQ_j^{(k)}$, then for every $j\in[L_k]$, we can find a corresponding $Q_{j,q_j}^{(k)}$ such that $\|Q_j^{(k),*}-Q_{j,q_j}^{(k)}\|_{\op}\leq\epsilon$ holds based on the definition of $\epsilon$-net. Accordingly, we can split $A^{(k)}$ as follows:
\begin{equation*}
\begin{aligned}
&\quad A^{(k)}=\left\| Z^{(k)}\left(Q_1^{(k),*}\otimes Q_2^{(k),*}\otimes\cdots\otimes Q_{L_k}^{(k),*}\right) \right\|_{\op}\\
&=\left\| Z^{(k)}\left(Q_{1,q_1}^{(k)}\otimes Q_2^{(k),*}\otimes\cdots\otimes Q_{L_k}^{(k),*}\right)
+Z^{(k)}\left[\left(Q_1^{(k),*}-Q_{1,q_1}^{(k)}\right)\otimes Q_2^{(k),*}\otimes\cdots\otimes Q_{L_k}^{(k),*}\right] \right\|_{\op}.
\end{aligned}
\end{equation*}
Keep splitting the first term gives us:
\begin{equation*}
\begin{aligned}
A^{(k)}&=\left\| Z^{(k)}\left(Q_1^{(k),*}\otimes Q_2^{(k),*}\otimes\cdots\otimes Q_{L_k}^{(k),*}\right) \right\|_{\op}\\
&=
\left\| Z^{(k)}\left(Q_{1,q_1}^{(k)}\otimes Q_{2,q_2}^{(k)} \otimes \cdots \otimes Q_{l,q_l}^{(k)}\otimes Q_{l+1}^{(k),*}\otimes\cdots\otimes Q_{L_k}^{(k),*}\right)+\right.\\
&\quad\left.\sum_{j=1}^l Z^{(k)}\left[Q_{1,q_1}^{(k)}\otimes Q_{2,q_2}^{(k)}\otimes\cdots\otimes\left(Q_j^{(k),*}-Q_{j,q_j}^{(k)}\right)\otimes Q_{j+1}^{(k),*}\otimes\cdots\otimes Q_{L_k}^{(k),*}\right] \right\|_{\op}\\
&=
\left\| Z^{(k)}\left(Q_{1,q_1}^{(k)}\otimes Q_{2,q_2}^{(k)} \otimes \cdots \otimes Q_{L_k,q_{L_k}}^{(k)}\right)+\right.\\
&\quad\left.\sum_{j=1}^{L_k} Z^{(k)}\left[Q_{1,q_1}^{(k)}\otimes Q_{2,q_2}^{(k)}\otimes\cdots\otimes\left(Q_j^{(k),*}-Q_{j,q_j}^{(k)}\right)\otimes Q_{j+1}^{(k),*}\otimes\cdots\otimes Q_{L_k}^{(k),*}\right] \right\|_{\op}\\
&\leq\left\| Z^{(k)}\left(Q_{1,q_1}^{(k)}\otimes Q_{2,q_2}^{(k)} \otimes \cdots \otimes Q_{L_k,q_{L_k}}^{(k)}\right)\right\|_{\op}+\\
&\quad\sum_{j=1}^{L_k} \left\|Z^{(k)}\left[Q_{1,q_1}^{(k)}\otimes Q_{2,q_2}^{(k)}\otimes\cdots\otimes\left(Q_j^{(k),*}-Q_{j,q_j}^{(k)}\right)\otimes Q_{j+1}^{(k),*}\otimes\cdots\otimes Q_{L_k}^{(k),*}\right] \right\|_{\op}\\
&=\left\|E^{(k)}_{(q_1,q_2,\cdots,q_{L_k})}\right\|_{\op}
+
\sum_{j=1}^{L_k} \left\|Z^{(k)}\left[Q_{1,q_1}^{(k)}\otimes\cdots\otimes\left(Q_j^{(k),*}-Q_{j,q_j}^{(k)}\right)\otimes\cdots\otimes Q_{L_k}^{(k),*}\right] \right\|_{\op}.
\end{aligned}
\end{equation*}
Note that for every $j\in[L_k]$, we have $\|Q_j^{(k),*}-Q_{j,q_j}^{(k)}\|_{\op}\leq\epsilon$. Accordingly, we can rewrite $Q_j^{(k),*}-Q_{j,q_j}^{(k)}=\epsilon \tilde{Q}_j^{(k)}$ where $\|\tilde{Q}_j^{(k)}\|_{\op}\leq1$. Then we have, for any $t>0$, that:
\begin{equation*}
\begin{aligned}
A^{(k)}
&=
\left\|E^{(k)}_{(q_1,q_2,\cdots,q_{L_k})}\right\|_{\op}
+
\sum_{j=1}^{L_k} \left\|Z^{(k)}\left[Q_{1,q_1}^{(k)}\otimes\cdots\otimes\epsilon\tilde{Q}_j^{(k)}\otimes\cdots\otimes Q_{L_k}^{(k),*}\right] \right\|_{\op}\\
&\leq\max_{\substack{
1\leq j\leq L_k\\
1\leq q_j\leq m_j^{(k)}
}}\left\|E^{(k)}_{(q_1,q_2,\cdots,q_{L_k})}\right\|_{\op}
+
\sum_{j=1}^{L_k} \epsilon\left\|Z^{(k)}\left[Q_{1,q_1}^{(k)}\otimes\cdots\otimes\tilde{Q}_j^{(k)}\otimes\cdots\otimes Q_{L_k}^{(k),*}\right] \right\|_{\op}\\
&\leq\max_{\substack{
1\leq j\leq L_k\\
1\leq q_j\leq m_j^{(k)}
}}\left\|E^{(k)}_{(q_1,q_2,\cdots,q_{L_k})}\right\|_{\op}+L_k\epsilon A^{(k)}\\
&\leq\sqrt{n}+\sqrt{r_{\Pi}^{(k)}}+t+L_k\epsilon A^{(k)},
\end{aligned}
\end{equation*}
holds with probability larger than $1-\exp\{(\sum_{j=1}^{L_k}p_j^{(k)}r_j^{(k)})[\log((4+\epsilon)/\epsilon)]-c_1t^2\}$. This gives us:
$$
A^{(k)}\leq\frac{1}{1-L_k\epsilon}\left(\sqrt{n}+\sqrt{r_{\Pi}^{(k)}}+t\right),
$$
with the same probability as the last inequalities. 
By setting $\epsilon=1/2L_k$, then $L_k\epsilon=1/2=O(1)$, and $\log((4+\epsilon)/\epsilon)=\log(8L_k+1)=O(1)$. Thus, there exist two constants $C_1$ and $C_2$ such that:
$$
A^{(k)}\leq C_1\left(\sqrt{n}+\sqrt{r_{\Pi}^{(k)}}+t\right),
$$
with probability larger than $1-\exp\{C_2(\sum_{j=1}^{L_k}p_j^{(k)}r_j^{(k)})-c_1t^2\}$. If we let $t=\tilde{C}\sqrt{\sum_{j=1}^{L_k}p_j^{(k)}r_j^{(k)}}+u$ for any $u>0$ and sufficiently large $\tilde{C}>0$, then there exist absolute constants $\tilde{c}$ and $c$ such that:
\begin{equation*}
\begin{aligned}
&\quad\P\left\{\|E^{(k)}\|_{\op}\leq C_1\left(\sqrt{n}+\sqrt{r_{\Pi}^{(k)}}+t\right)\right\}\\
&=
\P\left\{\|E^{(k)}\|_{\op}\leq C_1\left(\sqrt{n}+\sqrt{r_{\Pi}^{(k)}}+\tilde{C}\sqrt{\sum_{j=1}^{L_k}p_j^{(k)}r_j^{(k)}}+u\right)\right\}\\
&\geq
\P\left\{A^{(k)}\leq C_1\left(\sqrt{n}+\sqrt{r_{\Pi}^{(k)}}+\tilde{C}\sqrt{\sum_{j=1}^{L_k}p_j^{(k)}r_j^{(k)}}+u\right)\right\}\\
&\geq 1-\exp\left\{C_2\left(\sum_{j=1}^{L_k}p_j^{(k)}r_j^{(k)}\right)-c_1t^2\right\}\\
&=1-\exp\left\{C_2\left(\sum_{j=1}^{L_k}p_j^{(k)}r_j^{(k)}\right)-c_1\left(\tilde{C}\sqrt{\sum_{j=1}^{L_k}p_j^{(k)}r_j^{(k)}}+u\right)^2\right\}\\
&=1-\exp\left\{C_2\left(\sum_{j=1}^{L_k}p_j^{(k)}r_j^{(k)}\right)-c_1\left(\tilde{C}^2\sum_{j=1}^{L_k}p_j^{(k)}r_j^{(k)}+2\tilde{C}\sqrt{\sum_{j=1}^{L_k}p_j^{(k)}r_j^{(k)}}u+u^2\right)\right\}\\
&\geq1-\exp\left\{-\tilde{c}\sqrt{\sum_{j=1}^{L_k}p_j^{(k)}r_j^{(k)}}u-\tilde{c}u^2\right\}
\geq
1-\exp\left\{-cu^2\right\},
\end{aligned}
\end{equation*}
which finishes the proof of the first claim. 

$\\$
\textbf{\textit{Step 3:}}
We next derive the $\psi_2$-norm bound based on the results of Step 2. Define
\[
B^{(k)}
:=
\sqrt{n}
+\sqrt{r_\Pi^{(k)}}
+\sqrt{\sum_{j=1}^{L_k}p_j^{(k)}r_j^{(k)}}.
\]
The first part of the lemma shows that there exist constants
$c_0,C_0>0$ such that, for every $t\geq0$,
\begin{equation}\label{eq:noise-shifted-tail}
\P\left\{
\left\|E^{(k)}\right\|_{\op}\geq C_0(B^{(k)}+t)
\right\}
\leq
2\exp(-c_0t^2).
\end{equation}
Define the nonnegative random variable
$Y^{(k)}:=\left(\|E^{(k)}\|_{\op}-C_0B^{(k)}\right)_+$, where
$x_+:=\max\{x,0\}$. Since $Y^{(k)}\geq0$, we have
$|Y^{(k)}|=Y^{(k)}$. For every $s>0$, taking $t=s/C_0$ in
\eqref{eq:noise-shifted-tail} gives
\begin{equation*}
\begin{aligned}
\P\left\{|Y^{(k)}|>s\right\}
&=
\P\left\{Y^{(k)}>s\right\}
=
\P\left\{
\left\|E^{(k)}\right\|_{\op}>
C_0\left(B^{(k)}+\frac{s}{C_0}\right)
\right\}
\leq
2\exp\left(-\frac{c_0}{C_0^2}s^2\right).
\end{aligned}
\end{equation*}
The second equality holds as events $\{(x)_+>s\}$ and $\{x>s\}$ are equivalent for every $s>0$.
For $s=0$, the same bound holds trivially since
$\P\{|Y^{(k)}|>0\}\leq1\leq2$. Therefore,
\[
\P\left\{|Y^{(k)}|>s\right\}
\leq
2
=
2\exp\left(-\frac{c_0}{C_0^2}s^2\right).
\]
By the equivalence between sub-Gaussian tail decay and the
$\psi_2$-norm in Lemma 5.5 of
\cite{vershynin2010introduction}, it follows that
$\|Y^{(k)}\|_{\psi_2}\leq C C_0/\sqrt{c_0}=O(1)$ for an absolute
constant $C>0$.

Since
$
0\leq \left\|E^{(k)}\right\|_{\op}\leq C_0B^{(k)}+Y^{(k)},
$
the monotonicity and triangle inequality of the $\psi_2$-norm yield
\begin{equation*}
\begin{aligned}
\left\|
\left\|E^{(k)}\right\|_{\op}
\right\|_{\psi_2}
&\leq
\left\|C_0B^{(k)}\right\|_{\psi_2}
+\|Y^{(k)}\|_{\psi_2}
\lesssim B^{(k)}+1
\lesssim
\sqrt{n}
+\sqrt{r_\Pi^{(k)}}
+\sqrt{\sum_{j=1}^{L_k}p_j^{(k)}r_j^{(k)}}.
\end{aligned}
\end{equation*}
Moreover, the standard relation between Orlicz norms gives
\[
\left\|
\left\|E^{(k)}\right\|_{\op}
\right\|_{\psi_1}
\lesssim
\left\|
\left\|E^{(k)}\right\|_{\op}
\right\|_{\psi_2}.
\]

Finally, based on Assumption \ref{assum:4}, we have $r_\Pi^{(k)}=O(1)$
and
\[
\sum_{j=1}^{L_k}p_j^{(k)}r_j^{(k)}
\lesssim
\sum_{j=1}^{L_k}p_j^{(k)}=O(n).
\]
Consequently,
\[
\left\|
\left\|E^{(k)}\right\|_{\op}
\right\|_{\psi_1}
\lesssim
\left\|
\left\|E^{(k)}\right\|_{\op}
\right\|_{\psi_2}
=
O(\sqrt n),
\]
which completes the proof.
\end{proof}

\section{Useful results}

For the sake of convenience, this section collects some useful results used in the proofs of the theoretical arguments.

\begin{proposition}[Lemma 3.1 from \cite{vu2013fantope}]\label{prop:curvature}
    Let $X$ be a symmetric matrix and $P$ be the projection onto the subspace spanned by its leading $r$ eigenvectors. If $\delta_X = (\lambda_{r}-\lambda_{r+1})(X)>0$, then for all $0\preceq Q \preceq I$ and $\tr(Q) =r$,
    $$\frac{\delta_X}{2}\left\|P-Q\right\|^2_F\leq\left\langle X,P-Q\right\rangle.$$
\end{proposition}

\begin{proposition}
[Proposition D.3 in \cite{vu2013minimax}]
\label{prop:psi1_leq_psi2}
Let $X,Y$ be any random variables, we have:
$$
\left\|XY\right\|_{\psi_1}\leq\left\|X\right\|_{\psi_2}\left\|Y\right\|_{\psi_2}.
$$
Moreover, take $Y=1$ gives us 
$$
\left\|X\right\|_{\psi_1}\leq\left\|X\right\|_{\psi_2},
$$
for any random variable $X$.
\end{proposition}

\begin{proposition}[Lemma 2 from \cite{fan2019distributed}]\label{prop:taylor}
Let $\Sigma$ and $\hat{\Sigma}$ be $p\times p$ symmetric matrices with non-increasing eigenvalues $\lambda_1\geq \cdots\geq \lambda_p$ and $\hat{\lambda}_1\geq \cdots\geq \hat{\lambda}_p$, respectively. Let $\{u_i\}_{i=1}^{p}$, $\{\hat{u}_i\}_{i=1}^{p}$ be the corresponding eigenvectors such that $\Sigma u_i = \lambda_iu_i$ and $\hat{\Sigma} u_i = \hat{\lambda}_i\hat{u}_i$. Fix $s\in \{0,\cdots,p-r\}$, let $\delta=\min(\lambda_s-\lambda_{s+1},\lambda_{s+r}-\lambda_{s+r+1})>0$, with $\lambda_0=+\infty$ and $\lambda_{p+1}=-\infty$. For $E = \hat{\Sigma}-\Sigma$, we denote $S=\{s+1,\cdots,s+r\}$ and $S^c=[p]\setminus S$, and define $K$ of shape $p\times p$ where
$$K_{ij}=K_{ji}= 
\left\{ 
    \begin{array}{lc}
        (u_i^{\top} E u_j)/(\lambda_i-\lambda_j) & i\in S,\, j\in S^c, \\
       0& \text{otherwise}.\\
    \end{array}
\right.$$
Then, let $U = (u_1,\cdots, u_p)$, $\hat{U} = (\hat{u}_1,\cdots, \hat{u}_p)$, while $U_S = (u_{s+1},\cdots, u_{s+r})$, and $\hat{U}_S = (\hat{u}_{s+1},\cdots, \hat{u}_{s+r})$, we claim that
$$\hat{U}_S\hat{U}_S^{\top}-U_SU_S^{\top} = UKU^{\top}+\Delta,$$
where $\|K\|_F\leq (2r)^{1/2}\|E\|_{\op}/\delta$. In addition, when $\|E\|_{\op}/\delta\leq 1/10$ holds, we have $\|\Delta\|_F\leq 24r^{1/2}(\|E\|_{\op}/\delta)^2$.
\end{proposition}

\begin{proposition}[Lemma 1 from \cite{2dgb}]\label{prop:wedin}
    Let $\Theta$ and $\hat{\Theta}$ be $p\times q$ real rectangular matrices with non-trivial singular values $\sigma_1\geq \cdots \geq \sigma_{d}$ and $\hat{\sigma}_1\geq \cdots \geq \hat{\sigma}_{d}$ where $d=\min(p,q)$. Let $\{u_i,v_i\}_{i=1}^{d}$, $\{\hat{u}_i,\hat{v}_i\}_{i=1}^{d}$ be the corresponding singular vector pairs such that $\Theta v_i = \sigma_i u_i$, $\Theta^{\top}u_i = \sigma_i v_i$ and $\hat{\Theta}\hat{v}_i = \hat{\sigma}_i \hat{u}_i$, $\hat{\Theta}^{\top}\hat{u}_i = \hat{\sigma}_i \hat{v}_i$. Fix $s=\{0,\cdots,d-r\}$, let $\delta=\min(\sigma_s-\sigma_{s+1},\sigma_{s+r}-\sigma_{s+r+1})>0$ with $\sigma_0=+\infty$ and $\sigma_{d+1}=-\infty$ by convention. For $E = \hat{\Theta}-\Theta$, we denote $S=\{s+1,\cdots,s+r\}$ and $S^c=[d]\setminus S$, and define $\hat{\cX}  = \cX+\cE$ where
$$\hat{\cX} = \begin{pmatrix} 0& \hat{\Theta}^{\top}  \\ \hat{\Theta} & 0 \end{pmatrix},\quad \cX = \begin{pmatrix} 0& \Theta^{\top}  \\ \Theta & 0 \end{pmatrix}, \quad \cE = \begin{pmatrix} 0& E^{\top}  \\ E & 0 \end{pmatrix}.$$
In fact, the $(p+q)\times(p+q)$ matrix $\cX$ has non-trivial eigenvalues $\{\sigma_i,\sigma_{-i}\}_{i=1}^{d}$, where $\sigma_{-i}:=-\sigma_{i}$, whose corresponding eigenvectors are respectively
$$\theta_i = \frac{1}{\sqrt{2}}\begin{pmatrix} v_i\\u_i
	\end{pmatrix},\quad \theta_{-i} = \frac{1}{\sqrt{2}}\begin{pmatrix} v_i\\-u_i
	\end{pmatrix}.$$
Then, let $U_S = (u_i)_{i\in S}$, $V_S = (v_i)_{i\in S}$, $\hat{U}_S = (\hat{u}_i)_{i\in S}$, $\hat{V}_S = (\hat{v}_i)_{i\in S}$, $U_{S^c} = (u_j)_{j\in S^c}$, $V_{S^c} = (v_j)_{j\in S^c}$, we claim that
$$\begin{pmatrix}
	\hat{V}_S\hat{V}_S^{\top}-V_SV_S^{\top}&0\\
	0&\hat{U}_S\hat{U}_S^{\top}-U_SU_S^{\top}
\end{pmatrix}=L+ R,\quad \text{for}$$
$$L := \sum_{|i|\in S}\sum_{|j|\in S^c}\frac{\theta_i^{\top}\cE\theta_j}{\sigma_{i}-\sigma_{j}}\left(\theta_i\theta_j^{\top}+\theta_j\theta_i^{\top}\right),\quad\text{such that}$$
$$\|L\|_F \leq \frac{2\sqrt{2}}{\delta}\left(\|U_S^{\top}EV_{S^c}\|_F+\|V_S^{\top}E^{\top}U_{S^c}\|_F\right)\leq \frac{4\sqrt{2}r^{1/2}\|E\|_{\op}}{\delta}.$$
In addition, when $\|E\|_{\op}/\delta\leq 1/10$, we have $\|R\|_F\leq 48r^{1/2}(\|E\|_{\op}/\delta)^2$.
\end{proposition}

\begin{proposition}[Theorem 2.5 from \cite{bosq2000stochastic}]\label{prop:subexp}
	For independent random vectors $\{X_i\}_{i=1}^{n}$ in a separable Hilbert space with norm $\|\cdot\|$, if $\E(X_i) = 0$ and $\|\|X_i\|\|_{\psi_1}\leq L_i<\infty$, we have
    $$\left\|\left\|\sum_{i=1}^{n}X_i\right\|\right\|_{\psi_1}\lesssim \left(\sum_{i=1}^{n}L_i^2\right)^{1/2}.$$
\end{proposition}

\begin{proposition}
[Theorem 4.3.1  from \cite{horn2012matrix}]
\label{prop:weyl}
Let \( A, B \in\mathbb{R}^{n\times n}\) be two symmetric matrices and let the eigenvalues \( \lambda_i(A) \), \( \lambda_i(B) \), and \( \lambda_i(A + B) \) be arranged in decreasing order. For each \( k = 1, 2, \cdots, n \), we have
\[
\lambda_k(A)+\lambda_n(B)\leq\lambda_k(A + B)\leq\lambda_k(A)+\lambda_1(B)
\]
\end{proposition}

\begin{proposition}
[Lemma 2 in \cite{zhang2018tensor}]
\label{prop:packing}
Given such a subspace \( L \subset \mathbb{R}^p \) with \( \dim(L) = r \),
let \( U_L \in \cO(p, r) \) denote the orthonormal basis of \( L \) while $\cO(p, r):= \{U_L| L = \cG(p,r)\}$ equipped with Schatten \( q \)-norm distances for all \( q \in [1, +\infty] \):
$
d_q(U_{L_1}, U_{L_2}) := \|U_{L_1} U_{L_1}^{\top} - U_{L_2} U_{L_2}^{\top}\|_q.
$
For a given metric space $(T,d)$, the \( \epsilon \)-packing number of the $(T,d)$ is defined as:
\[
D(T, d, \epsilon) := \max \left\{n: \text{there are } t_1, \cdots, t_n \in T, \text{ such that } \min_{i \neq j} d(t_i, t_j) > \epsilon \right\}.
\]
Then, for any integers \( 1 \leq r \leq p \) such that \( r \leq p - r \), and all \( 1 \leq q \leq +\infty \), the following
bound holds
\[
\left( \frac{c}{\epsilon} \right)^{r(p - r)} \leq D(\cO(p,r), d_q, \epsilon r^{1/q}) \leq \left( \frac{C}{\epsilon} \right)^{r(p - r)}
\]
with absolute constants \( c, C > 0 \). 
\end{proposition}

\vskip 0.2in
\bibliography{main}

\end{document}